\documentclass{new_tlp}

\usepackage{amsmath}
\usepackage{amsfonts}
\usepackage{amssymb}
\usepackage{galois}
\usepackage{url}
\usepackage{tikz}
\usepackage{xspace}
\usepackage{multirow}
\usepackage{footnote}
\usepackage{stmaryrd}
\makesavenoteenv{tabular}
\makesavenoteenv{table}
\usepackage{listings}
\usepackage{pgfplots}
\usepackage{rotating}

\usepackage[colorlinks=true,citecolor=blue]{hyperref}
\usepackage{xargs}

\DeclareRobustCommand{\stirling}{\genfrac\{\}{0pt}{}}

\let\othelstnumber=\thelstnumber
\def\createlinenumber#1#2{
    \edef\thelstnumber{%
        \unexpanded{%
            \ifnum#1=\value{lstnumber}\relax
              #2%
            \else}%
        \expandafter\unexpanded\expandafter{\thelstnumber\othelstnumber\fi}%
    }
    \ifx\othelstnumber=\relax\else
      \let\othelstnumber\relax
    \fi
}

\lstdefinelanguage{xc}{
  keywords={int,if,return}
}[keywords]

\definecolor{check}{RGB}{0,0,150}
\definecolor{true}{RGB}{0,150,0}
\definecolor{false}{RGB}{150,0,0}
\definecolor{checked}{RGB}{0,100,0}
\definecolor{commentgreen}{RGB}{2,112,10}
\definecolor{eminence}{RGB}{108,48,130}
\definecolor{weborange}{RGB}{200,130,0}
\definecolor{frenchplum}{RGB}{129,20,83}
\usepackage{textcomp}
\newcommand{\prettylstc}[0]{
\lstset {
    language=c,
    frame=ltrb,
    rulecolor=\color{blue},
    tabsize=4,
    showstringspaces=false,
    commentstyle=\color{commentgreen},
    keywordstyle=\color{eminence},
    stringstyle=\color{red},
    basicstyle=\footnotesize\ttfamily,
    emph={int,char,double,float,unsigned,void,public,native,bool,pragma,check,if,return},
    emphstyle={\color{blue}},
    escapechar=\&,
    classoffset=1,
    otherkeywords={\#,>,<,.,;,-,\},\{,!,||,=,~,*,-,/,+,:},
    morekeywords={\#,>,<,.,;,-,\},\{,!,||,=,~,*,-,/,+,:},
    keywordstyle=\color{weborange},
    classoffset=0
}}

\newcommand{\prettylstciao}[0]{
\lstset{language=Prolog,
        frame=ltrb,
        rulecolor=\color{check},
        tabsize=4,
        showstringspaces=false,
        breaklines=true,breakatwhitespace=true,
        showlines=true,
        showspaces=false,showtabs=false,
    commentstyle=\color{commentgreen},
    keywordstyle=\color{eminence},
    stringstyle=\color{red},
    basicstyle=\footnotesize\ttfamily, 
    keywordstyle=\color{weborange},
    emphstyle={\color{check}},
    emph={pred,prop,trust,check,checked,true,rsize,cardinality,not_fails,module,exp,cost,costb,
      steps_ub,steps_lb,size_ub,size_lb,covered,mut_exclusive,cost,use_module,int,calls,success,head_cost,literal_cost,
      mshare,trust_default,int,atm,term,comp,
      is_det,num,var,list,ground,length,terminates,steps_o,resource,entry,impl_defined},
    otherkeywords={:,>,<,>=,=<,.,;,-,!,=,~,*,\&,+,:-,[,],|,->,:=},
    morekeywords= {:,>,<,>=,=<,.,;,-,!,=,~,*,\&,+,:-,[,],|,->,:=},
    escapechar=@,
    escapeinside=~~
      }}

\usetikzlibrary{arrows,positioning,shapes,shadows}
\usepackage{ifthen}

\newcommand{\inputlanguage}{input language\xspace}

\newcommand{\resourceusage}{resource usage\xspace}

\newcommand{\specialization}{specialization\xspace}

\newcommand{\includeapproxdef}[2]{\ifthenelse{\equal{\approxdef}{true}}{#1}{#2}}
\newcommand{\approxdef}{false}

\usepackage{algorithm}
\usepackage{algorithmicx}
\usepackage{algcompatible}

\newcommand{\assign}[2]{#1 $\leftarrow$ #2}
\newcommand{\return}{\textbf{return}}

\newcommand{\lsem}{\mbox{$\lbrack\hspace{-0.3ex}\lbrack$}}
\newcommand{\rsem}{\mbox{$\rbrack\hspace{-0.3ex}\rbrack$}}

\newcommand{\Inten}{\mbox{$ I$}}

\newcommand{\sem}[1]{\lsem #1 \rsem}
\newcommand{\prog}{p}
\newcommand{\p}{{\sem{\prog}}}
\newcommand{\nmetric}{nat}

\newcommand{\underapprox}{\p_{\alpha-}}
\newcommand{\overapprox}{\p_{\alpha+}}

\def\con{\wedge}

\newcommand{\kbd}[1]{\mbox{\tt #1}}
\newcommand{\nt}[1]{\mbox{\it #1}}
\newcommand{\level}{level\xspace}
\newcommand{\levels}{levels\xspace}

\newcommand{\hcir}{HC IR\xspace}
\newcommand{\llvmir}{LLVM IR\xspace}
\newcommand{\llvm}{LLVM\xspace}
\newcommand{\tool}{tool\xspace}

\newcommand{\ial}{\ciao assertion language\xspace}
\newcommand{\ialassertion}{\ciao assertion\xspace}
\newcommand{\ialassertions}{\ciao assertions\xspace}

\newcommand{\isa}{ISA\xspace}
\newcommand{\Isa}{ISA\xspace}
\newcommand{\xc}{XC\xspace}

\definecolor{lightgrey}{rgb}{0.95,0.95,0.95}

\newcommand{\summation}[3]{\sum_{#1}^{#2} #3}
\newcommand{\intoperator}{\Sigma \ \mathrm{d}x}
\newcommand\terminal[1]{`{\tt #1}'}
\newcommand{\gis}{$::=\ $}
\newcommand{\gor}{$|\ $}
\newcommand{\sym}[1]{$\langle #1 \rangle$}

\definecolor{darkgreen}{rgb}{0,0.75,0}
\definecolor{darkblue}{rgb}{0,0,0.75}
\definecolor{lightgreen}{rgb}{0.85,1,0.85}
\definecolor{lightblue}{rgb}{0.85,0.85,1}
\definecolor{grey85}{rgb}{0.45,0.45,0.45}
\definecolor{lightblue2}{rgb}{0.75,0.75,1}
\definecolor{lightblue3}{rgb}{0.70,0.70,0.92}
\definecolor{magenta}{rgb}{1,0,1}
\definecolor{sienna}{rgb}{1,0,1}
\definecolor{maroon4}{rgb}{0.55,0.11,0.38}
\definecolor{chocolate}{rgb}{0.82,0.41,0.12}
\definecolor{red4}{rgb}{0.55,0,0}
\definecolor{darkred}{rgb}{0.75,0,0}

\newcommand\rsd{{\mathbb R}}

\newcommand{\realintervalset}{IS} 
\newcommand{\realinterval}{I} 
\newcommand{\natintervalset}{IS} 
\newcommand{\natinterval}{I}

\newcommand\sizdom[1]{{\mathbb N}^{#1}}
\newcommand\sizdomsingle{{\mathbb N}}

\renewcommand\vec\bar

\newcommand\varx{\textup{\tt x}} 
 
\newcommand\vecx{\vec \varx}

\newcommand{\ciao}{Ciao\xspace}
\newcommand{\ciaopp}{CiaoPP\xspace}

\def\imp{\hbox{${\tt \ :\!-\ }$}}

\newcommand{\optionaltext}[1]{#1}
\renewcommand{\optionaltext}[1]{ }

\newtheorem{definition}{Definition}
\newtheorem{example}{Example} 
\newenvironment{examplebox}{\begin{example}}{\hfill $\Box$\end{example}} 
\newtheorem{theorem}{Theorem}
\newtheorem{lemma}{Lemma} 
\newtheorem{corollary}{Corollary} 

\title[Interval-based Resource Usage Verification] {Interval-based
  Resource Usage Verification\\ [-1mm]
  by Translation into Horn Clauses\\ [-1mm]
and an Application to Energy Consumption$^{\small
  \thanks{
    Research partially funded by EU FP7
      318337 \emph{ENTRA}, Spanish MINECO TIN2012-39391
      \emph{StrongSoft}, and TIN2015-67522-C3-1-R \emph{TRACES}
      projects, and the Madrid M141047003 \emph{N-GREENS}
      program. \newline}}$
}

\author[P. LOPEZ-GARCIA et al.] {P. LOPEZ-GARCIA$^{1,3}$ ~~ L. DARMAWAN$^{1}$ ~~ M. KLEMEN$^{1,2}$ ~~ \\ {\normalsize \emph{U. LIQAT$^{1,2}$ ~~ F. BUENO$^{2}$ ~~ M.V. HERMENEGILDO$^{1,2}$}} \\
\\ 
   $^1$IMDEA Software Institute \\
   \email{\{pedro.lopez,luthfi.darmawan,maximiliano.klemen,umer.liqat,manuel.hermenegildo\}@imdea.org} \\ 
   $^2$Universidad Polit\'ecnica de Madrid (UPM) \\
   \email{bueno@fi.upm.es} \\
   $^3$Spanish Council for Scientific Research  (CSIC) \\
   \vspace*{-1mm}
}

\pagerange{\pageref{firstpage}--\pageref{lastpage}}

\submitted{February 12, 2017}

\newcommand{\secbeg}{\vspace*{-3mm}}
\newcommand{\secend}{}

\newcommand{\ltintervals}{<_{f}}
\newcommand{\leqintervals}{\leq_{f}}
\newcommand{\infinitecalculus}{infinite calculus\xspace}
\newcommand{\finitecalculus}{finite calculus\xspace}

\begin{document}

\label{firstpage}

\maketitle

\begin{abstract}
\ \\ [-12mm]

 In many applications it is important to ensure conformance with
 respect to specifications that constrain the use of resources such as
 execution time, energy, bandwidth, etc.  We present a configurable
 framework for static resource usage verification where specifications
 can include data size-dependent resource usage functions, expressing
 both lower and upper bounds.  Ensuring conformance with respect to
 such specifications is an undecidable problem. Therefore, our
 framework infers resource usage functions (of the same type as the
 specifications, i.e., data-size dependent, and providing upper and
 lower bounds), which \emph{safely approximate} the actual resource
 usage of the program, and which are safely compared against the
 specification.  We start by reviewing how this framework is
 parametric with respect to the programming language by a) translating
 programs to an intermediate representation based on Horn clauses, and
 b) using the configurability of the framework to describe the
 resource semantics of the input language.  We then provide a more
 detailed formalization of the approach and extend the framework so
 that the outcome of the static checking of assertions can generate
 \emph{intervals} of the input data sizes for which assertions hold or
 not, i.e., a given specification can be proved for some intervals but
 disproved for others.  We also generalize the specifications to
 support preconditions expressing intervals within which the input
 data size of a program is supposed to lie.  Most importantly, we
 provide new techniques which extend the classes of resource usage
 functions that can be checked, such as functions containing
 logarithmic or summation expressions, or some functions with multiple
 variables.  We also report on and provide results from an
 implementation within the \ciao/\ciaopp~framework, as well as on a
 practical tool built by instantiating this framework for the
 verification of energy consumption specifications for
 imperative/embedded programs written in the \xc language and running
 on the XS1-L architecture. Finally, we illustrate with an example how
 embedded software developers can use this tool, in particular for
 determining values for program parameters that ensure meeting a given
 energy budget while minimizing the loss in quality of service.

\vspace*{-2mm}
\end{abstract}

\begin{keywords}
\begin{small}
  Static Analysis, Resource Usage Analysis and Verification, Horn
  Clause-based Analysis and Verification, Energy Consumption,
  Program Verification and Debugging.
\vspace*{-5mm}
\end{small}
\end{keywords}

\secbeg
\section{Introduction and Motivation}
\label{sec:intro}
\secend

The conventional understanding of software correctness is absence of
errors or bugs, expressed in terms of conformance of all possible
executions of the program with a functional specification (like type
correctness) or behavioral specification (like termination or possible
sequences of actions).
However, in an increasing number of computing applications,
ranging from those running on devices with limited resources (e.g.,
the ones used in \emph{Internet of Things} applications, sensors,
smart watches, smart phones, portable/implantable medical devices, or
mission critical systems), to large data centers and high-performance
computing systems,
it is also important and sometimes essential to ensure conformance
with respect to specifications expressing 
non-functional global properties such as energy consumption, maximum
execution time, memory usage, or user-defined resources.
For example, in a real-time application, a program completing an
action later than required is as erroneous as a program not computing
the correct answer. The same applies to an embedded application in a
battery-operated device (e.g., a portable or implantable medical
device, an autonomous space vehicle, or even a mobile phone) if the
application makes the device run out of batteries earlier than
required, making the whole system useless in practice.
In general, high performance embedded systems must control, react to,
and survive in a given environment, and this in turn establishes
constraints about the system's performance parameters including energy
consumption and reaction times. Therefore, a mechanism is necessary in
these systems in order to prove correctness with respect to
specifications about such non-functional global properties.

In previous work we have developed a general 
approach to automated verification 
based on a novel combination of assertion-based partial
specifications, static analysis, run-time checking, and testing
~\cite{aadebug97-informal,prog-glob-an,assrt-theoret-framework-lopstr99,ciaopp-sas03-journal-scp,testchecks-iclp09},
and which has been implemented in the \ciaopp~framework.
In addition to different functional properties
(supported by ``pluggable'' \emph{abstract domains}\footnote{By
  pluggable abstract domains we refer to the fact that in CiaoPP new
  abstract domains can be integrated easily as modules implementing
  a well-defined interface. This interface connects each abstract
  domain to the built-in abstract interpretation algorithms (the
  ``fixpoints''), giving
  rise to different program analyzers.
  The same interface also connects the domains to other parts of the
  system that are based on abstractions, such as, e.g., the abstract
  partial evaluators.  }), such as types,
modes, or groundness, this framework can also deal with a large class
of properties related to resource usage, including upper and lower
bounds on execution time, memory, energy, and, in general,
user-definable resources (the latter in the sense
of~\cite{resource-iclp07,resources-bytecode09}).
Such bounds are given as functions on input data sizes
(see~\cite{resource-iclp07}
for the different metrics that can be used to measure data sizes, such
as list length, term depth, or term size).

\begin{figure}[t]
\centering

\begin{tikzpicture}[
   node distance=0.7cm and 0.5cm,
   data/.style={ 
     black, draw, align=center,drop shadow,
     chamfered rectangle,
     chamfered rectangle corners=north west, black, fill=orange!10, draw,
     align=center, text width=2.15cm, 
     font=\scriptsize
   },
   component/.style={
      black, draw, align=center,drop shadow,
      shape=rectangle, fill=darkgreen!10, text width=2.4cm, 
      rounded corners=4pt,
      font=\footnotesize \bf
   },
   wide_component/.style={
      black, draw, align=center,drop shadow,
      shape=rectangle, fill=darkgreen!10, text width=4cm, 
      rounded corners=4pt,
      font=\footnotesize 
   },
   every edge/.style={<-,thick,draw,>=latex},
 ]


\node[data] (JAVASOURCE) {Java Source\\+Assertions};
\node[component] (JAVAC) 
   [below=of JAVASOURCE] {javac};
\node[data] (JAVABYTECODE) 
   [below=of JAVAC] {Java Bytecode};
\node[component] (BYTECODE_TRA) 
   [below=of JAVABYTECODE] {Transformation\\ (\textit{soot})};
   
\node[component] (JAVA_TRA) 
   [right=of BYTECODE_TRA] {Transformation\\ (\textit{java parser})};
   
\node[data] (CIAO) 
   [right=4cm and 2.5cm of JAVASOURCE] {Ciao Source\\+Assertions};

\node[data] (XC) 
   [right=of CIAO] {XC Source\\+Assertions};
\node[component] (XCC) 
   [below=of XC] {xcc};
\node[data] (ISALLVM) 
   [below=of XCC] {ISA/LLVM/...};

\node[component] (ISALLVM_TRA) 
   [below=of ISALLVM] {Transformation\\ (\textit{llvm/isa})};
   
\node[black, draw, align=center,drop shadow,
     chamfered rectangle,
     chamfered rectangle corners=north west, black, fill=orange!10, draw,
     align=left, text width=10cm, 
     font=\footnotesize] (HCIR) 
   [below left=5cm and -5.5cm of CIAO] {\hcir\\ (Horn clauses)\\ \quad};
   
\node[data] (CODE) [right=3cm and 3cm of HCIR.west]{Code};   
\node[data] (ASSRTS) [right= 1.5cm and 1.5cm of CODE]{Assertions};


\node[wide_component] (ANALYSIS) 
[below=of HCIR] {{\bf Static Analyzers}\\
  (e.g. \emph{Resource Analysis})

};
   
\node[data] (ANALYSIS_RESULT) 
[below=of ANALYSIS] { 
Analysis Results};

\node[component] (VERIFICATION) 
   [right=of ANALYSIS] {Static\\ Comparator};
   
\node[data] (VERIFICATION_RESULT) 
[below=of VERIFICATION] { 
Verification Results};   
   
   

\path[->] 
 (JAVAC) edge  node[pos=0.5]  {}  (JAVASOURCE)
 (JAVABYTECODE) edge  node[pos=0.5]  {}  (JAVAC)
 (BYTECODE_TRA) edge  node[pos=0.5]  {}  (JAVABYTECODE) 
 
 (XCC) edge  node[pos=0.5]  {}  (XC)
 (ISALLVM) edge  node[pos=0.5]  {}  (XCC)

 (ISALLVM_TRA) edge  node[pos=0.5]  {}  (ISALLVM)
 (JAVA_TRA) edge[bend right]  node[pos=0.5]  {}  (JAVASOURCE)
 
 (HCIR.20) edge  node {}  (ISALLVM_TRA)
 (HCIR) edge  node[pos=0.5]  {}  (JAVA_TRA)
 (HCIR) edge  node[pos=0.5]  {}  (BYTECODE_TRA)
 (HCIR.31) edge  node[pos=0.5]  {}  (CIAO)  
  
 (ANALYSIS) edge  node[pos=0.5]  {}  (HCIR)
 (VERIFICATION) edge  node[pos=0.5]  {}  (ASSRTS)
 (VERIFICATION) edge  node[pos=0.5]  {}  (ANALYSIS)

(VERIFICATION_RESULT) edge  node[pos=0.5]  {}  (VERIFICATION)
(ANALYSIS_RESULT) edge  node[pos=0.5]  {}  (ANALYSIS)
 ;

 

 \end{tikzpicture}


\caption{Overview of the framework for analysis and verification of
different input programming languages, using Horn clauses as
intermediate language.
}
\label{fig:hcpipeline_ir}
\end{figure}
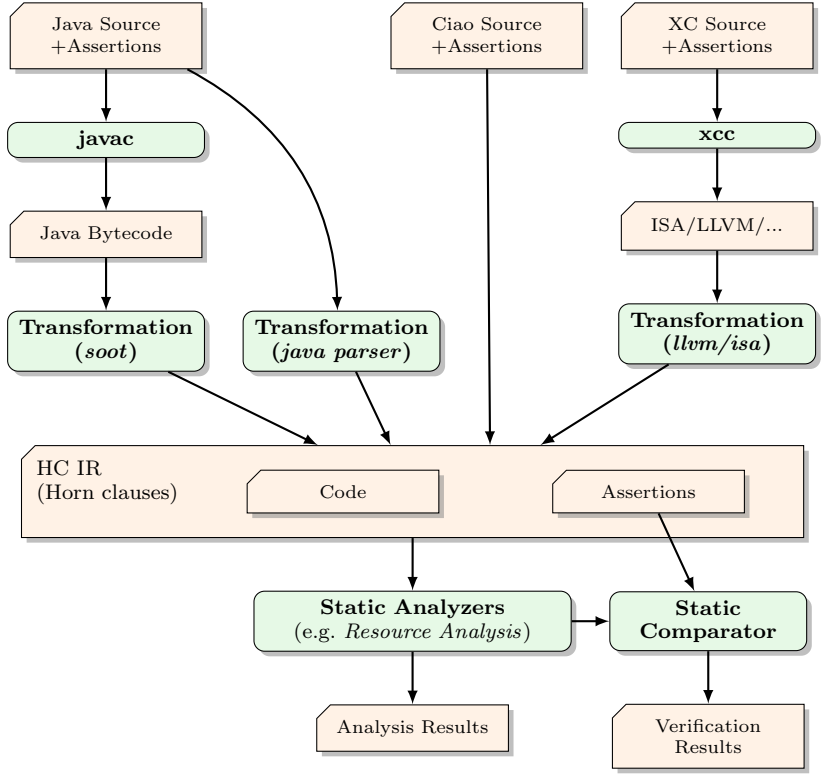

In order to
make our
framework parametric with respect to programming languages and
program representations at different compilation \levels,
each \emph{\inputlanguage} supported (e.g., Java source, Java
bytecode, \xc source, \ciao, LLVM intermediate representation
--\llvmir, or Instruction Set Architecture --\isa) is translated into
an \emph{intermediate
  program representation which is based on Horn
  clauses}~\cite{decomp-oo-prolog-lopstr07} --see
Fig.~\ref{fig:hcpipeline_ir}.
All analysis and verification is
performed on this Horn clause-based representation, that we will refer
to as ``\hcir'' from now on.
I.e., given program $p$ in an \inputlanguage
$L_p$ plus a definition of the semantics of $L_p$, $p$ is translated
into 
a set of Horn clauses capturing the semantics of the program,
$\sem{p}$, or an abstraction of it, $\sem{p}_\alpha$ (see
Sect.~\ref{sec:framework} for a description of this notation).
A Horn clause (HC) is a first-order predicate logic formula of the
form $\forall (S_1 \wedge \ldots \wedge S_n \rightarrow S_0)$ where
all variables in the clause are universally quantified over the whole
formula, and $S_0,S_1,\ldots,S_n$ are atomic formulas, also called
literals.  It is usually written $S_0 \imp S_1,\ldots,S_n$.  This
\hcir~consists of a set 
of connected code \emph{blocks}, each block
represented by a \emph{Horn clause}:
\mbox{$<block\_id>( <params> ) \imp \ S_1, \ \ldots \ ,S_n$}. Each
such block has an entry point, that we call the \emph{head} of the
block (to the left of the $\imp$ symbol), with a number of parameters
$<params>$, and a sequence of steps (the \emph{body}, to the right of
the $\imp$ symbol). Each of these $S_i$ steps (or \emph{literals}) is
either (the representation of) a \emph{call} to another (or the same)
block or
an operation. Such operations depend on the \inputlanguage
represented, i.e., they can be bytecode instructions (from a Java
bytecode program), \isa instructions (from an \isa program), calls to
built-ins or constraints (from a logic program), \llvm instructions,
etc.  The semantics of each bytecode, instruction, built-in, etc. is
provided compositionally to the analyzers by means of \emph{trust}
assertions (see~\ref{sec:ciao-assertion-language}). In the case of
resources, the set of these assertions constitutes the \emph{resource
model} (see Fig.~\ref{fig:hcpipeline_ir} and
Fig.~\ref{fig:energymodel}).
The \hcir~representation offers a good number of features that make it
very convenient for analysis such as supporting naturally Static
Single Assignment (SSA) and recursive forms, making all variable
scoping explicit,
reducing the semantics of all constructs (loops, conditionals,
switches, etc.) to a simple form,
etc.~\cite{decomp-oo-prolog-lopstr07}.

The CiaoPP analyzers handle the \hcir uniformly, regardless of its
origin.  In particular, the resource analysis infers \resourceusage
functions in terms of input data sizes, for all the predicates in the
\hcir program, which are then reflected back to the input language or
representation also as assertions.  This analysis can infer different
classes of \resourceusage functions such as, e.g., polynomial,
exponential, summation, or logarithmic, using the techniques
of~\cite{granularity,caslog,low-bou-sas94,resource-iclp07,NMHLFM08,plai-resources-iclp14}.
\emph{Verification} implies comparing specifications (in our case,
the resource consumption specifications, given in the form of assertions)
against analysis results.
Our focus in this paper is on this comparison process, rather than on
the resource
analysis, which is described
in~\cite{resource-iclp07,plai-resources-iclp14} and its references.
We do not cover the debugging aspect either, i.e., process of finding
the cause of an assertion violation.
Since both static analysis and verification are in general undecidable
our techniques used are necessarily approximate. Nevertheless, such
approximations are \emph{safe}, in the sense that they are guaranteed
to be correct considering all possible executions, i.e., they provide
correct answers or return ``unknown.''

\begin{examplebox}
\label{ex:fact}

\begin{small}
\begin{figure}[t]
\prettylstc
\begin{lstlisting}
#pragma check fact(n) 
        : (1 <= n) ==> (6.0 <= energy_nJ <= 2.3*n+9.0)

int fact(int N) {
  if (N <= 0) return 1;
  return N * fact(N - 1);
}
\end{lstlisting}
\vspace{-2.5mm}
\caption{An XC source (factorial) function.}
\label{fig:xcfactorial}
\vspace{-2.5mm}
\end{figure}
\end{small}

\begin{small}
\begin{figure}[t]
\noindent\begin{minipage}{.38\textwidth}
\prettylstc
\begin{lstlisting}[numbers=left,  numberblanklines=false]
.
.
.
.

<fact>:
001: entsp 0x2	
002: stw   r0, sp[0x1]
003: ldw   r1, sp[0x1]
004: ldc   r0, 0x0
005: lss   r0, r0, r1
006: bf    r0, <008>



007: bu    <010>  
010: ldw   r0, sp[0x1]
011: sub   r0, r0, 0x1
012: bl    <fact>

013: ldw   r1, sp[0x1]
014: mul   r0, r1, r0
015: retsp 0x2


008: mkmsk r0, 0x1
009: retsp 0x2

\end{lstlisting}
\end{minipage}\hfill
\begin{minipage}{.50\textwidth}
\bgroup
\createlinenumber{12}{12a}
\createlinenumber{13}{12b}
\createlinenumber{19}{19a}
\createlinenumber{20}{19b}

\prettylstciao
\begin{lstlisting}[escapechar=@,numbers=left, numberblanklines=false]
:- check pred fact(N, Ret)
   : intervals(@\nmetric@(N),[i(1,inf)])
   + costb(energy_nJ,6.0,
           2.3*@\nmetric@(N)+9.0).

fact(R0,R0_3) :-	
   entsp(0x2),
   stw(R0,Sp0x1),
   ldw(R1,Sp0x1),
   ldc(R0_1,0x0),
   lss(R0_2,R0_1,R1),
   bf(R0_2,0x8),
   fact_aux(R0_2,Sp0x1,R0_3,R1_1).

fact_aux(1,Sp0x1,R0_4,R1) :-
   bu(0x0A),
   ldw(R0_1,Sp0x1),
   sub(R0_2,R0_1,0x1),
   bl(fact),
   fact(R0_2,R0_3),
   ldw(R1,Sp0x1),
   mul(R0_4,R1,R0_3),
   retsp(0x2).

fact_aux(0,Sp0x1,R0,R1) :-
   mkmsk(R0,0x1),
   retsp(0x2).
\end{lstlisting}
\egroup
\end{minipage}
\vspace{-2.5mm}
\caption{\Isa program for Fig.~\ref{fig:xcfactorial} (left) and its
  Horn-clause representation (right).}
\label{fig:asmfactorial}
\end{figure}
\end{small}

\begin{figure*}[t]
\centerline{\includegraphics[scale=0.65,clip,trim=0 15 0 0]{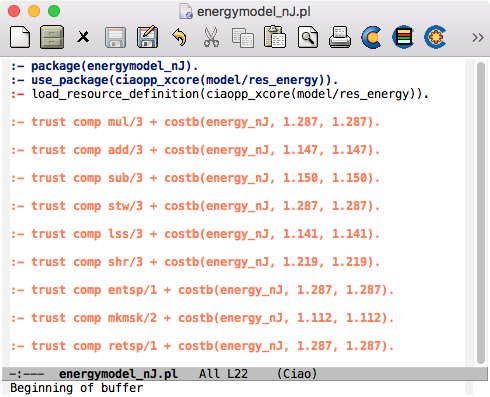}}
\caption{A simple energy model, expressed in the \ial.}
\label{fig:energymodel}
\end{figure*}

Assume that we are interested in verifying specifications about
\emph{energy consumption}.  Consider for 
example the recursive factorial function definition \texttt{fact} in
Fig.~\ref{fig:xcfactorial}, written in the XC C-style
language~\cite{Watt2009}.
The \isa program corresponding to it is generated using the XC
compiler, XCC (left hand side of Fig.~\ref{fig:asmfactorial}). The
resulting \isa program is passed to a translator (see
Fig.~\ref{fig:hcpipeline_ir}) which generates the associated Horn
clauses (right hand side of Fig.~\ref{fig:asmfactorial}). Such \hcir
program, together with the information contained in the energy models
at the \isa level (represented also by using assertions, see
Fig.~\ref{fig:energymodel} for a simple example),
is passed to the resource analysis (as represented in
Fig.~\ref{fig:hcpipeline_ir}), which outputs the energy consumption
analysis results and the verification results for all procedures in
the \hcir program.  
More specifically, the energy model provides the information on the
energy consumed by basic operations 
(\isa instructions in this case). This information is
taken (trusted) by the static analyzer which propagates it,
during the abstract interpretation of the program, through code
segments, conditionals, loops, recursions, etc. , mimicking the
actual execution of the program with symbolic ``abstract'' data
instead of concrete data, in order to infer energy consumption
functions for higher-level entities, such as procedures and functions
in the program.  The analysis of recursive procedures gives rise to
recurrence equations, whose closed form solutions are the resource
usage functions, which depend on input data sizes, resulting from the
analysis.
The \xc assertion:
\begin{center}
\begin{small}
\begin{tabbing}
\texttt{\#pragma \textcolor{check}{check} fact(n) : (1 <= n) ==> (6.0 <= energy\_nJ <= 2.3*n+9.0).}
\end{tabbing}
\end{small}
\end{center}

\noindent is a resource usage specification which also gets translated
into the \hcir representation to be checked by \ciaopp (the
\ial~\cite{assert-lang-disciplbook,hermenegildo11:ciao-design-tplp}),
as shown in lines 1-3 in Fig.~\ref{fig:asmfactorial}
(right):\footnote{ See Sect.~\ref{sec:application-energy} for further
  details on specifications in XC syntax
  and Sect.~\ref{sec:ciao-assertion-language} 
  for their counterpart in the \hcir.}
\begin{center}
\begin{small}
\begin{tabbing}
:-falsecompnr\=\kill
\texttt{:- \textcolor{check}{check} pred fact(N,Ret) : intervals(\nmetric(N),[i(1,inf)])} \\
   \> \texttt{+ costb(energy\_nJ,6.0,2.3*\nmetric(N)+9.0).} 
\end{tabbing}
\end{small}
\end{center}

\noindent
The assertion expresses
that the cost of \texttt{fact(N,Ret)}, in terms of the resource 
``energy in nano-Joules,''\footnote{$1$~nano-Joule = $10^{-9}$~Joules}
must lie in the interval
$[6.0, 2.3*\mathtt{\nmetric(N)}+9.0]$~nJ.
In the \hcir representation, the return values of functions are
represented as additional arguments (\texttt{Ret} as second argument
to \texttt{fact}).
The assertion uses the \texttt{costb/3} property for expressing both a
lower and an upper bound, in the second and third arguments respectively, on a
cost given in terms of a particular resource, in the first argument.
The \texttt{intervals/2} property specifies the set of input sizes,
under a particular metric, for which the assertion has to be checked.
The first argument indicates the input argument that is being
considered, together with the corresponding size metric. The second
argument indicates the set of values as a union of intervals,
represented by a list of \texttt{i/2} properties, which in this example
contains only one interval, $(1,\infty)$.
It provides bounds on the energy to be consumed by
\texttt{fact(N,Ret)} given as functions on the size of the input
argument \texttt{N}. 
Since such argument is numeric, the size metric used is its 
``non-negative value'', 
defined as $\mathtt{\nmetric(N)}
\overset{\mathrm{def}}{=} max(0,\mathtt{N})$.
The $\nmetric(N)$ size metric is applied to a numeric variable $N$,
not to arithmetic expressions. However, our size analysis understands
arithmetic expressions, and can give the size of an output argument as
an arithmetic function that depends on the $\nmetric(N)$ values of
variables that represent input arguments.
\end{examplebox}

\begin{figure}[t]
\includegraphics[width=\textwidth]{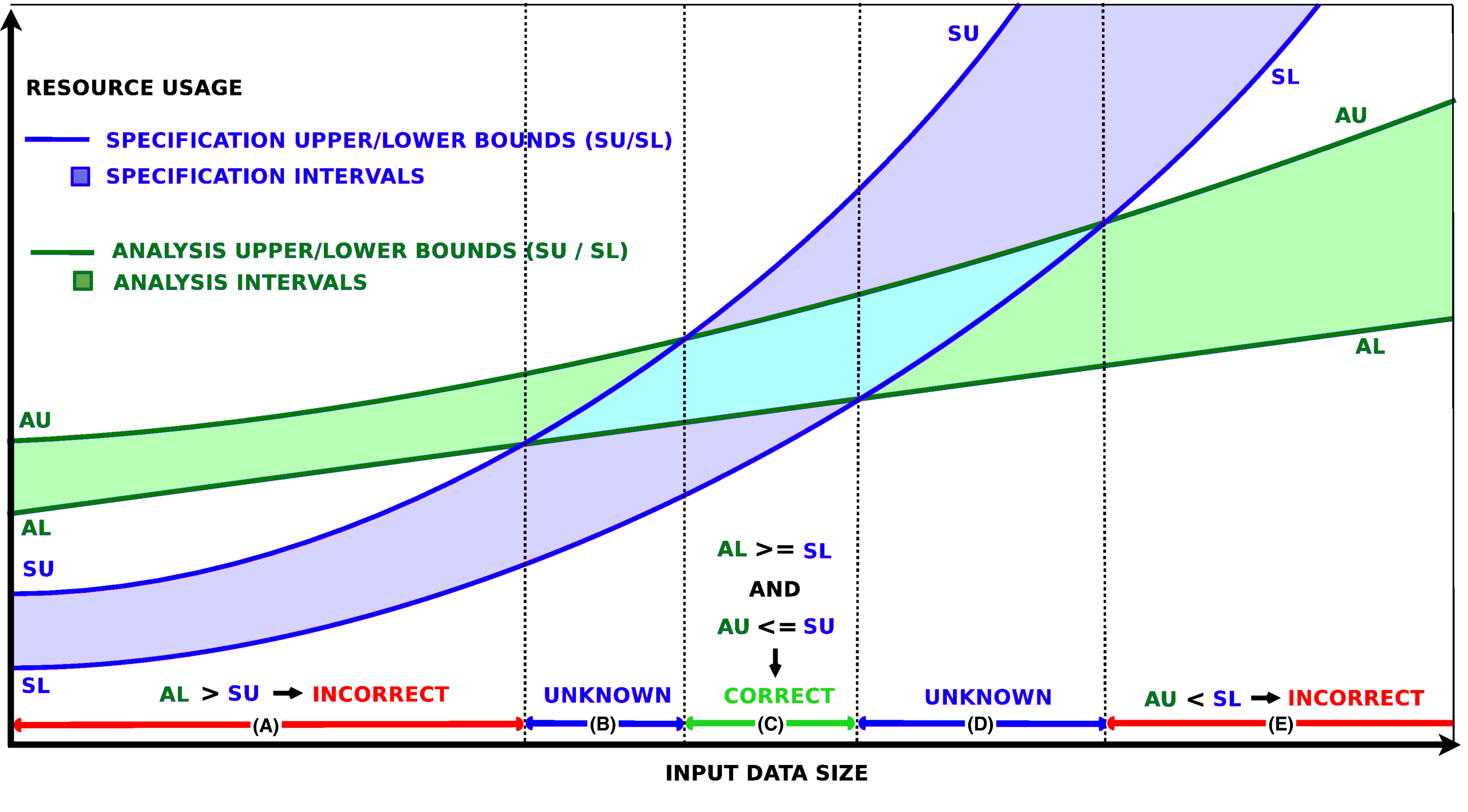}
\vspace*{-10mm}
\caption{Interval-based resource usage verification.}
\label{fig:interval_res_ver}
\end{figure}

As mentioned before, 
the verification of \resourceusage specifications is performed by
comparing the 
abstract intended semantics (i.e., the \resourceusage specifications) with the
safe approximation of the concrete semantics.
inferred by the resource analysis. We say that a program
property $\phi^\#$ is a \emph{safe approximation} of a property
$\phi$, if the set of program traces where $\phi$ holds is included in
the set of program traces where $\phi^\#$ holds. The idea of using
safe approximations is further explained in
Sect.~\ref{sec:framework}.
In our original work on resource usage verification, reported, e.g.,
in~\cite{ciaopp-sas03-journal-scp} and previous papers, for each
property expressed in 
an assertion, the possible outcomes are
\emph{true} (property proved to hold), \emph{false} (property proved
not to hold), and \emph{unknown} (the analysis cannot prove true or
false).
However, it is very common for the cost functions involved in the
comparisons to have intersections, so that for some
input data sizes one of them is smaller than the other one, and for
others it is the other way around.  The first major contribution
of this paper is to generalize our approach so that the answers of the
comparison
process can now
include \emph{conditions} under which the truth or falsity of the
property can be proved. Such conditions can be parameterized by
attributes of inputs, such as input data sizes or value ranges.
In particular, the outcome of the comparison process can now be that
the original specification holds for input data sizes that lie within
a given set of intervals, does not hold for
other intervals, and the result may 
be unknown for some others.
This is illustrated in Fig.~\ref{fig:interval_res_ver}. We can see
that the specification gives both a lower and upper bound cost
function, so that for any input data size $n$ (ordinate axis), the
specification expresses that the resource usage of the computation
with input data of that size
must lie in the interval determined by both functions (which depend on
$n$). Similarly, the bound cost functions inferred by the static
analysis determine a resource usage interval for any $n$, in which the
resource usage of the computation (with input data of size $n$) is
granted to lie. We can see that in 
the (input data size) interval \emph{C}
in the ordinate axis, the program is correct (i.e., it meets
the specification), because for any $n$ in such interval, the resource
usage intervals inferred by the analysis are included in those
expressed by the specification.  In contrast, the program is incorrect in the 
data size intervals \emph{A} and \emph{E}
because the resource usage
intervals inferred by the analysis and those expressed by the
specification are disjoint.  In 
interval \emph{A},
this is proved by
the sufficient condition that says that the lower bound cost function
inferred by the analysis is greater than the upper bound cost function
expressed in the specification (in that interval). A similar reasoning
applies to the 
interval \emph{E}
(using the upper bound of the analysis
and the lower bound of the specification). However, nothing can be
ensured for the 
intervals \emph{B} and \emph{D}.
This is because for any
data size $n$ in such intervals, the resource usage of the computation
for some input data of size $n$ may lie within the interval expressed
by the specification; but for other input data of the same size, the
resource usage may lie outside the interval expressed by the
specification.

Furthermore, intervals can now also appear in specifications, i.e.,
our approach can check specifications that 
include preconditions expressing intervals of input data sizes.  In
that case, the data size intervals automatically generated by the
system 
are 
sub-intervals of the ones given in the specification by the
user.

\begin{examplebox}
\label{ex:fact-cont}

Continuing with Example~\ref{ex:fact}, using the techniques proposed
herein (and the prototype implemented) the outcome of static checking
for the assertion in Figs.~\ref{fig:xcfactorial}
and~\ref{fig:asmfactorial} is the following set of assertions: \\

\begin{center}
\begin{small}
\begin{tabbing}
:-falsecompnr\=\kill
\texttt{:- \textcolor{false}{false} pred fact(N,Ret) : intervals(\nmetric(N),[i(1,1),i(13, inf)])} \\
    \> \texttt{+ ( costb(energy\_nJ, 6.0, 2.3*\nmetric(N)+9.0) ).} \\ \\
\texttt{:- \textcolor{checked}{checked} pred fact(N,Ret) : intervals(\nmetric(N),[i(2,12)])} \\
   \>  \texttt{\ \ + ( costb(energy\_nJ, 6.0, 2.3*\nmetric(N)+9.0)  ).}
\end{tabbing}
\end{small}
\end{center}

\noindent
meaning that the specification does not hold for values of $n$
belonging to the interval $[1,1] \cup [13,\infty]$, and that it does
hold for values of $n$ in the interval $[2,12]$, where $n =
\mathtt{\nmetric(N)}$.
In order to produce that outcome, first \ciaopp's resource analysis
infers the upper and lower bound functions for
the 
energy consumption of the factorial program, 
which in this particular case
are both the same: the function
$(2.845 \ n + 1.94)$~nJ, which obviously implies that this is 
the \emph{exact} cost
function for \texttt{fact/2}.  It is
depicted as a continuous line in Fig.~\ref{fig:fact-plot}.  Thus, the
resource usage of the computation of \texttt{fact/2} with input data
of a given size $n$, is granted to lie in the resource usage
interval
$[2.845 \ n + 1.94,  2.845 \ n + 1.94]$.%
\footnote{As mentioned before, we refer the reader
  to~\cite{resource-iclp07,plai-resources-iclp14} for more details on
  the user-definable version of the resource analysis and
  references.} 
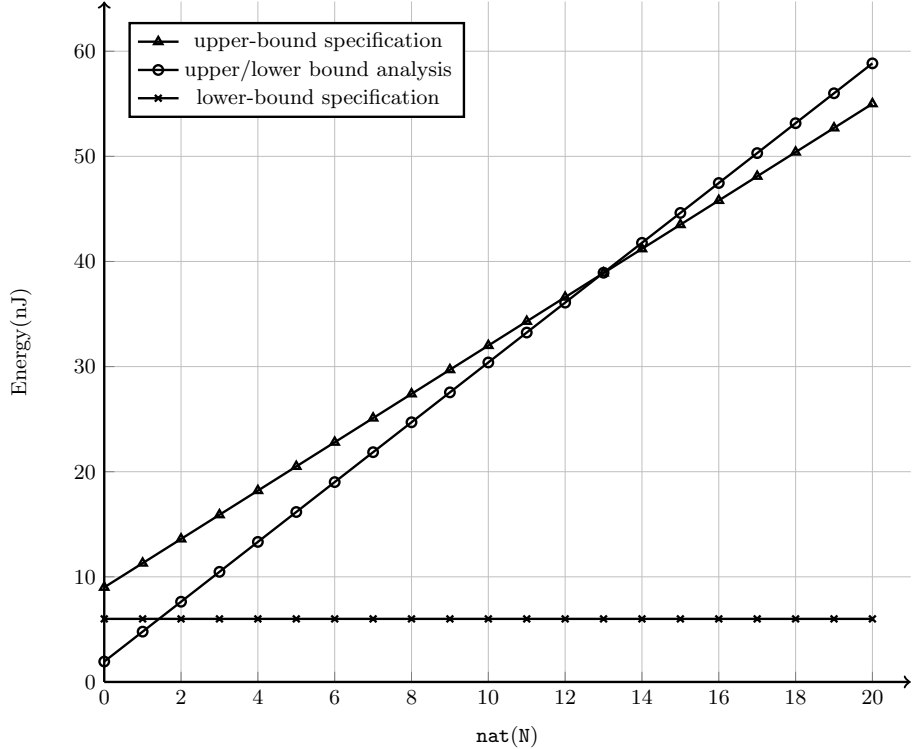
\begin{figure}[t]
\begin{tikzpicture}[scale=0.9]
  \begin{axis}[
    axis x line*=bottom,
    axis y line*=left,
    x axis line style=->,
    y axis line style=->,
    legend style={legend pos=north west,font=\small},
    grid=both,
    grid style={line width=.1pt, draw=gray!10},
    major grid style={line width=.1pt,draw=gray!50},
    xlabel=$\mathtt{nat(N)}$,
    xmin = 0,
    ymin = 0,
    xmax=21,
    samples at = {0,1,2,3,4,5,6,7,8,9,10,11,12,13,14,15,16,17,18,19,20},
    ylabel={Energy(nJ)},
    tick label style={font=\small},
    label style={font=\small},
    line width=1.0pt,
    smooth,    
    width=\linewidth,
  ] 
    \addplot [color=black, mark=triangle] {2.300000*x+9.000000};
    \addlegendentry{upper-bound specification}    

	\addplot [color=black, mark=o]{2.845229*x+1.940746}; 
    \addlegendentry{upper/lower bound analysis}    

    \addplot [color=black, mark=x] {6.000000};    
    \addlegendentry{lower-bound specification}
   
  \end{axis}
\end{tikzpicture}
\caption{Resource usage functions for the factorial program:
  specification and analysis results.}
\label{fig:fact-plot}
\end{figure}
These functions are then compared against the specification resource
functions, depicted in Fig.~\ref{fig:fact-plot} as dashed lines.
For any value $n$ (ordinate axis) of the input data size in
the interval $[2,~ 12]$, the resource usage interval inferred by
\ciaopp~(i.e., $[2.8 \ n + 1.9, 2.8 \ n + 1.9]$) is
included in
the resource usage interval expressed by the
specification, namely: $[6.0, 2.3 \ n + 9.0]$. Therefore, after
performing the resource usage function comparison, using the
techniques that we present, \ciaopp's output indicates that the
assertion is \emph{checked} in that data size interval. Conversely,
the assertion is reported as \emph{false} for $n = 1$ or
$n \in [13,~ \infty]$, because for this interval the
lower bound resource usage function inferred by the analysis is
greater than the upper bound resource usage function expressed in the
specification (and consequently, the corresponding resource usage
intervals are disjoint).
\end{examplebox}

The process of checking of
\resourceusage specifications against
the 
analysis information obviously involves the comparison of arithmetic
functions. In our previous work (again,
see~\cite{ciaopp-sas03-journal-scp} and its references), the approach
to cost function comparison was relatively simple, basically
consisting on performing function normalization and then using some
syntactic and asymptotic comparison rules.  The second major
contribution of this work is to
provide stronger techniques for this purpose, extending the types of
functions that can be dealt with in the specifications and in the
analysis results to a much larger class. 
We also provide
benchmarking results for the proposed interval based, function
comparison techniques.

As a final contribution, and in order to illustrate the usefulness of
the techniques developed, we report on a specialization of the
proposed framework for a practical application: verifying energy
consumption specifications, i.e., comparing inferred energy bound
functions and specifications.
We study the particular case of 
programs written in the \xc language and running on the XMOS XS1-L
architecture, already illustrated in the previous examples.  However,
using our Horn-clause translation approach, the proposed approach and
its implementation in \ciaopp\ 
are general and can be applied to the resource verification of other
programming languages and architectures.
We also illustrate through a case study, 
how embedded software developers can use the tool
developed, in particular for determining values for program parameters
that ensure meeting a given energy budget while minimizing the loss in
quality of service.

\medskip

This paper unifies, improves, and extends our previous work
in~\cite{resource-verif-2012,resource-verif-iclp2010,energy-verification-hip3es2015},
specially by adding operations that allow dealing with a richer set of
usage functions, including summation, exponential and logarithmic cost
functions, as well as multi-variable functions (see
Sect.~\ref{sec:function-comparison}).  We also present a more detailed
formalization
than in~\cite{resource-verif-2012,resource-verif-iclp2010-short}.

\medskip

The overall contributions of this work 
can be summarized as follows:

\vspace*{-2mm}
\begin{itemize}

\item We have developed a configurable framework for static resource
  usage verification where specifications can include data
  size-dependent resource usage functions, expressing both lower and
  upper bounds.

\item We have extended the criteria of software correctness to
  resource usage specifications. In particular, we have defined a
  resource usage semantics and its approximation, and devised
  sufficient conditions for program correctness/incorrectness based on
  such semantics.

\item We have defined operations to check such sufficient conditions,
  that compare the (possibly abstract)
  intended semantics of a program with
  approximated semantics inferred by static analysis. Such comparison
  can deal with a rich class of resource usage functions (polynomial,
  summation, exponential, logarithmic), as well as multi-variable
  functions.

\item Our framework produces a refined output of the assertion
  checking process, that may determine a partition of the set of
  possible input values (by inferring intervals for input data sizes),
  in place of a unique interval of values. Each sub-interval of such
  partition may correspond to different outcomes of the verification.

\item Our framework also
deals with \emph{specifications} containing assertions that include
preconditions expressing intervals for the input data
sizes.

\item We have implemented a prototype and provided experimental
  results.

\item We have specialized our framework for its application to the
  energy consumption verification of imperative (\xc)
  programs.

\end{itemize}

In the rest of the paper
Sect.~\ref{sec:framework} provides an overview of the foundations of
the \ciaopp~verification framework, and of the Ciao assertion language
used for specifications.
Then,
Sect.~\ref{sec:resources-framework} describes how
this traditional framework is extended for the data size, 
\emph{interval-dependent} verification of resource usage properties, 
presenting also the formalization of the framework. In particular, we
define an abstract semantics for resource usage properties and
operations to compare the (approximated) intended semantics of a
program
with approximated semantics inferred by static analysis.
Sect.~\ref{sec:function-comparison} presents our extended techniques
for 
the comparison of (arithmetic) resource usage functions.
Sect.~\ref{sec:implementation} reports on the implementation of our
techniques within the \ciao/\ciaopp~system, providing experimental
results.
Sect.~\ref{sec:application-energy} describes a \specialization of
the framework for its application to the energy consumption analysis
of \xc programs, and 
explains
how embedded software developers can use this tool in the case study
already mentioned. 
Finally, Sect.~\ref{sec:relatedworks} discusses related work and
Sect.~\ref{sec:conclusions} summarizes our conclusions.

\section{Basics of the Verification Framework}
\label{sec:framework}

This section 
summarizes some relevant parts of our previous work
in~\cite{ciaopp-sas03-journal-scp} and previous
papers~\cite{aadebug97-informal,prog-glob-an,assrt-theoret-framework-lopstr99},
that together form the basis for the resource usage verification
techniques described in the following sections, which are the
contributions of this paper.  The section is based mostly
on~\cite{aadebug97-informal}, which provides a basic introduction to
abstract verification from a conceptual point of view. A more detailed
description of the verification framework can be found
in~\cite{assrt-theoret-framework-lopstr99}.

As mentioned before, the verification
framework of~\ciaopp~uses analyses, based on the abstract
interpretation technique, which are provably correct and also
practical, in order to statically compute
safe approximations of the program semantics. These safe
approximations are compared with
specifications, in the form of assertions that are written by the
programmer, in order to prove such specifications correct or incorrect.
In the following we restrict ourselves 
to the important class of fixpoint semantics. Under these assumptions,
the meaning of a program $\prog$, i.e., its 
\emph{concrete semantics}, denoted by 
$\p$, is the least fixpoint of a monotonic operator associated with the
program $\prog$, denoted $S_{\prog}$, i.e., $\p = {\rm lfp}(S_{\prog})$. Such
operator is a function defined on a 
domain $D$, 
which we assume to be a complete lattice.
We will refer to $D$ as the \emph{concrete} domain.
We will assume for simplicity that the elements of $D$ are sets and that
the order relation in $D$ is set inclusion. 

In the abstract interpretation technique, a domain $D_\alpha$ is
defined, called the {\em abstract} domain, which also has a lattice
structure and is \emph{simpler} than the domain $D$. 
In particular, $D$ is finite or, 
if the lattice contains infinite ascending chains, the abstract domain
defines operations that accelerate the convergence of the fixpoint
computation, ensuring termination.
The concrete and abstract domains are related via a pair of monotonic
mappings: {\em abstraction} $\alpha: D\mapsto D_\alpha$, and {\em
  concretization} $\gamma: D_\alpha\mapsto D$, which relate the two
domains by a Galois
connection~\cite{Cousot77}.  Abstract operations over $D_\alpha$ are
also defined for each of the (concrete) operations over $D$.  The
abstraction of a program $\prog$ is obtained by replacing the
(concrete) operators in $\prog$ by their abstract counterparts.  The
{\em abstract semantics} of a program $\prog$, i.e., its semantics
w.r.t.\ the abstract domain $D_\alpha$, is computed (or approximated)
by interpreting the abstraction of the program $\prog$ over the
abstract domain $D_\alpha$. One of the fundamental results of abstract
interpretation is that an abstract semantic operator
$S_{\prog}^\alpha$ for a program $\prog$ can be defined which is
correct w.r.t.\ $S_{\prog}$ in the sense that $\gamma({\rm
  lfp}(S_{\prog}^\alpha))$ is an approximation of $\p$, and, if
certain conditions hold, then the computation of ${\rm
  lfp}(S_{\prog}^\alpha)$ (i.e., the analysis of $\prog$) terminates
in a finite number of steps. We will denote ${\rm
  lfp}(S_{\prog}^\alpha)$, i.e., the result of abstract interpretation
for a program $\prog$, its abstract semantics, as $\p_\alpha$.

Typically, abstract interpretation guarantees that $\p_\alpha$ is a safe 
\emph{over}-approxima\-tion of the abstraction of the 
concrete semantics
of $\prog$ ($\alpha(\p)$), i.e., $\alpha(\p) \subseteq
\p_\alpha$. When $\p_\alpha$ meets such a condition we 
denote it as $\p_{\alpha^+}$.  Alternatively, the analysis can be
designed to safely \emph{under}-approximate the abstraction of the
concrete semantics of $\prog$, i.e., to meet the condition $\p_\alpha\subseteq
\alpha(\p)$. In this case, we use the notation $\p_{\alpha^-}$ to
express that the result of the analysis, $\p_\alpha$, meets such a
condition.

Program verification 
compares the 
{\em concrete semantics} 
$\p$ of a program $\prog$ with an
{\em intended semantics} for the same program, which we will denote by
$\Inten$.  This intended semantics embodies the user's requirements,
i.e., it is an expression of the user's expectations.  In
Table~\ref{tab:corr-comp} we summarize the classical understanding of
some verification problems in a set-theoretic formulation as simple
relations between $\p$ and $\Inten$.
Using the 
concrete or intended semantics for automatic verification
is in general not realistic, since the
concrete semantics 
is typically only partially known, infinite, too
expensive to compute, etc.  
Since the technique of abstract interpretation allows computing
safe approximations of the program semantics, the key idea of the
\ciaopp~approach~\cite{aadebug97-informal,prog-glob-an,assrt-theoret-framework-lopstr99}
is to use the abstract approximation $\p_\alpha$ directly in program
verification
tasks (and in an integrated way with other techniques such as run-time
checking and with the use of assertions).

\begin{table}[t]
  \begin{center}
     \begin{tabular}{l|l}\hline\hline
     {\bf Property}& {\bf Definition}
\\ \hline\hline
     $\prog$ is partially  correct w.r.t.\ $\Inten$& $ \p \subseteq \Inten$
\\
\hline
     $\prog$ is complete w.r.t.\ $\Inten$ & $\Inten\ \subseteq\  \p$
\\
\hline
   $\prog$ is not partially correct w.r.t.\ $\Inten$ & $\p \not\subseteq \Inten$
\\
\hline
     $\prog$ is incomplete w.r.t.\ $\Inten$& $\Inten\ \not\subseteq \p$
\\
 \hline\hline
    \end{tabular}
  \end{center}

  \caption{Set theoretic formulation of verification problems.}
  \label{tab:corr-comp}
\end{table}

\subsection{\textbf{Abstract Verification}}
In the \ciaopp~framework the abstraction
$\p_\alpha$  of the
concrete semantics $\p$ of the program is actually computed and
compared directly to the 
abstract intended semantics, 
which is given
in terms of \emph{assertions}~\cite{assert-lang-disciplbook},
following almost directly the scheme of Table~\ref{tab:corr-comp}.
A program specification $\Inten_\alpha$ is
an abstract value $\Inten_\alpha\in D_\alpha$, where $D_\alpha$ is the
abstract domain of computation.  Program verification is then
performed by comparing $\Inten_\alpha$ and $\p_\alpha$.
Table~\ref{tab:app-corr-comp} shows sufficient conditions for
correctness and completeness w.r.t.\ $\Inten_\alpha$, which can be
used when $\p$ is approximated.  Several instrumental conclusions can
be drawn from these relations.

Analyses which over-approximate the 
concrete
semantics (i.e., those
denoted as $\p_{\alpha^{+}}$), are specially suited for proving
partial correctness and incompleteness with respect to the abstract
specification $\Inten_\alpha$. It will also be sometimes possible to
prove incorrectness in the case in which the semantics
inferred for the program is incompatible with the abstract
specification, i.e., when $\p_{\alpha^{+}}\cap \Inten_\alpha =
\emptyset$. 
On the other hand, we use $\underapprox$ to denote the (less frequent)
case in which analysis under-approximates the 
concrete
semantics. In
such case, it will be possible to prove completeness and
incorrectness. 

\begin{table}[t]
  \begin{center}
    \begin{tabular}{l|c|c}\hline\hline
     {\bf Property}& {\bf Definition}&{\bf Sufficient condition}
\\ \hline\hline
     $\prog$ is partially  correct w.r.t.\ $\Inten_\alpha$& $ \alpha(\p) \subseteq 

\Inten_\alpha$ &$ \overapprox
     \subseteq \Inten_\alpha$ \\
\hline
     $\prog$ is complete w.r.t.\ $\Inten_\alpha$ & $\Inten_\alpha \subseteq 
\alpha(\p)$& $
\Inten_\alpha \subseteq \underapprox$ \\ \hline
$\prog$ is not partially correct
w.r.t.\ $\Inten_\alpha$ & 
$\alpha(\p) \not\subseteq \Inten_\alpha$ & $\underapprox \not\subseteq \Inten_\alpha$, or\ \\
& & 
$ \overapprox \ \cap\ \Inten_\alpha\ =\ \emptyset \wedge \overapprox\neq\emptyset 
\wedge \underapprox\neq\emptyset$ \\ \hline
$\prog$ is incomplete w.r.t.\ $\Inten_\alpha$  & 
$\Inten_\alpha \not\subseteq \alpha(\p)$ & 
$\Inten_\alpha \not\subseteq \overapprox$ \\\hline\hline
\end{tabular}
\end{center}
\caption{Verification problems using approximations.}
\label{tab:app-corr-comp}
\end{table}

Since most of the properties being inferred are in general undecidable,
the technique used to infer such properties, in our
case abstract interpretation, is necessarily approximate.
Nevertheless, such approximations are also always
guaranteed to be safe, in the sense that they are never incorrect,
i.e., they are strict over- (conversely under-) approximations of a
property for the set of all possible program behaviors.

\subsection{Expressing $\Inten_\alpha$: A Relevant Subset of the \ciao~Assertion Language}
\label{sec:ciao-assertion-language}

In order to instantiate the language used to express the intended
semantics, $\Inten_\alpha$, and, in particular, resource usage
properties, we introduce the assertion language that we will use
throughout the paper.  These assertions are part of the
\ciao~assertion language.  For brevity, we only introduce here the
class of ``\texttt{pred}'' assertions, since they suffice for our
purposes.  We refer the reader
to~\cite{assert-lang-disciplbook,ciaopp-sas03-journal-scp,hermenegildo11:ciao-design-tplp}
and their references for a full description of the \ciao~assertion
language.

\paragraph{Pred assertions:} These assertions follow the schema:\\
\centerline{
  \fbox{
    \kbd{:- pred} \nt{Pred} [\kbd{:}
    \nt{Precond\/}] [\kbd{=>} \nt{Postcond\/}] [\kbd{+}
    \nt{Comp-Props\/}]\kbd{.}
  }
}\\
\noindent
where \nt{Pred} is a predicate symbol applied to distinct free
variables,\footnote{We do not consider assertion syntactic sugar such as \emph{modes} for
simplicity.} and \nt{Precond} and \nt{Postcond} are logic formulae about
execution states. An execution state is defined by 
the set of variable/value bindings associated with a given execution step.
The assertion indicates that
in any call to \nt{Pred}, if \nt{Precond} holds in the
calling state and the computation of the call succeeds, then
\nt{Postcond} should hold in the success state.
Also, the set of \nt{Precond}s for all the \nt{pred} assertions for a given  
\nt{Pred} describes all the possible call states, i.e.,
for any call state for a predicate, 
there must be at least one \nt{pred} assertion for
that predicate whose \nt{Precond} holds in that state.

A new
property we introduce in this work and use throughout the paper 
is the following (see
Sect.~\ref{sec:comparing-anstract-semantics} for further details):

\centerline{
  \fbox{
    \textbf{intervals(\nt{$Size_A$},
      [$Int_1$, \ldots, $Int_n$])}
  }
}
\noindent
which expresses that the size $Size_A$ for a given 
argument $A$ belongs to some of the intervals in the list [$Int_1$,
  \ldots, $Int_n$], where $Int_j = \mathtt{i}(Lo,Up)$, $j \geq 1$ and
$\{Lo,Up\} \in \rsd \cup \{\mathtt{inf}\}$.  Finally, the
\nt{Comp-Props} field (appearing after the ``\texttt{+}'' operator) is
used to describe properties of the whole computation for calls to
predicate \nt{Pred} that meet \nt{Precond}. In our application the
\nt{Comp-Props} are precisely the resource usage properties.  As
already shown in Example~\ref{ex:fact}, another
global non-functional property we introduce in this paper is
\texttt{costb/3}, which expresses such resource usages, and follows
the schema:

\centerline{
  \fbox{
    \textbf{costb(\nt{Res\_Name},
      \nt{Low\_Arith\_Expr, \nt{Upp\_Arith\_Expr})} }
  }
}

\noindent
where \nt{Res\_Name} is a user-provided identifier for the resource
the assertion refers to, \nt{Low\_Arith\_Expr} and
\nt{Upp\_Arith\_Expr} are arithmetic functions that map input data
sizes to resource usages, representing respectively lower and upper
bounds on the resource consumption.
Similarly to \texttt{costb/3}, the \texttt{cost/3} property allows
expressing only one resource usage function on input data sizes that
follows this schema:
\centerline{
  \fbox{
    \textbf{cost(\nt{Bound\_Type},
      \nt{Res\_Name}, \nt{Arith\_Expr})}
  }
}
\noindent
where \nt{Res\_Name} is the same as in \texttt{costb/3}, \nt{Arith\_Expr} is
similar to \nt{Low\_Arith\_Expr} and \nt{Upp\_Arith\_Expr} in \texttt{costb/3},
but it can be either upper or lower bound depending on the value of \nt{Bound\_Type}
which are \texttt{lb} for lower bounds and \texttt{ub} for upper bounds. 
This is illustrated in Example~\ref{ex:nrevcomp}.

\begin{figure}[t]
\prettylstciao

\begin{lstlisting}
:- pred append(A, B, C) 
   : ( list(A), list(B), var(C),
       intervals(length(A),[1,inf])
     ) 
  => list(C)
   + costb(steps, length(A)+1, length(A)+1).
\end{lstlisting}
\caption{An example Ciao \emph{resource} assertion for
  \texttt{append/3}.}
\label{fig:ciao_assertion_example}
\end{figure}


\begin{examplebox}
\label{ex:append_assertion}
Fig.~\ref{fig:ciao_assertion_example} shows an assertion for a typical
\texttt{append/3} predicate.   
The assertion states that for any call to predicate \texttt{append/3}
with the first and second arguments bound to lists and the third one
unbound, where the length of the first list
lies in the interval $[1,\infty]$, it holds that
if the call succeeds, then the third argument will also be bound to a
list.
It also states that \texttt{length(A) + 1} is both a lower and upper
bound on the number of resolution steps required to execute any of
such calls.  The property \texttt{length/1} represents
a size metric, in particular,
the length of a list.  In this case, the assertion
expresses an exact cost, since the lower- and upper-bound cost
functions coincide.

\end{examplebox}

\paragraph{Assertion status:} Each assertion 
has an associated \emph{status}, marked with one of the following
prefixes, placed just before the \texttt{pred} keyword: \texttt{check}
(indicating that the assertion is to be checked), \texttt{checked}
(the assertion has been checked and proved correct by the system),
\texttt{false} (it has been checked and proved incorrect by the
system; a compile-time error is reported in this case), \texttt{trust}
(the assertion provides information coming from the programmer in
order to guide the analyzer, and it will be trusted), or \texttt{true}
(the assertion is a result of static analysis and thus correct, i.e.,
it is a safe approximation of the concrete semantics).  The default
status, i.e., if no status appears before \texttt{pred}, is
\texttt{check}.

\section{Extending  the Framework to Data Size-Dependent Resource Usage Verification}
\label{sec:resources-framework}

As mentioned before, our data size-dependent resource usage
verification framework is characterized by being able to deal with
specifications that include both lower and upper bound resource usage
functions (i.e., specifications that express intervals where the
resource usage is supposed to be included in), and, in an extension of
the classical model~\cite{aadebug97-informal,ciaopp-sas03-journal-scp}
and~\cite{resource-verif-iclp2010}, that include preconditions
expressing intervals within which the input data size of a program is
supposed to lie~\cite{resource-verif-2012}.

We start by providing a formalization of our data size-dependent
resource usage verification framework, assuming that the programs that
we are dealing with are written in the \hcir language (i.e., they are
logic programs). However, as mentioned before, the techniques apply to
other languages, by applying our transformation to Horn
clauses. Furthermore, the concepts are in fact also applicable
directly to other languages, with some adaptations and changes in
terminology.

\subsection{Resource usage semantics}
\label{sec:resource-usage-semantics}

Given a program $\prog$, let ${\cal C}_p$ be the set of all calls to $\prog$.
The concrete resource usage semantics of a program $\prog$, for a
particular resource of interest, $\p$, is a set of pairs $(p(\bar
t), r)$ such that $\bar t$ is a tuple of terms (not necessarily ground), $p(\bar t) \in {\cal
  C}_p$ is a call to 
$\prog$ with actual parameters $\bar t$,
and $r$ is a number expressing the amount of resource usage of the
computation of the call $p(\bar t)$. Such a semantic object can be
computed by a suitable operational semantics, such as SLD-resolution,
adorned with the computation of the resource usage. We abstract away
such computation, since it will in general be dependent on the
particular resource $r$ it refers to. 
The concrete resource usage semantics can be defined as a relation $\p
\subseteq {\cal C}_p \times \rsd$,
where $\rsd$ is the set of real numbers (note that depending on the
type of resource we can take another set of numbers, e.g., the set of
natural numbers).  Such relation is usually a function.  In other
words, the domain $D$ of the concrete semantics is $2^{{\cal C}_p
  \times \rsd}$, so that $\p \in D$.
Recall that, as described in Sect.~\ref{sec:framework},
$D$ is
a complete lattice, and the abstract domain, $D_\alpha$ has also a
lattice structure.
The concretization and abstractions functions ($\gamma$ and $\alpha$
respectively) are mappings that relate both domains, altogether
composing a Galois
connection~\cite{Cousot77}.

We define an abstract domain $D_\alpha$ whose elements are sets of
pairs of the form $(p(\bar v):c(\bar v), \Phi)$, where $p(\bar
v):c(\bar v)$, is an abstraction of a set of calls and $\Phi$ is an
abstraction of the resource usage of such calls.  We refer to such
pairs as \emph{call-resource} pairs.  
Specifically, $\bar v$ is a
tuple of variables and $c$ is a property on terms, so that $p(\bar
v):c(\bar v)$ represents the set of all calls $p(\bar t)$ such that
$\bar v = \bar t \rightarrow c(\bar v)$ holds.

The abstraction $c(\bar v)$ is some subset of the abstract domains
available for the analyzer, i.e., those loaded in the \ciaopp~system,
expressing program states.  An example of $c(\bar v)$ (in fact, the
one used in Sect.~\ref{sec:implementation} in our experiments) is a
combination of properties which are in the domain of the regular type
analysis, \emph{eterms}~\cite{eterms-sas02}, and properties such as
groundness and freeness present in the \emph{shfr} abstract
domain~\cite{ai-jlp}.  For conciseness, we refer to such combination
as the mode/type abstract domain.
A regular type is a set of terms which is the language accepted by a
(possibly non-deterministic) \emph{finite tree automaton}, although
regular types can be expressed using several type representations. Internally, the
\emph{eterms} regular type analysis~\cite{eterms-sas02}
uses a representation based on \emph{regular term grammars}, equivalent
to~\cite{Dart-Zobel} but with some adaptations. This analysis produces
abstractions, represented by using \emph{regular term grammars}, that
overapproximate the set of terms that can occur at all program
points. Such abstractions are presented to the user in the form
of predicates, as will be illustrated later.

We refer to $\Phi$ as a {\em resource usage interval function} for
$\prog$, defined as follows:

\begin{definition}
\label{def:bound-function}
A \emph{resource usage bound function} for $\prog$ is a monotonic
arithmetic function, $\Psi_p: S \mapsto \rsd_{\infty}$, for a
given subset 
$S \subseteq \sizdom{k}$,
where 
$\sizdomsingle$ is 
is the set of 
natural
numbers, $k$ is the number of input arguments to predicate
$\prog$, and $\rsd_{\infty}$ is the set of real numbers augmented with
the special symbols $\infty$ and $-\infty$. We use such functions to
express lower and upper bounds on the resource usage of predicate $\prog$
depending on its input data sizes.
\end{definition}

\begin{definition}
\label{def:interval-function}
A \emph{resource usage interval function} for $\prog$ is an arithmetic
function, $\Phi: S \mapsto {\cal RI}$, where $S$ is defined as before
and ${\cal RI}$ is the set of intervals of real numbers, such that
$\Phi(\bar n) = [\Phi^{l}(\bar n), \Phi^{u}(\bar n)]$ for all $\bar n
\in S$, where $\Phi^l(\bar n)$ and $\Phi^u(\bar n)$ are {\em resource
  usage bound functions} that denote the lower and upper endpoints of
the interval $\Phi(\bar n)$ respectively for the tuple of input data
sizes $\bar n$.\footnote{Although $\bar n$ is typically a tuple of
  natural numbers, we do not restrict the framework to this case.}  We
require that $\Phi$ be well defined so that $\forall \bar n
\ (\Phi^l(\bar n) \leq \Phi^u(\bar n))$.
\end{definition}

\noindent
Intuitively, $\Phi$ defines a resource usage band, and $\Phi(\bar n) =
[\Phi^{l}(\bar n), \Phi^{u}(\bar n)]$ is resource usage interval.

In order to relate the elements $p(\bar v):c(\bar v)$ and $\Phi$ in
a call-resource pair as the one described previously, we assume the
existence of two functions $input_p$ and $size_p$ associated with each
predicate $\prog$ in the program.  Assume that $\prog$ has $k$ arguments and
$i$ input arguments ($i \leq k$). The function $input_p$ takes a
$k$-tuple of terms $\bar t$ (the actual arguments of a call to $\prog$)
and returns a tuple with the input arguments to $\prog$. This function is
generally inferred by using existing 
analysis that infer groundness, freeness and sharing information, but
can also be given by the user by means of assertions.  The function
$size_p(\bar w)$ takes a $i$-tuple of terms $\bar w$ (the actual input
arguments to $\prog$) and returns a tuple with the sizes of those
terms under a given metric. The metric used for measuring the size of
each argument of $\prog$ is automatically inferred (based on type
analysis information), but again can also be given by the user by
means of assertions~\cite{resource-iclp07}.

\begin{small}
\begin{figure}[t]
\centering
\prettylstciao
\begin{lstlisting}
:- module(rev, [nrev/2], [assertions,regtypes,
                          nativeprops,predefres(res_steps)]).

:- entry nrev(A,B) : (list(A, gnd), var(B)).
:- check pred nrev(A,B)  
     + costb(steps, length(A), 10*length(A)).

nrev([],[]).
nrev([H|L],R) :- nrev(L,R1), append(R1,[H],R).
\end{lstlisting}
\vspace*{-4mm}
  \caption{A module for the naive reverse program.}
  \label{fig:nrev}
\end{figure}
\end{small}

\begin{examplebox}
\label{ex:size-input}
Consider for example the naive reverse (\ciao) Prolog program in Fig.~\ref{fig:nrev},
with the classical definition of predicate \texttt{append}.
The first argument of \texttt{nrev/2} is declared input, and the two first 
arguments of \texttt{append} are consequently inferred to be also input.
The size measure for all of them is inferred to be \emph{list-length}.
Then, we have that: \\
$input_{nrev}((x, y)) = (x)$, $input_{app}((x, y, z)) = (x, y)$, \\ 
$size_{nrev}((x)) = (length(x))$ and $size_{app}((x, y)) = (length(x), length(y))$.
\end{examplebox}

\noindent
We define the concretization function $\gamma: D_\alpha \mapsto D$ as
follows: \\ \centerline{$\forall E \in D_\alpha, \gamma(E) =
  \bigcup_{e \in E} \gamma_1(e)$}

\noindent
where $\gamma_1$ is another concretization function, applied to
call-resource pairs $e$'s
of the form $(p(\bar v):c(\bar v), \Phi)$.
We define: \\
\centerline{$\gamma_1((p(\bar v):c(\bar v), \Phi)) = \{(p(\bar t), r) \ | \ \bar t \in
  \gamma_m(c(\bar v)) \con \bar n = size_{p}(input_{p}(\bar t)) \con r \in [\Phi^{l}(\bar n), \Phi^{u}(\bar n)]\}$}

\noindent
where $\gamma_m$ is the concretization function of the mode/type
abstract domain. We use the subscript $m$ as a short name for such a
mode/type domain for conciseness.
The concretization function $\gamma_1$ returns a set of concrete pairs
$(p(\bar t), r)$. As already stated, each such set is an element of the
concrete domain $D = 2^{{\cal C}_p \times \rsd}$, where $\bar t$ is a
tuple of terms, $p(\bar t) \in {\cal C}_p$ is a call to predicate
$\prog$ with actual parameters $\bar t$, and $r$ is a number
expressing the amount of resource usage of the complete computation of
the call $p(\bar t)$.

\begin{examplebox}
\label{ex:concretization}
Assume that $p$ is the predicate \texttt{nrev} in Fig.~\ref{fig:nrev},
$\bar v$ is $(x, y)$, and $c(\bar v)$ is the property defined as the
conjunction $\texttt{list}(x) \land \texttt{var}(y)$, represented as
$(\texttt{list}(x), \texttt{var}(y))$ in the assertions, since we use
the comma (,) as the symbol for the conjunction operator.  The
property $\texttt{list}(\_)$ is a \emph{regular type}, which can be inferred
by \ciaopp by performing the analysis with the \emph{eterms} abstract
domain~\cite{eterms-sas02}, and is represented as a predicate:

\begin{small}
\centering
\prettylstciao
\begin{lstlisting}
list([]).
list([H|R]) :- list(R).
\end{lstlisting}
\end{small}

\noindent
The property $\texttt{var}(\_)$ can also be inferred by \ciaopp, with the
\emph{shfr} abstract domain~\cite{ai-jlp}.

Under these assumptions, $\gamma_m(c(\bar v))$ is the infinite set:\\
\centerline{$\gamma_m(c(\bar v)) = \gamma_m(\texttt{list}(x) \land \texttt{var}(y)) =
  \{([], y), ([a], y), ([a, b], y), ([a, b, c], y), \ldots\}$.}
\noindent
Assume also that $input_{nrev}((x, y)) = (x)$ and $size_{nrev}((x)) =
(length(x))$, as explained in Example~\ref{ex:size-input}. Let
$\{e_\alpha \} \in D_\alpha$, such that:\\ \centerline{$e_\alpha
  \equiv ((nrev(x, y):(\texttt{list}(x) \land \texttt{var}(y))), [\Phi_{nrev}^{l},
    \Phi_{nrev}^{u}])$,}
\noindent
where the resource usage bound functions $\Phi_{nrev}^{l}$ and
$\Phi_{nrev}^{u}$ are defined as:\\
\centerline{$\Phi_{nrev}^{l}(n) = 2 \times n$, and $\Phi_{nrev}^{u}(n) = 1+n^2$.}
\noindent
We have that $([a,b,c], y) \in \gamma_m(\texttt{list}(x) \land \texttt{var}(y))$ and
$size_{nrev}(input_{nrev}([a,b,c], y)) = size_{nrev}([a,b,c]) =
length([a,b,c]) = 3$. Thus, $\Phi_{nrev}^{l}(3) = 2 \times 3 = 6$ and
$\Phi_{nrev}^{u}(3) = 1 + 3^2 = 10$, which means that any pair
$(nrev([a,b,c], y), r)$ such that $r \in [6,10]$, belongs to
$\gamma_1(e_\alpha)$, e.g., $(nrev([a,b,c], y), 6) \in
\gamma_1(e_\alpha)$ and $(nrev([a,b,c], y), 7) \in
\gamma_1(e_\alpha)$.

\noindent
Therefore, we have that $\gamma_1(e_\alpha) = e$, where $e \in D$ is
the infinite set:

$
\begin{array}{l}
e = \{ (nrev([], y), 0), (nrev([], y), 1), (nrev([a], y), 2),
(nrev([a, b], y), 4),  \\ 
(nrev([a, b], y), 5), (nrev([a, b, c], y),
6), (nrev([a, b, c], y), 7), (nrev([a, b, c], y), 10) \ldots\}
\end{array}
$

\noindent
Finally, $\gamma(\{e_\alpha \}) = \gamma_1(e_\alpha) = e$.
\end{examplebox}

The definition of the abstraction function $\alpha: D\mapsto D_\alpha$
is straightforward, given the definition of the concretization
function $\gamma$ above.

\paragraph{\bf Intended meaning.}
As already mentioned, the intended semantics is an expression of the user's
expectations, and is typically only partially known. For this reason
it is in general not realistic to use the exact intended
semantics and we use an approximated intended semantics 
instead.  We define the approximated intended semantics
$\Inten_\alpha$ of a program as a set of \emph{call-resource} pairs
$(p(\bar v):c(\bar v), \Phi)$, identical to those previously used in
the abstract semantics definition. However, the \emph{call-resource}
pairs defining the approximated intended semantics are provided by
the user 
by means of the 
\ciao~assertion language,
introduced in Sect.~\ref{sec:ciao-assertion-language},
while the pairs corresponding to the approximated 
semantics
of the program are automatically inferred by \ciaopp's analysis
tools. In particular, each one of such pairs is represented as a
resource usage assertion for predicate $\prog$ in the program.

As mentioned in Sect.~\ref{sec:ciao-assertion-language}, we will be
using \texttt{pred} assertions.  The most common syntactic schema of a
\texttt{pred} assertion that describes resource usage and its
correspondence to the \emph{call-resource} pair it represents is the
following:

\begin{center}
\fbox{\kbd{:- pred} \nt{$p(\bar v)$} \kbd{:} \nt{$c(\bar v)$}\  \kbd{+} \nt{$\Phi$}\kbd{.} }
\end{center}
which expresses that for any call to predicate \nt{$\prog$}, if
(precondition) \nt{$c(\bar v)$} is satisfied in the calling state,
then the resource usage of the computation of the call is in the
interval represented by \nt{$\Phi$}. Note that \nt{$c(\bar v)$} is a
conjunction of program execution state properties, i.e., properties
about the terms to which program variables are bound to.
As already said, we use the
comma (,) as the symbol for the conjunction operator. If the
precondition \nt{$c(\bar v)$} is omitted, then it is assumed to be the
``top'' element of the lattice representing calls, i.e., the one that
represents any call to predicate $\prog$. The syntax used to express
the resource usage interval function \nt{$\Phi$} is a conjunction of
\texttt{costb/3} or \texttt{cost/3}
properties. 

Assuming that $\Phi(\bar n) = [\Phi^{l}(\bar n),
\Phi^{u}(\bar n)]$, where $\bar n = size_{p}(input_{p}(\bar
v))$, $\Phi$ can be represented in the resource usage assertion as the
conjunction: \\
\centerline{(\texttt{cost}(lb, $r$, $\Phi^{l}(\bar n)$),
\texttt{cost}(ub, $r$, $\Phi^{u}(\bar n)$))} 

\noindent
or, alternatively, using the \texttt{costb/3} property: \\
\centerline{\texttt{costb}($r$, $\Phi^{l}(\bar n)$, $\Phi^{u}(\bar n)$)} 

\noindent
We use Prolog syntax for variable names (variables start with
uppercase letters).

\begin{examplebox}
\label{ex:nrevcomp}
In the program of Fig.~\ref{fig:nrev} one could use the assertion:

\begin{center}
\begin{small}
\begin{tabbing}
\texttt{:- pred} \= \texttt{nrev(A,B) : ( list(A, gnd), var(B) )}\\ 
\> \texttt{+ ( cost(lb, steps, 2 * length(A)),} \\ 
\> \texttt{\ \ \ \ cost(ub, steps, 1 + exp(length(A), 2) )).}
\end{tabbing}
\end{small}
\end{center}

\noindent
to express that for any call to \texttt{nrev(A,B)} with the first
argument bound to a ground list and the second one a free variable, a
lower (resp.\ upper) bound on the number of resolution \texttt{steps}
performed by the computation is $2 \times length(A)$
(resp.\ $1+length(A)^2$). The property $\texttt{list}(\_, \_)$ is
represented as a higher order predicate:

\begin{small}
\centering
\prettylstciao
\begin{lstlisting}
list([], T).
list([H|R], T) :- T(H), list(R).
\end{lstlisting}
\end{small}

\noindent
and the property $\texttt{gnd}(\_)$, expressing ``groundness'', can
also be inferred by \ciaopp, with the \emph{shfr} abstract
domain~\cite{ai-jlp}.

\noindent
In this example, $\prog$ is $nrev$, $\bar v$ is \texttt{(A, B)}, $c(\bar
v)$ is \texttt{( list(A, gnd), var(B) )}, $\bar n =
size_{nrev}(input_{nrev}((A, B))) = (length(A))$, where the functions
$size_{nrev}$ and $input_{nrev}$ are those defined in
Example~\ref{ex:size-input}, and the interval $\Phi_{rev}(\bar n)$
approximating the number of resolution steps is $[2 \times length(A),
1+length(A)^2]$ (in other words, we are assuming that
$\Phi_{nrev}^{l}(x) = 2 \times x$ and  $\Phi_{nrev}^{u}(x) = 1+x^2$).
If we omit the \texttt{cost} property expressing the lower bound
(\texttt{lb}) on the resource usage, the minimum of the interval is
assumed to be zero (since the number of resolution steps cannot be
negative). If we assume that the resource usage can be negative, the
interval would be $(-\infty, 1+n^2]$. This information can be given by
  the user when providing the assertions that constitute the
  definition of a particular resource and its cost model (which
  expresses the resource usage of basic elements of a
  program/language). A detailed description of our user-definable
  resource analysis framework is given in~\cite{resource-iclp07}.
Similarly, if the upper bound (\texttt{ub}) is omitted, the upper
limit of the interval is assumed to be $\infty$.
\end{examplebox}

\begin{examplebox}
The assertion in Example~\ref{ex:nrevcomp}
is applicable for the following concrete semantic pairs:
\begin{verbatim}
( nrev([a,b,c,d,e,f,g],X), 35 )       ( nrev([],Y), 1 )
\end{verbatim}
but it is not applicable to the following ones:
\begin{verbatim}
( nrev([A,B,C,D,E,F,G],X), 35 )       ( nrev(W,Y), 1 )
( nrev([a,b,c,d,e,f,g],X), 53 )       ( nrev([],Y), 11 )
\end{verbatim}
Those in the first line above 
do not meet the
assertion's precondition $c(\bar v)$: the leftmost one because
\texttt{nrev/2} is called with the first argument bound to a list of
unbound variables (denoted by using uppercase letters), and the other
one because the first argument of \texttt{nrev/2} is an unbound
variable. The concrete semantic pairs on the second line will never
occur during execution because they violate the assertion, i.e., they
meet the precondition \nt{$c(\bar v)$}, but the resource usage of
their execution is not within the limits expressed by \nt{$\Phi$}.
\end{examplebox}

\subsection{Comparing Abstract Semantics: Correctness}
\label{sec:comparing-anstract-semantics}

The definition of partial correctness has been given by the condition
$\p \subseteq \Inten$ in Table~\ref{tab:corr-comp}. However, we
have already argued that we are going to use an approximation
$I_\alpha$ of the intended semantics $\Inten$, where $I_\alpha$ is
given as a set of \emph{call-resource} pairs of the form $(p(\bar v):c(\bar v),
\Phi)$.

\begin{definition}[Input-size set]
Let $e_\alpha$ be a call-resource abstract pair $(p(\bar v):c(\bar v),
\Phi)$. We define the \emph{input-size set} of $e_\alpha$, denoted
$input\_size\_set(e_\alpha)$ as the set $\{\bar n \ | \ \exists \ \bar
t \in \gamma_m(c(\bar v)) \con \bar n = size_{p}(input_{p}(\bar
t))\}$.  The input-size set is represented as an interval (or a union
of intervals).  We obviously require that $input\_size\_set(e_\alpha)
\subseteq Dom(\Phi)$ for any call-resource abstract pair $e_\alpha$,
where $Dom(\Phi)$ denotes de domain of function $\Phi$.  $\Box$
\end{definition}

\begin{definition}
\label{def:correct-pair}
We say that $\prog$ is partially correct with respect to a
call-resource pair $(p(\bar v):c_I(\bar v), \Phi_I)$ if for all
$(p(\bar t), r) \in \p$ (i.e., $p(\bar t) \in {\cal C}_p$ and $r$ is
the amount of resource usage of the computation of the call $p(\bar
t)$), it holds that: if $\bar t \in \gamma_m(c_I(\bar v))$ and $\bar n
= size_{p}(input_{p}(\bar t))$, then $r \in \Phi_I(\bar n)$, where
$\gamma_m$ is the concretization function of the mode/type abstract
domain.
\end{definition}

\begin{lemma}
  $\prog$ is partially correct with respect to $\Inten_{\alpha}$,
  i.e. $\p \subseteq \gamma(\Inten_{\alpha})$ if:

\begin{itemize}
\item For all $(p(\bar t), r) \in \p$,
  there is a pair $(p(\bar v):c_I(\bar v), \Phi_I)$ in
  $\Inten_{\alpha}$ such that $\bar t \in \gamma_m(c_I(\bar v))$, and

\item $\prog$ is partially correct with respect to every pair in $\Inten_{\alpha}$.
\end{itemize}

\end{lemma}

Note that the notion of $\prog$ being partially correct with respect
to a call-resource pair $(p(\bar v):c_I(\bar v), \Phi_I)$ is different
from the notion of $\prog$ being partially correct with respect to a
singleton set $\{(p(\bar v):c_I(\bar v), \Phi_I)\}$, i.e., an 
intended semantics: if for all $(p(\bar t), r) \in \p$ it holds that
$\bar t \not\in \gamma_m(c_I(\bar v))$, then $\prog$ is partially
correct with respect to
$(p(\bar v):c_I(\bar v), \Phi_I)$ but $\prog$ is not partially correct
with respect to
$\{(p(\bar v):c_I(\bar v), \Phi_I)\}$.

As mentioned before, we use a safe over-approximation of the program
semantics $\p$, that we denote $\overapprox$, and is automatically
computed by the static analysis
in~\cite{resource-iclp07,plai-resources-iclp14} as a set of
\emph{call-resource} pairs of the form $(p(\bar v):c(\bar v), \Phi)$.
For simplicity, we assume that
$\overapprox$ is a set made up of a single call-resource pair.
The description of how the resource usage bound functions appearing in
$\overapprox$ are computed is out of the scope of this paper, and it
can be found in~\cite{resource-iclp07,plai-resources-iclp14} and its
references.  The safety of such resource usage analysis can
be expressed as follows:

\begin{lemma}[Safety of the static resource usage analysis]
\label{lem:analysis-safety}
Let $e_\alpha = (p(\bar v):c(\bar v), \Phi)$ and $\overapprox = \{
e_\alpha \}$.  For all $(p(\bar t), r) \in \p$, it holds that:
$\bar t \in \gamma_m(c(\bar v))$,
$input\_size\_set(e_\alpha) \subseteq Dom(\Phi)$, and $r \in \Phi(\bar
n)$, where $\bar n = size_{p}(input_{p}(\bar t))$.  $\Box$
\end{lemma}

Let $c_1(\bar v)$ and $c_2(\bar v)$ be two elements of the mode/type
abstract domain already mentioned, each one representing a set of
calls. 
The inclusion operator $\sqsubseteq_m$
is the order relation in such abstract domain, and meets the condition:
$
\begin{array}{c}
c_1(\bar v) \sqsubseteq_m c_2(\bar v) \text{ \ if and only if \ }
\gamma_m(c_1(\bar v)) \subseteq \gamma_m(c_2(\bar v)) \text{.}
\end{array}
$
In our case, we use the comparison operator $\sqsubseteq_m$
implemented in the \ciaopp system, which uses finer grain comparison
operators for program state properties. In particular, it uses the
type comparison operator of the \emph{eterms} abstract
domain~\cite{eterms-sas02} (based on adaptations of the type inclusion
operations of~\citeS{Dart-Zobel}) and the mode comparison operator of
the \emph{shfr} abstract domain~\cite{ai-jlp} (which
represents \emph{groundness} and \emph{freeness} properties).

\begin{examplebox}
Let $c_1(\bar v)$ be $\texttt{list}(x, \texttt{gnd}) \land
\texttt{var}(y)$, and $c_2(\bar v)$ be $\texttt{list}(x) \land
\texttt{var}(y)$. We have that $c_1(\bar v) \sqsubseteq_m c_2(\bar
v)$, but $c_2(\bar v) \not\sqsubseteq_m c_1(\bar v)$.  Similarly,
$(\texttt{list}(x, \texttt{gnd}) \land \texttt{var}(y)) \sqsubseteq_m
(\texttt{gnd}(x) \land \texttt{var}(y))$,
but $(\texttt{gnd}(x) \land \texttt{var}(y)) \not\sqsubseteq_m
(\texttt{list}(x, \texttt{gnd}) \land \texttt{var}(y))$.
\end{examplebox}

\begin{definition}
Let $\Phi_1$ and $\Phi_2$ be two resource usage interval functions
i.e., $\Phi_{1}: Dom(\Phi_{1}) \mapsto {\cal RI}$, and $\Phi_{2}:
Dom(\Phi_{2}) \mapsto {\cal RI}$, where $Dom(\Phi_{1}) \subseteq
\rsd^{k}$ and $Dom(\Phi_{2}) \subseteq \rsd^{k}$.  Let $S$ be a set
such that $S \subseteq Dom(\Phi_{1})$ and $S \subseteq Dom(\Phi_{2})$.
We define the inclusion relation $\sqsubseteq_S$ and the intersection
operation $\sqcap_S$ as follows:

  \begin{itemize}
   \item $\Phi_1 \sqsubseteq_S \Phi_2$ if and only if for all $\bar n
     \in S$,
$\Phi_1(\bar n) \subseteq \Phi_2(\bar
  n)$. 

  \item We say that $\Phi_1 \sqcap_S \Phi_2 = \Phi_3$ if and only if for all $\bar
    n \in S$,
$\Phi_1(\bar n) \cap \Phi_2(\bar n) = \Phi_3(\bar n)$.

\end{itemize}

\end{definition}

\begin{definition}
\label{def:inclusion}

Let $e_I$ be a pair $(p(\bar v):c_I(\bar v), \Phi_I)$ in the
intended
meaning $\Inten_\alpha$, and $e_\alpha$ the pair $(p(\bar v):c(\bar
v), \Phi)$ in the computed abstract semantics $\overapprox$. For
simplicity, we assume the same tuple of variables $\bar v$ in all
abstract objects.
We say that $e_\alpha \sqsubseteq e_I$
iff $c_I(\bar v) \sqsubseteq_m c(\bar v)$ and $\Phi \sqsubseteq_S
\Phi_I$, where $S = input\_size\_set(e_I)$.  $\Box$
\end{definition}

Note that the condition $c_I(\bar v) \sqsubseteq_m c(\bar v)$ is
needed to ensure that we select resource analysis information that can
safely be used to verify the assertion corresponding to the pair
$(p(\bar v):c_I(\bar v), \Phi_I)$.  If $c_I(\bar v) \sqsubseteq_m
c(\bar v)$, then $input\_size\_set(e_I) \ \subseteq
\ input\_size\_set(e_\alpha)$.

\begin{definition}
\label{def:intersection}
  We say that $(p(\bar v):c(\bar v), \Phi) \sqcap (p(\bar v):c_I(\bar
  v), \Phi_I) = \emptyset$ if: \\ 
\centerline{$c_I(\bar v) \sqsubseteq_m c(\bar v)$
  and $\Phi \sqcap_S \Phi_I = \Phi_{\emptyset}$,} 
\noindent where
  $\Phi_{\emptyset}$ represents the constant function identical to
  the empty interval.
\end{definition}

\begin{theorem}
\label{lem:partial-correctness}
Let $e_\alpha = (p(\bar v):c(\bar v), \Phi)$ and $\overapprox = \{
e_\alpha \}$. Let $e_I = (p(\bar v):c_I(\bar v), \Phi_I)$. If
$e_\alpha \sqsubseteq e_I$ then $\prog$ is partially correct with
respect to $e_I$.
\end{theorem}
\begin{proof}
If $e_\alpha \sqsubseteq e_I$
then $c_I(\bar v) \sqsubseteq_m c(\bar v)$ (by
Definition~\ref{def:inclusion}), what implies that $\gamma_m(c_I(\bar
v)) \subseteq \gamma_m(c(\bar v))$ and hence $input\_size\_set(e_I)
\ \subseteq \ input\_size\_set(e_\alpha)$. We are going to prove that
the condition of Definition~\ref{def:correct-pair} holds.  For all
$(p(\bar t), r) \in \p$, it holds that: if $\bar t \in
\gamma_m(c_I(\bar v))$ then $\bar t \in \gamma_m(c(\bar v))$ (because
$\gamma_m(c_I(\bar v)) \subseteq \gamma_m(c(\bar v))$),
and thus $r \in \Phi(\bar n)$, where $\bar n = size_{p}(input_{p}(\bar
t))$
(by Lemma~\ref{lem:analysis-safety}). Since $\Phi \sqsubseteq_S
\Phi_I$, where $S = input\_size\_set(e_I)$
(Definition~\ref{def:inclusion}), and $input\_size\_set(e_I)
\ \subseteq \ input\_size\_set(e_\alpha)$, we have that $r \in
\Phi_I(\bar n)$.
\end{proof}

Similarly, we have the following result:

\begin{theorem}
\label{lem:incorrectness}
  If $(p(\bar v):c(\bar v), \Phi) \sqcap (p(\bar v):c_I(\bar
  v), \Phi_I) = \emptyset$ and $(p(\bar v):c(\bar v), \Phi) \neq \emptyset$ then $\prog$ is 
not partially correct
w.r.t. $(p(\bar v):c_I(\bar v), \Phi_I)$. $\Box$
\end{theorem}

In order to prove or disprove program partial correctness 
we compare
call-resource pairs by using
Theorems~\ref{lem:partial-correctness} and~\ref{lem:incorrectness}
(thus ensuring the sufficient conditions given in
Table~\ref{tab:app-corr-comp}).
This means that whenever $c_I(\bar v) \sqsubseteq_m c(\bar v)$ we have
to determine whether $\Phi \sqsubseteq_S \Phi_I$ or $\Phi \sqcap_S
\Phi_I = \Phi_{\emptyset}$. To do this in practice, we compare
resource usage bound functions in the way expressed by
the following Corollary~\ref{func-comparison-corollary} of
Theorems~\ref{lem:partial-correctness} and~\ref{lem:incorrectness}.

\begin{corollary}
\label{func-comparison-corollary}
Let $(p(\bar v):c_I(\bar v), \Phi_I)$ be a pair in the intended
abstract semantics $\Inten_\alpha$ (given in a specification), and
$\overapprox = \{(p(\bar v):c(\bar v), \Phi)\}$ the abstract semantics
inferred by analysis.
Let $S$ be the input-size set of $(p(\bar v):c_I(\bar v), \Phi_I)$.
Assume that $c_I(\bar v) \sqsubseteq_m c(\bar v)$.  Then, we have
that:

\begin{enumerate} 

\item \label{suffcond1}
If $\forall \bar n \in S: (\Phi_{I}^{l}(\bar n) \leq
  \Phi^{l}(\bar n) \land \Phi^{u}(\bar n) \leq \Phi_{I}^{u}(\bar n))$,
  then $\prog$ is partially correct with respect to $(p(\bar v):c_I(\bar
  v), \Phi_I)$.

\item \label{suffcond2} If $\forall \bar n \in S: (\Phi^{u}(\bar n) <
  \Phi_{I}^{l}(\bar n) \lor \Phi_{I}^{u}(\bar n) < \Phi^{l}(\bar n))$,
  then $\prog$ is not partially correct
with respect to $(p(\bar v):c_I(\bar v),
  \Phi_I)$.

\end{enumerate}

\end{corollary}

Note that the sufficient condition~\ref{suffcond1}
(resp.,~\ref{suffcond2}) above implies that $\Phi \sqsubseteq_S
\Phi_I$ (resp.\ $\Phi \sqcap_S \Phi_I = \Phi_{\emptyset}$, where, as
already said, $\Phi_{\emptyset}$ represents the constant function
identical to the empty interval.
In practice, we also use the condition $(\forall \bar n \in S:
\Phi^{u}(\bar n) < \Phi_{I}^{l}(\bar n)) \lor (\forall \bar n \in S:
\Phi_{I}^{u}(\bar n) < \Phi^{l}(\bar n))$, although it is stronger
than condition~\ref{suffcond2}.
When $\Phi_{I}^{u}$ (resp., $\Phi_{I}^{l}$) is not present in a
specification, we assume that $\forall \bar n$ $(\Phi_{I}^{u}(\bar n)=
\infty)$ (resp., $\Phi_{I}^{l}=-\infty$ or $\Phi_{I}^{l}(\bar n) = 0$,
depending on the resource). With this assumption, one of the resource
usage bound function comparisons in the sufficient
condition~\ref{suffcond1} (resp.,~\ref{suffcond2}) above is always
true (resp., false) and the truth value of such conditions depends on 
the other comparison.

\paragraph{\textbf{Inferring Preconditions on Data Sizes for Different Verification Outcomes.}}

If none of the conditions~\ref{suffcond1} or~\ref{suffcond2} in
Corollary~\ref{func-comparison-corollary} hold for the input-size set
$S$ of the pair $(p(\bar v):c_I(\bar v), \Phi_I)$, our proposal is to
partition $S$ in a number of $n_S$ subsets $S_j$, $1 \leq j \leq n_S$,
for which either condition holds. Thus, as a result of the
verification of $(p(\bar v):c_I(\bar v), \Phi_I)$ we produce a set of
pairs
$(p(\bar v):c_{I}^j(\bar v), \Phi_I)$, $1 \leq j \leq n_S$, whose
input-size set is $S_j$. Such pairs will be represented as assertions
in the output of our implementation prototype.

For the particular case where resource usage bound functions depend on
one argument, the element $c_{I}^j(\bar v)$ (in the assertion
precondition) is of the form $c_I(\bar v) \con d_j$, where $d_j$
defines an interval for the input data size $n$ to $p$.  This allows
us to give intervals $d_j$ of input data sizes for which a program
$\prog$ is (or is not) partially correct.

The definition of \emph{input-size set} can be extended to deal with
data size intervals $d_j$'s in a straightforward way: \\
\centerline{$S_j = \{n \ | \ \exists \ \bar t \in \gamma_m(c(\bar v))
  \con n = size_{p}(input_{p}(\bar t)) \con n \in d_j\}$.}

From the practical point of view, in order to represent properties
like $n \in d_j$, we have
added to the 
\ial
a new \texttt{intervals(A, B)} property, which expresses that
the input data size \texttt{A} belongs to
some of the intervals in
the list \texttt{B}. To this end, in order to show the result of the
assertion checking process to the user, we group all the $(p(\bar
v):c_{I}^j(\bar v), \Phi_I)$ pairs that meet the above sufficient
condition~\ref{suffcond1} (applied to the set $S_j$) and, assuming
that $d_{f_1}, \ldots, d_{f_b}$ are the computed input data size
intervals for such pairs, an assertion with the following syntactic
schema is produced as output:

\begin{center}
\fbox{\kbd{:- \textcolor{checked}{checked} pred} $p(\bar
  v):c_{I}^j(\bar v)$,\kbd{intervals(}$size_{p}(input_{p}(\bar
  v)),$\kbd{[}$d_{f_1}, \ldots, d_{f_b}$\kbd{])} \kbd{+} $\Phi_I$
  \kbd{.}}
\end{center}

\noindent 
Similarly, the pairs meeting the sufficient condition~\ref{suffcond2}
are grouped and the following assertion is produced:

\begin{center}
\fbox{\kbd{:- \textcolor{false}{false} pred} $p(\bar v):c_{I}^j(\bar v)$,\kbd{intervals(}$size_{p}(input_{p}(\bar v)),$\kbd{[}$d_{g_1}, \ldots, d_{g_e}$\kbd{])}
         \kbd{+} $\Phi_I$ \kbd{.}}
\end{center}

Finally, if there are intervals complementary to the previous ones
w.r.t. $S$ (the input-size set of the original assertion), say
$d_{h_1}, \ldots, d_{h_q}$, the following assertion is produced: 

\begin{center}
\fbox{\kbd{:- \textcolor{check}{check} pred} $p(\bar v):c_{I}^j(\bar
  v)$,\kbd{intervals(}$size_{p}(input_{p}(\bar v)),$ \kbd{[}$d_{h_1},
    \ldots, d_{h_q}$\kbd{])} \kbd{+} $\Phi_I$ \kbd{.}}
\end{center}

The description of how the input data size intervals $d_j$'s are
computed is given in Sect.~\ref{sec:function-comparison}.  

\paragraph{\textbf{Dealing with Preconditions Expressing Input Data Size Intervals.}}

So far, we have seen that a call-resource pair in the 
intended semantics 
$I_\alpha$ has the form $(p(\bar v):c_I(\bar v),
\Phi_I)$, where $c_I(\bar v)$ is a conjunction of type and mode
properties that is used to represent a set of calling data to $p$.
In order to allow checking assertions which include preconditions
expressing intervals within which the input data size of a program is
supposed to lie (i.e., using the \texttt{intervals(A, B)} property),
we also allow adding conjuncts to $c_I(\bar v)$ that are constraints
over the sizes of the data represented by $c_I(\bar v)$. Such
constraints can represent intervals for such data sizes. Accordingly,
we replace the concretization function $\gamma_m$ by an extended
version $\gamma'_m$.
To this end, given an abstract call-resource pair: $(p(\bar
v):c_I(\bar v) \con d, \Phi_I)$,
\noindent
where $d$ represents an interval, or the union of several intervals,
for the input data sizes to $p$, we define:
\\ \centerline{$\gamma'_m(c_I(\bar v) \con d) = \{\bar t \ | \ \bar t
  \in \gamma_m(c_I(\bar v)) \con size_{p}(input_{p}(\bar t)) \in
  d\}$.} We also extend the definition of the $\sqsubseteq_m$ relation
accordingly.  With these extended operations, all the previous results
in Sect.~\ref{sec:resources-framework} are applicable.

In the case where there are multi-variable resource usage bound
functions, instead of intervals represented as pairs of numbers, we
use arithmetic expressions that represent more general \emph{size
  constraints} (see Sect.~\ref{sec:multi-var-function-comparison}),
usually inequalities. In this case, the interval $d$ above will be replaced by the
set of values that satisfy such size constraints.

\secbeg
\section{Resource Usage Bound Function Comparison}
\label{sec:function-comparison}
\secend

Fundamental to our approach to verification are the operations that
compare two cost bound functions. In particular, sufficient
conditions~\ref{suffcond1} and~\ref{suffcond2} of
Corollary~\ref{func-comparison-corollary} for proving and disproving program
correctness and incorrectness respectively, involve comparisons of a
cost bound function inferred by the static analysis with another given
in a specification as an assertion present in the program.

Since our resource analysis is able to infer different types of
functions (e.g., polynomial, exponential, summation, logarithmic,
factorial, etc.),
it is also desirable to be able to compare as many classes as possible
of these functions.

Assume that we have to compare two cost functions $f(\vecx)$ and
$g(\vecx)$ that depend on input data sizes $\vecx \in S$ for a given
input-size set $S$.
Also, given a function $f(\vecx)$, let $f^{l}(\vecx)$ and
$f^{u}(\vecx)$ denote a lower and an upper bound on $f(\vecx)$
respectively, i.e., $\forall \vecx \in S: f^{l}(\vecx) \leq f(\vecx)$
and $\forall \vecx \in S: f(\vecx) \leq f^{u}(\vecx)$.  In the cases
in which the techniques we will describe in the following sections
cannot be applied to give sound results for a given comparison, say
$\forall \vecx \in S: f(\vecx) \le g(\vecx)$, then we replace any of
the functions by an upper or lower bound on it,
in a way that 
ensures obtaining sufficient conditions for such comparison.
This is expressed by the following lemma.

\begin{lemma}
\label{lemma:basic-safety-comparison}
Let be $f(\vecx)$ and $g(\vecx)$ be cost functions and $S$ an input-size set. Then 
\begin{enumerate} 
\item if any of the conditions:

$
\begin{array}{l}
\forall \vecx \in S: f^{u}(\vecx) \le g^{l}(\vecx) \text{,} \\
\forall \vecx \in S: f^{u}(\vecx) \le g(\vecx) \text{, or} \\ 
\forall \vecx \in S: f(\vecx) \le g^{l}(\vecx)
\end{array}
$

\noindent
holds, then $\forall \vecx \in S: f(\vecx) \le g(\vecx)$ holds; and

\item if $\forall \vecx \in S:f^{u}(\vecx) \neq f(\vecx)$ and $\forall \vecx \in S: g^{l}(\vecx) \neq
  g(\vecx)$, then any of the conditions above is also a sufficient
  condition for $\forall \vecx \in S: f(\vecx) < g(\vecx)$.
\end{enumerate}
\end{lemma}

\subsection{Single-Variable Cost Function Comparison}
\label{sec:single-var-fun-comparison}

We define two operations for comparing cost functions, namely
$\ltintervals$ and $\leqintervals$.  The definition of
$\ltintervals(\Psi_{1}, \Psi_{2}, S)$ is described in
Fig.~\ref{func:less-than-intervals} as a function. Function
$\leqintervals$ is similar to $\ltintervals$, 
but it uses the condition $\Psi_{1}(n) \leq \Psi_{2}(n)$, which
implies that there are endpoints of the intervals in
Step~\ref{step:build-intervals-from-roots} that are closed. As already
said, $S$ is a subset of natural numbers, $S \subseteq \sizdomsingle$,
and usually $S = \sizdomsingle$, which is extracted from the
specification, taking into account its precondition. In general, $S$
is given as a union of intervals of natural numbers. However, the cost
bound functions $\Psi_{1}$ and $\Psi_{2}$ are continuous functions
defined over a subset of real numbers, i.e., $Dom(\Psi_i) \subseteq
\rsd$ and $S \subset Dom(\Psi_i)$ for $i=1,2$. Thus, for simplicity,
in the definition of $\ltintervals$ and $\leqintervals$, we first
infer intervals of real numbers (see
Steps~\ref{step:substract-functions}-\ref{step:filter-intervals}
of Fig.~\ref{func:less-than-intervals}), and, from them, we produce
the intervals of natural numbers with the appropriate endpoints, as
described in
Steps~\ref{step:round-to-nat-intervals}-\ref{step:intersect-with-precondition}.
Note that in Step~\ref{step:find-roots}
we ignore the negative roots of $f(x)$ because they cannot be
endpoints of any interval of natural numbers.  Since $\Psi_{1}$ and
$\Psi_{2}$ are continuous, in Step~\ref{step:filter-intervals} we have
that $\forall (a, b) \in \natintervalset_2: (\forall x \in (a, b):
(\Psi_{1}(x) < \Psi_{2}(x)))$. Then, in
Step~\ref{step:round-to-nat-intervals} we generate intervals of
natural numbers, and it holds that for any interval of real numbers
$(a, b) \in \natintervalset_2$, we have that $(\left\lceil a
\right\rceil, \left\lfloor b \right\rfloor)$ is the largest interval
of natural numbers included in $(a, b)$, and hence it holds that
$\forall n \in [\left\lceil a \right\rceil, \left\lfloor b
  \right\rfloor] : (\Psi_{1}(n) < \Psi_{2}(n))$.

\begin{figure}
\figrule
\begin{tabbing}
123\=\kill
$\ltintervals(\Psi_{1}, \Psi_{2}, S)$ \\
\> Takes two single-variable cost bound functions, $\Psi_{1}$ and
  $\Psi_{2}$, and an input-size set $S$,
$S \subseteq \sizdomsingle$. \\
\> Returns a set $\natintervalset$ of intervals such that
$\forall \natinterval \in \natintervalset : (\forall n \in \natinterval: (\Psi_{1}(n) < \Psi_{2}(n)  \land  n \in S))$.
\end{tabbing}
\begin{enumerate}

\item \label{step:substract-functions} 

Let $f(x) = \Psi_{2}(x)-\Psi_{1}(x)$, and assume that $Dom(f)
\subseteq \rsd$;

\item \label{step:find-roots} 

Let $x_1, \ldots, x_m$ be the non-negative real roots of equation
$f(x)=0$, i.e.:

$\forall i (1 \leq i \leq m): (x_i \in \rsd \land x_i \geq 0 \land f(x_i) = 0)$;

\item \label{step:build-intervals-from-roots} 

Let $\realintervalset_1 = \{ [0, x_1), (x_1, x_2), \ldots, (x_{m-1},
  x_m), (x_m, \infty)\}$;

\item 
\label{step:filter-intervals}
Let $\natintervalset_2 = \{ I \mid I \in \realintervalset_1 \land f(v) > 0, \text{
  for an arbitrary value } v \in I \}$;

\item 
\label{step:round-to-nat-intervals}
Let $\natintervalset_3 = \{ [\left\lceil a \right\rceil, \left\lfloor b
\right\rfloor] \mid (a, b) \in \realintervalset_2 \}$;

\item \label{step:intersect-with-precondition}

Let $\natintervalset = \{ I \cap S \mid I \in \realintervalset_3 \}$;

\item return~$\natintervalset$.

\end{enumerate}
\caption{A function for comparing two single-variable cost functions.}  
\label{func:less-than-intervals}
\figrule
\end{figure}

As already explained, given the input-size set $S$ of a call-resource
pair in an 
intended semantics, which can also express data
size intervals in the precondition,
our goal is to partition $S$ in a number of $n_S$ subsets $S_j$ such
that for any $S_j$, $1 \leq j \leq n_S$, either sufficient
condition~\ref{suffcond1} or~\ref{suffcond2} of
Corollary~\ref{func-comparison-corollary} holds.
This can be done by using the comparison operators $\ltintervals$ and
$\leqintervals$ described above, with the appropriate values for
$\Psi_{1}$ and $\Psi_{2}$, and performing intersections or unions of
the resulting intervals, depending on whether the condition is a
conjunction or disjunction respectively.

Consider again Step~\ref{step:find-roots} of
Fig.~\ref{func:less-than-intervals}. If $f(x)$ is a polynomial
function, then there exist efficient algorithms for obtaining its
roots. For the other functions (e.g., exponential, logarithmic or
summation), we have to approximate them using polynomials.  We discuss
this in the following sections, including a detailed description of
the concept of ``safety'' of such approximations in
Section~\ref{sec:safety-approx}.

\subsection{Finding Roots of Polynomial Functions}
\label{sec:comparing-polynomial}

According to the fundamental theorem of algebra, a polynomial equation
of order $m$ has $m$ roots, whether real or complex numbers.
General methods exist that allow computing all these roots, although
in our approach we discard complex roots and negative real roots 
since they are not needed.
All the roots of a polynomial equation can be obtained analytically
until polynomial order four. 
Numerical methods must be used
for polynomial orders greater than four.
In our implementation we have used the GNU Scientific
Library~\cite{gsl} for this purpose. This library offers specific
polynomial function 
root finding methods that are analytical or numerical depending on the
polynomial order, as mentioned above.

\subsection{Finding Roots of Non-Polynomial Functions}
\label{sec:approx-non-polynomial}
Two non-polynomial cost
function classes that the \ciaopp~analyses can infer are exponential
and logarithmic. We approximate exponential functions with Taylor
polynomials
and for
approximating logarithmic functions we replace them with other
functions that bound them from above or below.  After finding the
roots of the approximant polynomials by using the method described
above,
we apply a post-process for checking whether the original functions
have additional roots, which is described in
Sect.~\ref{sec:correct-non-polynomial-roots}.

\paragraph{\bf Exponential function approximation using polynomials} 

This approximation is carried out using these formulae:
\[e^x = \Sigma_{n=0}^\infty \frac{x^n}{n!}=1+x+\frac{x^2}{2!}+
\frac{x^3}{3!}+\ldots \qquad for\ all\ x\]
\[a^x=e^{x\ ln\ a} = 1+x\ ln\ a+\frac{(x\ ln\ a)^2}{2!}+ 
\frac{(x\ ln\ a)^3}{3!}+\ldots \qquad for\ all\ x\]

\noindent
Our experiments show that in practice these series can typically be
limited to order 8, since higher orders do not bring significant
differences.
Also, in the implementation, the computation of the factorials
is done separately and the results are memoized
in order to reuse them.

\begin{small}
\begin{figure}[t]
\centering
\prettylstciao
\begin{lstlisting}
hanoi(N,A,_B,C) :- N=1, print(A,C).
hanoi(N,A,B,C) :-
   N > 1,
   N1 is N - 1,
   hanoi(N1,A,C,B),
   print(A,C),
   hanoi(N1,B,A,C).
\end{lstlisting}
\vspace*{-4mm}
  \caption{A ``Towers of Hanoi'' program.}
  \label{fig:hanoi}
\end{figure}
\end{small}

\begin{examplebox}
\label{ex:hanoi}
Consider the program in Fig.~\ref{fig:hanoi} which prints the
shortest sequence of moves
to solve the ``Towers of Hanoi'' problem with $N$ disks. The first
argument of \texttt{hanoi/4} represents the number of disks to move,
while the remaining ones represent the
peg where the disks are, the auxiliary peg and the target peg, in that
order.

Consider the following assertion:

\begin{center}
\begin{small}
\begin{tabbing}
\texttt{:- \textcolor{check}{check}} \= \texttt{hanoi(N,\_,\_,\_)}  \\
  \> \texttt{: intervals(\nmetric(N),[i(1,inf)])} \\
  \> \texttt{+} \texttt{costb(steps,2**(\nmetric(N)-3) + 2, 2**(\nmetric(N)-3) + 30).}
\end{tabbing}
\end{small}
\end{center}

\noindent
which expresses that for any call to \texttt{hanoi(N,T1,T2,T3)}, a
lower (resp.\ upper) bound on the number of resolution \texttt{steps}
performed by the computation is $2^{n-3} + 2$ (resp. $2^{n-3} + 30$),
where $n = \mathtt{\nmetric(N)}$.

The analysis infers $2^{n + 1} - 2$ as both upper and
lower bound cost function for $n \geq 1$. The output of the assertion checking
considering this result is (see Fig.~\ref{fig:hanoi-plot}):
\begin{tabbing}
\texttt{:- \textcolor{false}{false}} \= \texttt{pred hanoi(N,\_,\_,\_)} \\
     \>  \texttt{: intervals(\nmetric(N),[i(1,1),i(5,inf)])} \\
     \>  \texttt{+ } \= \texttt{costb(steps,2**(\nmetric(N)-3) + 2, 2**(\nmetric(N)-3) + 30).}\\ 
\ \\
\texttt{:- \textcolor{checked}{checked} pred hanoi(N,\_,\_,\_)} \\
     \>  \texttt{: intervals(\nmetric(N),[i(2,4)])} \\
     \>  \texttt{+ } \= \texttt{costb(steps,2**(\nmetric(N)-3) + 2, 2**(\nmetric(N)-3) + 30).}\\ 
\end{tabbing}
\noindent
which express that for $n \in [2,4]$, the specification given by the
assertion is met, while for $n \in [1,1] \cup [5,\infty]$ it is never
met.  The real interval verifying $2^{n-3} + 2 \leq 2^{n+1} - 2
\leq 2^{n-3} + 30$ is approximately $[1.09311,4.09311]$, and the
largest 
interval of natural numbers included in it, and in the interval
expressed in the precondition of the specification, is $[\left\lceil
  1.09311 \right\rceil,\left\lfloor 4.09311 \right\rfloor] = [2,4]$.
Therefore the result obtained from the comparison is exact, in the
context of the specification and the $\sizdomsingle$ domain.
\begin{figure}[t]
\begin{tikzpicture}[scale=0.9]
  \begin{axis}[
    axis x line*=bottom,
    axis y line*=left,
    x axis line style=->,
    y axis line style=->,
    no markers,
    legend style={legend pos=north west,font=\small},
    grid style={line width=.1pt, draw=gray!10},
    major grid style={line width=.1pt,draw=gray!50},
    xlabel=$\mathtt{nat(N)}$,
    xmin = 1,
    ymin = 0,
    ymax=60,
    xmax=10,
    samples at = {1,2,3,4,5,6,7,8,9,10},
    ylabel={steps},
    tick label style={font=\small},
    label style={font=\small},
    line width=1.0pt,
    smooth,    
    width=\linewidth,
  ] 
    \addplot [dashdotted,ultra thick] {2^(x-3) + 30}; 
    \addplot [ultra thick] {2^(x+1) - 2}; 
    \addplot [dashed,ultra thick] {2^(x-3) + 2};

    \addplot +[gray,dashed,thin, forget plot] coordinates {(2,0) (2,60)};
    \addplot +[gray,dashed,thin, forget plot] coordinates {(4,0) (4,60)};
    \addplot +[gray,dashed,thin, forget plot] coordinates {(1,0) (1,60)};
    \addplot +[gray,dashed,thin, forget plot] coordinates {(5,0) (5,60)};
    \legend{upper-bound specification,upper/lower bound analysis,lower-bound specification,\texttt{checked} interval, \texttt{false} interval};
  \end{axis}
  \draw [thick,decoration={brace,mirror},decorate] (1.3,-0.5) -- (4,-0.5) node[midway,below,yshift=-.15cm] {\texttt{checked}};
  \draw [thick,decoration={brace,mirror},decorate] (5.3,-0.5) -- (11.9,-0.5) node[midway,below,yshift=-.15cm] {\texttt{false}};
\draw[thick, ->] (0,-0.7) -- (0,-0.4) node[midway,below,yshift=-.1cm] {\texttt{false}};

\end{tikzpicture}


\caption{Resource usage functions for \texttt{hanoi}:
  specification and analysis results.}
\label{fig:hanoi-plot}
\end{figure}
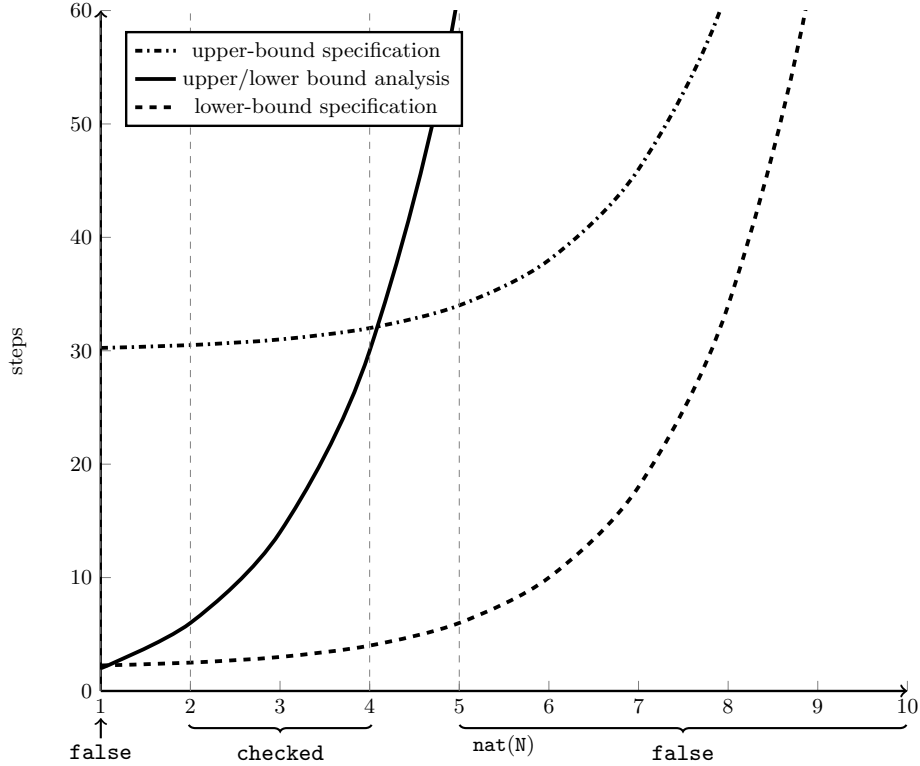
\end{examplebox}

\paragraph{\bf  Logarithmic function approximation}
Assume that we have to perform the comparison $f(x) \le g(x)$, where
any of the two functions $f$ or $g$ is logarithmic.  In this case, by
Lemma~\ref{lemma:basic-safety-comparison}, we can replace such
functions by upper or lower bounds on them, depending on the case, to
obtain sufficient conditions. For example, given the logarithmic
function $log(h(x))$, our approach will use $h(x)$ as an upper bound
on it.

Thus, $log(h(x)) \le g(x)$ would be replaced by the sufficient
condition $h(x) \le g(x)$.

\begin{figure}[t]
\centering
\prettylstciao
\begin{lstlisting}
simple_log(N, N) :-
   N=<1,!.
simple_log(N, S) :-
   N>1,
   N1 is N/2,
   simple_log(N1,S1),
   S is S1 + N.
\end{lstlisting}
\vspace*{-4mm}
  \caption{A simple example with logarithmic cost.}
  \label{fig:simplelog}
\end{figure}

\begin{examplebox}
\label{ex:simplelog}
Consider the program in Fig.~\ref{fig:simplelog} which calculates
the sum $N + N/2 + N/2^2 + \ldots + 1$, given $N$ as input. Consider
the following assertion:

\begin{center}
\begin{tabbing}
\=\kill
\texttt{:- \textcolor{check}{check} pred simple\_log(N,\_) + costb(steps, 0, 3000).} 
\end{tabbing}
\end{center}

\noindent
in order to find intervals of possible sizes of $N$ for which the
number of resolution steps of any call to \texttt{simple\_log(N,\_)}
will be less or equal than $3000$. Let $n = \mathtt{\nmetric(N)}$, the
analysis infers that the cost of a call to this predicate will be
upper/lower bounded by $log_2(\frac{n}{8}) + 4$. With this
information, the assertion checking process returns the following two
assertions:
\vspace*{-2mm}
\begin{tabbing}
\texttt{:- \textcolor{check}{check}} \= \texttt{pred simple\_log(N,\_)} \\
         \> \texttt{: intervals(\nmetric(N),[i(23969,inf)])} \\
         \> \texttt{+ costb(steps,0,3000).} \\
\ \\ [-2mm]
\texttt{:- \textcolor{checked}{checked} pred simple\_log(N,\_)} \\
         \> \texttt{: intervals(\nmetric(N),[i(0,23968)])} \\
        \> \texttt{+ costb(steps,0,3000).} 
\end{tabbing}

\noindent
which expresses that for $n \leq 23968$ the specification given by
the assertion is met, while for $n > 23969$ the assertion cannot
be proved nor disproved. This result is correct but obviously it is an
approximation. 
\end{examplebox}

\secbeg
\subsection{Checking Additional Roots for Non-polynomial Functions}
\label{sec:correct-non-polynomial-roots}

In this section we describe a post-process that ensures the
correctness of the function comparison approach that we have presented
so far, for the cases in which there are functions that have been
approximated by polynomials, e.g., exponential functions, for which
generally the number of roots is unknown.

Consider the comparison operator $\ltintervals$ described in
Sect.~\ref{sec:single-var-fun-comparison}, in particular
Step~\ref{step:substract-functions} of
Fig.~\ref{func:less-than-intervals} where we define $f(x) =
\Psi_{2}(x)-\Psi_{1}(x)$.
Assume that we approximate $f(x)$ by a polynomial $P(x)$ and find the
non-negative real roots of $P(x)$, say $x'_1, \ldots, x'_k$.  Then
$x'_1, \ldots, x'_k$ might not include all the non-negative real roots
of $f(x)$, denoted $x_1, \ldots, x_m$ in Step~\ref{step:find-roots}.

To ensure that there is no other root of $f(x)$ inside any of the
computed
intervals for $P(x)$, i.e., $[0, x'_1), (x'_1, x'_2), \ldots,
  (x'_{k-1}, x'_k), (x'_k, \infty)$, we proceed as follows.  We first
  consider all the intervals but the last one $(x'_k, \infty)$, i.e.,
  let $\realintervalset' = \{ [0, x'_1), (x'_1, x'_2), \ldots,
    (x'_{k-1}, x'_k)\}$, and $\realintervalset'' = \{ [\left\lceil a
      \right\rceil, \left\lfloor b \right\rfloor] \mid (a, b) \in
    \realintervalset'\}$. First, we check that:
$$\forall \realinterval \in \realintervalset'': (\forall n \in
    \realinterval: f(n) > 0)$$

\noindent
by enumerating the finite number of values, i.e., natural numbers, in
each interval $\realinterval$. It is always possible of course to give
up and return \emph{unknown} if this number is above a certain
threshold, or use the procedure below.

However, in the last interval $(x'_k, \infty)$ we obviously have to
use a different procedure to ensure whether a function is indeed
\emph{always} bigger than the other. Our procedure uses a set of
syntactic rules to compare the two functions $\Psi_{1}(x)$ and
$\Psi_{2}(x)$ together with a constraint $x > x'_k$, which expresses
that the comparison only holds from the largest root to infinity.
More specifically, we have implemented a modification of the
comparison algorithm in~\cite{AlbertAGHP09,Albert2015-cost-func-comp}.
Note that we only use such comparison algorithm
for this very particular case, since it can be given constraints of
the form $x > c$, where $c$ is a constant, which represents the
interval $(c, \infty)$ in our approach. If such comparison returns
\emph{true}, then it is ensured that one of the functions to compare
is greater than the other, in the context of the given constraints;
otherwise, nothing can be ensured. Thus, such a comparison is
complementary to ours for this particular case, i.e., checking the last
interval already computed by our approach, when non-polynomial
functions are approximated by polynomials.
However, we do not use it for anything else, since, among other
things, it cannot infer preconditions involving intervals for which one
function is greater or smaller than the other, as our approach does.

In addition, we also use the derivatives of the functions, which tend
to be simpler and easier to verify. In particular, we exploit the fact
that if $\Psi_{1}(x) < \Psi_{2}(x)$ on $x=a$, then such functions will
never intersect for all $x > a$ as long as their derivatives
satisfy $\Psi'_{1}(x) < \Psi'_{2}(x)$ for all $x > a$.  

Although our algorithm is not complete, it is correct in the sense
that when checking $\Psi_{1}(x) < \Psi_{2}(x)$, if the algorithm
returns \emph{true}, then for some $x'_k$, such inequality holds for
all $x \in (x'_k, \infty)$. If this cannot be ensured by our
algorithm, then the algorithm returns \emph{unknown}.

\subsection{Safety of the Approximation}
\label{sec:safety-approx}

The roots obtained for function comparison are in some cases
approximations of the actual
roots. The errors in approximations come from two sources: a) the
numerical method for root calculation of polynomials, and b) the
difference between the original non-polynomial function and its
polynomial approximant.
In any case, we must guarantee that their values are safe, in the
sense that they can be used for verification purposes,
in particular for proving sufficient conditions~\ref{suffcond1}
and~\ref{suffcond2} in Corollary~\ref{func-comparison-corollary}.  In
turn, such conditions depend on the comparison operators
$\ltintervals$ and $\leqintervals$ already described.  To this end, the concept of
\emph{safety} of the roots is meaningful in the context of a given
comparison operator. Consider for example operator $\ltintervals$, and
Steps~\ref{step:substract-functions}-\ref{step:find-roots} of its
definition in Fig.~\ref{func:less-than-intervals}, assuming that $x_1,
\ldots, x_m$ are exact roots of function $f(x)$.

\begin{definition}
\label{def:safe-root-approx}
Let $f(x)$ be a continuous function such that $Dom(f) \subseteq \rsd$,
and let $X =\{x_1, \ldots, x_m\}$ be the set of its exact non-negative
real roots. Let $\realintervalset = \{ [0, x_1), (x_1, x_2), \ldots,
  (x_{m-1}, x_m), (x_m, \infty)\}$. Then, for any root $x \in X$ and
  for any interval $I \in \realintervalset$, we say that $x'$ is a
  safe approximation of $x$ for $\realinterval$ if:
$$
((\exists a: (a, x) = \realinterval)  \rightarrow  x' \leq x) \  \land \
((\exists b: (x, b) = \realinterval)  \rightarrow  x \leq x') 
$$
$\Box$

\end{definition}

In the context of this definition, given any interval $\realinterval$
such that $\forall x \in \realinterval: f(x) > 0$, it is clear that if
we replace any endpoint (or both) of $\realinterval$ by safe
approximations for $\realinterval$, obtaining $\realinterval'$, then,
it holds that $\forall x \in \realinterval': f(x) > 0$.

For example, in Step~\ref{step:filter-intervals} of
Fig.~\ref{func:less-than-intervals}, it holds that $\forall
\realinterval \in \natintervalset_2: (\forall x \in \realinterval:
f(x) > 0)$, which implies that $\forall \realinterval \in
\natintervalset_2: (\forall x \in \realinterval: \Psi_{1}(x) <
\Psi_{2}(x))$.  Thus, if we replace the endpoints of the intervals in
$\natintervalset_2$ by safe approximated roots for them,
we can ensure that, if $\natintervalset$ is the result of
$\ltintervals(\Psi_{1}, \Psi_{2}, S)$, then $\forall \natinterval \in
\natintervalset: (\forall n \in \natinterval: (\Psi_{1}(n) <
\Psi_{2}(n)))$.  A similar reasoning can be done for operator
$\leqintervals$.

When we say that we safely check a given condition, we mean that we
possibly use safe approximated roots for building intervals for which
our algorithm says that the condition holds, and thus such intervals
may be smaller than the ones for which the condition actually
holds. In addition, our verification approach works with
approximations of the concrete semantics and safely checks sufficient
conditions to prove or disprove program partial correctness and
incorrectness. This implies that our approach may infer stronger
sufficient conditions.

Assume for example that we want to 
check whether $\forall x \in S: \Phi^{u}(x) \leq \Phi_{I}^{u}(x)$,
where $\Phi^{u}$ and $\Phi_{I}^{u}$ are resource usage bound
functions, the former is part of the result of program analysis and the
latter appears in an assertion declared in the program. This check is part of the
sufficient condition~\ref{suffcond1} in
Corollary~\ref{func-comparison-corollary}.  In this case, we can use
the operator $\leqintervals(\Phi^{u}, \Phi_{I}^{u}, S)$, which defines
$f(x)=\Phi_{I}^{u}(x)-\Phi^{u}(x)$. Assume that $\forall x \in S: f(x)
\geq 0$. Then it holds that $\forall x \in S: \Phi^{u}(x) \leq
\Phi_{I}^{u}(x)$. Since $\leqintervals$ may use safe approximated
roots, it may return a set $S'$ smaller than $S$, i.e., $S' \subset
S$.  Assume also that $\Phi_{I}^{l}$ is not given in the assertion,
meaning that the specification does not state any lower bound for the
resource usage, i.e., the lower endpoint of any resource usage
interval is $0$,
which means that $\forall x \in S: \Phi_{I}^{l}(x) \leq \Phi^{l}(x)$ is 
true. Thus, if $\forall x \in S: f(x) \geq 0$ we can 
state that sufficient condition~\ref{suffcond1} of
Corollary~\ref{func-comparison-corollary} holds.
Similarly, assume that we use
$\ltintervals(\Phi_{I}^{u}, \Phi^{l}, S)$, which defines
$f(x)=\Phi^{l}(x)-\Phi_{I}^{u}(x)$.
Then we can 
say that $\forall x \in S: \Phi_{I}^{u}(x) < \Phi^{l}(x)$ if $\forall
x \in S: f(x)>0$, proving that sufficient condition~\ref{suffcond2} of
Corollary~\ref{func-comparison-corollary} holds.
We can reason
similarly in the comparisons involving a lower bound in the assertion,
i.e., $\Phi_{I}^{l}$.  Thus, we focus exclusively on
checking that $\forall x \in S: f(x)>0$ or $\forall x \in S: f(x) \geq
0$, where $f(x)$ is conveniently defined in each case.

We now focus on a method we propose for obtaining safe approximated
roots.  Assume that the exact roots of function $f(x)$ are
$x_1,...,x_m$, and that $x'_1,...,x'_m$ are approximated roots
obtained by using the techniques already explained, so that for each
approximated root $x'_i$, $1 \leq i \leq m$, there is a value
$\varepsilon$ such that
$x_i \in [x'_i-\varepsilon, x'_i+\varepsilon]$.
Consider an interval $\realinterval$ for which we need to ensure that
$\forall x \in \realinterval: f(x) > 0$. Assume that $\realinterval =
(x'_i, b)$ for some $1 \leq i \leq n$ and some endpoint $b$.  In this
case, the condition for $x'_i$ to be a safe root of $x_i$ for
$\realinterval$ is $x_i \leq x'_i$. Then, we first determine the
actual relative position of $x'_i$ and $x_i$, and, if it is not
compatible with condition $x_i \leq x'_i$, i.e., if $x'_i$ is ``to the
left'' of $x_i$, then we start an iterative process that increments
$x'_i$ by some $0 < \delta < 1$ so that after $m$ iterations we have
that $x''_i = x'_i + m \ \delta$, and $x''_i$ is a safe root of $x_i$
for $\realinterval$.  We can reason similarly for the case in which
$\realinterval = (b, x'_i)$.  In this case, if $x'_i$ is ``to the
right'' of $x_i$, then we start an iterative process that increments
$x'_i$ by some $-1 < \delta < 0$, so that $x''_i$ is a safe root of
$x_i$. This is explained in more detail in the rest of this section.

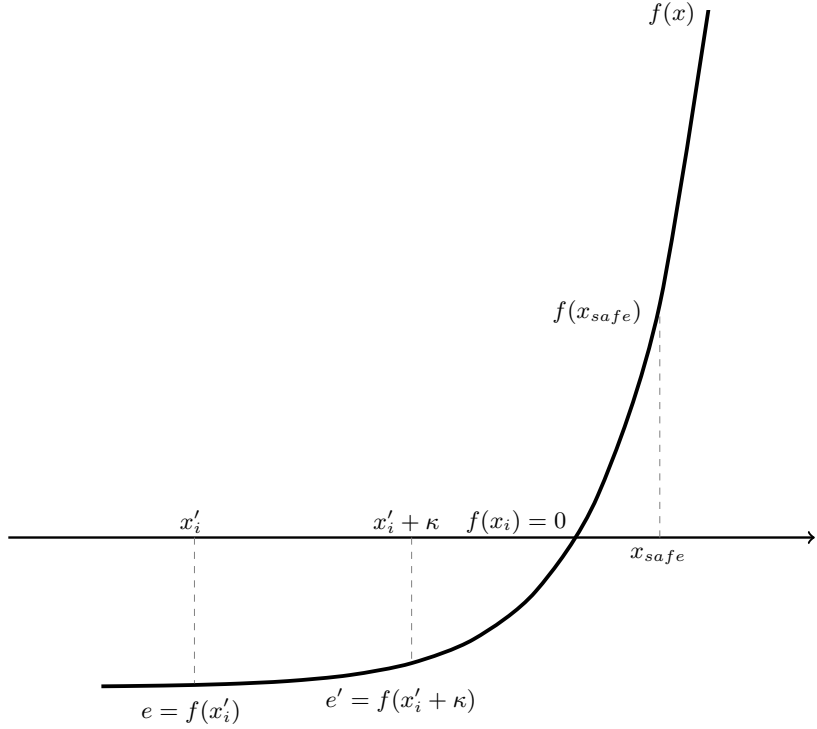
\begin{figure}[t]
\centering
\begin{tikzpicture}[scale=0.9]
  \node[draw=none] at (9.8,10) {$f(x)$};
  \node[draw=none] at (8.7,5.6) {$f(x_{safe})$};
  \node[draw=none] at (9.6,2) {$x_{safe}$};
  \node[draw=none] at (7.5,2.5) {$f(x_i) = 0$};
  \node[draw=none] at (5.88,2.5) {$x'_i + \kappa$};
  \node[draw=none] at (5.8,-0.15) {$e'=f(x'_i + \kappa)$};
  \node[draw=none] at (2.7,2.5) {$x'_i$};
  \node[draw=none] at (2.7,-0.3) {$e=f(x'_i)$};
  \begin{axis}[
    hide y axis,
    axis x line=middle,
    axis y line=middle,
    x axis line style=->,
    y axis line style={draw=none},
    no markers,
    legend style={legend pos=north west,font=\small},
    enlarge x limits=0.15,
    enlarge y limits=0.15,
    every axis x label/.style={at={(current axis.right of origin)},anchor=north west},
    every axis y label/.style={at={(current axis.above origin)},anchor=north east},
    grid style={line width=.1pt, draw=gray!10},
    major grid style={line width=.1pt,draw=gray!50},
    xmin = 0,
    ymin = -10,
    ymax=60,
    xmax=10,
    samples at = {0,1,2,3,4,5,6,7,8,9,10},
    ticks=none,
    tick label style={font=\small},
    label style={font=\small},
    line width=1.0pt,
    smooth,    
    width=\linewidth,
    ]
    
    \addplot+[black,ultra thick] {2^(x)/10 -20}; 

    \addplot +[gray,dashed,thin, forget plot] coordinates {(9,0) (9,30)};
    \addplot +[gray,dashed,thin, forget plot] coordinates {(5,0) (5,-17)};
    \addplot +[gray,dashed,thin, forget plot] coordinates {(1.5,0) (1.5,-19.5)};
  \end{axis}

\end{tikzpicture}

\caption{Case 1. $x_i>x'_i$ (since $e'>e$). $x_{safe}$ is a safe approximated root
  of $x_i$.}
\label{fig:f1}
\end{figure}

\secbeg
\paragraph{\bf Determining the relative position of the exact root}
To determine the relative position of the exact root $x_i$ and its
approximated value $x'_i$ we use the gradient of $f(x)$ around $x=x'_i$.  For
determining the gradient we use the values of $e=f(x'_i)$ and
$e'=f(x'_i+ \kappa)$,
with 
$\kappa > 0$ a relatively small number.  Whether the approximated root
is greater or smaller than the exact root depends on the following
conditions:

\begin{enumerate}
\item if $e<0$ and $e'>e$ then $x_i>x'_i$
\item if $e>0$ and $e'>e$ then $x_i<x'_i$
\item if $e>0$ and $e'<e$ then $x_i>x'_i$
\item if $e<0$ and $e'<e$ then $x_i<x'_i$
\end{enumerate}

From Fig.~\ref{fig:f1} we can see the rationale behind the first case.
If $e'>e$ then $f(x)$ is increasing, but, since $e<0$, then $f(x)>0$
can only occur for values of $x$ greater than $x'_i$.
The other cases follow an analogous reasoning.

\paragraph{\bf Iterative process for computing the safe root}
Once we have determined the relative position of the exact root $x_i$
and its approximated value $x'_i$, we set up an appropriate value for
$\delta$. If we have to ensure that $x_i \leq x'_i$ but it actually
holds that $x_i > x'_i$, then we take $0 < \delta < 1$ so that we
iterate on the addition \(x''_i = x'_i+ \delta\) until
\(f(x''_i)>0\). In this case, the iteration goes to the right.
Such an iteration is apparent in the following pseudo-code:

\begin{algorithmic}[1]
\STATE \assign{$x_{safe}$}{$x'_i$} 
\WHILE{$f(x_{safe})<0$}
\assign{$x_{safe}$}{$x_{safe} + \delta$}
\ENDWHILE
\STATE \return~ $x_{safe}$

\end{algorithmic}

Conversely, if we have to ensure that $x'_i \leq x_i$ but it actually
holds that $x'_i > x_i$, then we take $-1 < \delta < 0$ so that the
iteration goes to the left.

Our approach ensures that there are no other roots of $f(x)$ between
$x'_i$ and $x_{safe}$. As already said, we approximate $f(x)$ by a
polynomial $P(x)$, and the techniques we use can find all the roots of
polynomials. If $f(x)$ is not a polynomial, then $f(x)$ can have more
roots than $P(x)$, but we use the techniques described in
Section~\ref{sec:correct-non-polynomial-roots} to deal with this
possible case and ensure that there are no additional roots inside the
inferred intervals. In addition, as already said, based on the sign of
the gradient, we infer whether $f(x)$ is increasing or decreasing. But
we also check this after computing $x_{safe}$: if the derivative of
$f(x)$ is positive (resp. negative) between $x'_i$ and $x_{safe}$ then
$f(x)$ is increasing (resp. decreasing) between $x'_i$ and $x_{safe}$,
which implies that there are no other roots of $f(x)$ between $x'_i$
and $x_{safe}$.

\begin{examplebox}
  Consider the following assertion 
  for the classical \texttt{fibonacci} program:\\
\noindent 
\texttt{
  :- \textcolor{check}{check} pred fib(N,F) : (\nmetric(N), var(F)) \\
  \hspace*{41mm}     + cost(ub, steps, exp(2, \nmetric(N))-1000 ).
}\\
\noindent 
which expresses that for any call to \texttt{fib(N,F)} with the first
argument bound to 
a natural number
and the second one a free variable, an
expected upper bound on the number of resolution \texttt{steps}
performed by its whole computation is given by the function
$\Phi_{I}^{u}(x) = 2^x-1000$, where $x$ is the size of the first
argument \texttt{N}. Since such argument is 
a natural number the size
metric used for it is its 
value.

The lower bound inferred by the static analysis is $\Phi^l(x) =$ 1.45
$\times$ 1.62$^x - 1$.
The intersection of $\Phi^l(x)$ and  $\Phi_{I}^{u}(x)$ occurs at $x\approx$10.22. 
However, the root obtained by our root finding algorithm is $x\approx$
10.89.
By doing an iterative approximation from 10.89 to the left, we finally
obtain a safe approximate root of $x\approx$10.18.

As already said, and this example illustrates, usually cost functions
depend on variables which range over natural
numbers. For this reason, in this case, we will take the closest
natural number to the left or right of the safely approximated root
computed by the iterative algorithm described above, depending on the
gradient, to obtain a safe value in the domain of the resource usage
function. Thus, in this example, we will take the value 10 for $x$.

It turns out that the analysis also infers the same cost function as
both a lower and upper bound (i.e., it infers the exact
function). Thus, the upper bound cost function is given by $\Phi^u(x)
=$ 1.45 $\times$ 1.62$^x - 1$.

Once the interval endpoints have been computed, we can reason as
follows:  to the left of the safe root $x=10$, the cost upper bound
declared in the specification given by the check assertion 
is less than the (safe) lower bound inferred by the analysis, therefore the
assertion is false in the interval $[0,10]$. Since in this example we
are dealing with exponential functions, we also have to verify every
point in such interval, as already explained in
Sect.~\ref{sec:correct-non-polynomial-roots}. Moreover, to the right
of the safe root $x=10$, the cost upper bound declared in the
specification is greater than the (safe) upper bound inferred by the
analysis, and therefore the assertion is true in the interval
$[11,\infty]$. Our algorithm from
Sect.~\ref{sec:correct-non-polynomial-roots} also verifies that the
functions never intersect in such interval, and thus we can ensure that
the specification is met in it.  Finally, the output of our assertion
checking algorithm for the \texttt{fibonacci} program is:
\begin{tabbing}
\texttt{:- \textcolor{false}{false} pred fib(N,F)} \=  \texttt{: intervals(\nmetric(N), [i(0,10)])} \\
        \> \texttt{+ cost(ub, steps, exp(2,\nmetric(N))-1000 ).}\\
\texttt{:- \textcolor{checked}{checked} pred fib(N,F)} \= \texttt{: intervals(\nmetric(N), [i(11,+inf)])}\\
        \> \texttt{+ cost(ub, steps exp(2,\nmetric(N))-1000 ).}
\end{tabbing}
meaning that the system has proved that the assertion is false for
values of the input argument \texttt{N} in the interval $[0, 10]$, and
true for \texttt{N} in the interval $[11, \infty)$. Thus, showing
  extra conditions (an interval of natural numbers) on which the
  assertion can be proved false, on one hand, and another condition
  (the rest of the range of the natural numbers) on which it can be
  proved true, on the other hand.
\end{examplebox}

\subsection{Comparing Summation Functions} 
\label{sec:comparing-summation}

Dealing with summation functions can be important in the analysis of
recursive programs, and hence of imperative programs that contain
loops. However, the function comparison operation is not
straightforward when at least one of the operands contains a summation
function, even in the case in which other operands are just simple
arithmetic functions.

A summation cost function $C$ is an expression of the form $C(n) =
\summation{i=a}{n}{f(i)}$, where $a, n \in {\mathbb N}$, and $f$ is a
cost function. Our approach consists in transforming it into an
equivalent
closed form
function $C^{t}$, i.e., an expression that does not contain any
subexpressions built by using the $\sum$ operand.
Instead, $C^{t}$ is built
by using only
elementary arithmetic functions, e.g., constants, addition,
subtraction, multiplication, division, exponential, or even factorial
functions.
Such transformation is based on
\emph{\finitecalculus}~\cite{Gleich-2005-finite-calculus}.
The closed form function $C^{t}$ can be a polynomial, but also other
non-polynomial function. Thus, the set of functions that can be
represented as summation expressions is a superset of the functions
that can be represented as polynomials.  Finally, we replace the
summation cost function $C$ by its closed form transformation $C^{t}$,
and use the function comparison techniques explained in the previous
sections.

Prior to explaining our algorithm for obtaining $C^{t}$, we provide some
necessary background.
We start by recalling the relation between 
\infinitecalculus\ 
 and \finitecalculus, 
focusing on the concepts of \emph{derivative} and
\emph{antiderivative} functions.

\paragraph{\bf Relating finite and \infinitecalculus.}

In \infinitecalculus, the \emph{derivative} of a function $f(x)$,
denoted $\frac{\mathrm{d}}{\mathrm{d}x} f(x)$ or $f'(x)$, is defined
as $\frac{\mathrm{d}}{\mathrm{d}x} f(x) =
\lim_{h\to0}\frac{f(x+h)-f(x)}{h}$.  A similar concept is defined in
\finitecalculus for a \emph{discrete} function $f(x)$, the
\emph{discrete derivative}, denoted $\Delta f(x)$, by assuming
discrete increments $h$ for variable $x$. Since the closest we can get
to $0$ is $1$, in the limit, i.e., $h=1$, we obtain the following definition:

\begin{definition}
\label{def:discrete-derivative}
The \emph{discrete derivative} of function $f(x)$ is $\Delta f(x)=
f(x+1)-f(x)$
\end{definition}

In \infinitecalculus, if $\frac{\mathrm{d}}{\mathrm{d}x} F(x) = f(x)$,
then we say that $F(x)$ is an \emph{antiderivative} function of
$f(x)$.  
For any constant $c$, $F(x) + c$, is also an \emph{antiderivative} of
$f(x)$. Since the number of \emph{antiderivatives} of $f(x)$ is
infinite, we denote the class of such antiderivatives $F(x) + c$ as
$\int f(x) \ \mathrm{d}x$, which is also called the \emph{indefinite
  integral} of $f(x)$.  Also, the \emph{definite integral} of $f(x)$
over the interval $[a, b]$ is denoted as $\int_a^b f(x)
\ \mathrm{d}x$. According to the fundamental theorem of calculus, if
$f(x)$ is a real-valued continuous function on $[a,b]$ and $F(x)$ is
an antiderivative of $f(x)$ in $[a,b]$, then $\int_a^b f(x)
\ \mathrm{d}x = F(x)|_a^b = F(b) - F(a)$. Similarly, in
\finitecalculus, if $\Delta F(x) = f(x)$, then $F(x)$ is a
\emph{discrete antiderivative} 
of $f(x)$, and $\sum f(x) \ \mathrm{d}x$ denotes the \emph{discrete
  indefinite integral}
of $f(x)$, i.e., $F(x) + c$, where $c$ is an arbitrary constant. The
following definition allows extending the analogy. \\

\begin{definition}
\label{def:discrete-definite-integral}
The \emph{discrete definite integral} of $f(x)$ over the discrete
interval $[a, b]$, denoted as $\summation{a}{b}{f(x) \ \mathrm{d}x}$,
is defined as:
$$\summation{a}{b}{f(x) \ \mathrm{d}x}= F(x)|_a^b=F(b)-F(a)$$
\end{definition}

\noindent
where $F(x)$ is a \emph{discrete indefinite integral} of $f(x)$, i.e.,
$F(x) = \sum f(x) \ \mathrm{d}x$. Then, we get the following result,
which makes it possible to transform a summation into a \emph{definite
  integral}, and further into a closed form function.

\begin{theorem}
\label{th:fundamental-finite-calculus}
The fundamental theorem of \finitecalculus is:
$$\summation{x=a}{b}{f(x)} = \summation{a}{b+1}{f(x)} \ {\mathrm{d}x}$$
\end{theorem}
\begin{proof}
Let $F(x)$ be a \emph{discrete indefinite integral} 
of $f(x)$, i.e., $\Delta F(x) = f(x)$. According to
Definition~\ref{def:discrete-derivative}, we have that $\Delta F(x) =
F(x+1) - F(x) = f(x)$. Then:

$
\begin{array}{l}  
\summation{x=a}{b}{f(x)} = \summation{x=a}{b}{(F(x+1) - F(x))}  \\
  = F(a+1) - F(a) + F(a+2) - F(a+1) + \cdots + F(b) - F(b-1) + F(b+1) - F(b) \\
  = F(b+1) - F(a) \\
  = \summation{a}{b+1}{f(x) \ \mathrm{d}x} \  \text{(according to Definition~\ref{def:discrete-definite-integral})} 
\end{array}
$
\end{proof}

The \emph{falling power} in \finitecalculus is defined as:

$
    \begin{array}{lll}
      x^{\underline{0}} & = & 1 \\ [-1mm]
      x^{\underline{m}} & = & (x-(m-1)) \  x^{\underline{m - 1}} \ \ \text{ if } m > 0 
    \end{array}
$

Equivalently, if $m>0$, then $x^{\underline{m}} = x \ (x-1) \ (x-2)
\cdots (x-(m-1))$.  For example: $x^{\underline{1}}=x$,
$x^{\underline{2}}=x \ (x-1)$, $x^{\underline{3}}= x \ (x-1) \ (x-2)$,
and so on.

The use of the \emph{falling power} allows to define derivative and
integration rules
in \finitecalculus that are analogous to the corresponding ones in
\infinitecalculus. For example in \infinitecalculus, given
the function $f(x) = x^m$, its derivative is given by
$\frac{\mathrm{d}}{\mathrm{d}x} \ f(x) = m \ x^{m-1}$, and its
\emph{indefinite integral} is $\int f(x) \ \mathrm{d}x = \frac{1}{m+1}
\ x^{m+1} + c$, where $c$ is an arbitrary constant.  The rules for the
\emph{falling power} in \finitecalculus are analogous: given a
discrete function $f(x) = x^{\underline{m}}$, its derivative is given
by $\Delta f(x) = m \ x^{\underline{m-1}}$, and its \emph{discrete
  indefinite integral} is $\sum f(x) \ \mathrm{d}x = \frac{1}{m+1}
\ x^{\underline{m+1}} + c$.

Table~\ref{table:finite-calculus} provides a set of rules for
computing integrals and derivatives in \finitecalculus, including the
ones already seen for the falling power.

We can perform a translation from regular powers into falling powers,
which is needed prior to applying some rules in
Table~\ref{table:finite-calculus}, by using the following theorem:

\begin{equation}
\label{eq:power-to-falling-power}
x^m = \summation{k=0}{m}{\stirling{m}{k} \ x^{\underline{k}}}
\end{equation}

\noindent 
where $\stirling{m}{k}$ is a Stirling number of the second kind, which
represents the number of ways of partitioning $n$ distinct objects
into $k$ non-empty sets~\cite{Gleich-2005-finite-calculus}.
For example: 

\medskip
$
\begin{array}{l}
x^0=x^{\underline{0}} \ \text{since by definition} \ x^0 = 1 \ \text{and} \ x^{\underline{0}} = 1 \text{, but also:} \\
x^0 = \stirling{0}{0} \ x^{\underline{0}} = 1 \ x^{\underline{0}} \\
x^1 =  \stirling{1}{0} \ x^{\underline{0}} + \stirling{1}{1} \ x^{\underline{1}} = 0 \ x^{\underline{0}} + 1 \ x^{\underline{1}} = x^{\underline{1}}  \\
x^2 = \stirling{2}{0} \ x^{\underline{0}} + \stirling{2}{1} \ x^{\underline{1}} + \stirling{2}{2} \ x^{\underline{2}} = x^{\underline{2}}+x^{\underline{1}} \\
x^3=  \stirling{3}{0} \ x^{\underline{0}} + \stirling{3}{1} \ x^{\underline{1}} + \stirling{3}{2} \ x^{\underline{2}} + \stirling{3}{3} \ x^{\underline{3}} =x^{\underline{3}}+3x^{\underline{2}}+x^{\underline{1}} \\
\end{array}
$
\medskip

Thus, the $\Delta f(x)$ and $\sum f(x) \ \mathrm{d}x$ functions in
\finitecalculus are analogous to the \emph{derivative}
($\frac{\mathrm{d}}{\mathrm{d}x} f(x)$) and \emph{antiderivative}
($\int f(x) \ \mathrm{d}x$) functions in 
\infinitecalculus respectively.  Note also, that the integer number
$2$ in \finitecalculus is analogous to Euler's number $e$ in 
\infinitecalculus, in the sense that $\Delta 2^x = 2^x$ and
$\frac{\mathrm{d}}{\mathrm{d}x} e^x = e^x$, as well as $\sum 2^x
\ \mathrm{d}x = 2^x + c$ and $\int e^x \ \mathrm{d}x = e^x + c$.

\paragraph{\bf Our algorithm for rewriting summations.}

Based on Theorem~\ref{th:fundamental-finite-calculus} and
Definition~\ref{def:discrete-definite-integral}, given a summation of
the form $\summation{x=a}{b}{f(x)}$, where $a, b \in {\mathbb N}$, we
rewrite it as a \emph{definite integral} in \finitecalculus:\footnote{For simplicity of exposition we assume that $a, b
  \in {\mathbb N}$, but our algorithm can be also applied even when $a$
  and $b$ are arithmetic expressions, i.e., functions $a, b: {\mathbb
    N} \rightarrow {\mathbb N}$.}

\begin{equation}
\label{eq:sum2integral}
\summation{x=a}{b}{f(x)} = \summation{a}{b+1}{f(x)}
\ {\mathrm{d}x} = F(b+1) - F(a)
\end{equation}

\noindent where $F(x)$ is the \emph{indefinite integral} function of
$f(x)$, i.e., $F(x) = \sum f(x) \ \mathrm{d}x$, and is obtained by
using the integration rules provided in the fourth column of
Table~\ref{table:finite-calculus} for different classes of functions
$f(x)$, specified in the second column of the table. The third column
of the table shows some rules for obtaining the derivatives of the
functions in the second column, which are needed for the application
of the integration rule 8 provided in the fourth column, row 8.

\begin{table}
\centering
\begin{tabular}{l|l|l|l}\hline \hline
\#Rule & $f(x)$ & $\Delta f(x)$ & $\Sigma f(x) \mathrm{d}x$ \\ \hline
1 & $x^{\underline{m}}$ & $m \ x^{\underline{m-1}}$  & 
$\frac{1}{m+1} \ x^{\underline{m+1}} + c$\\
2 & $2^x$ & $2^x$ & 
$2^x + c$ \\
3 & $a^x$ & $(a-1) \ a^x$ & 
$\frac{1}{a-1} \ a^x + c$ \\
4 & $a^{mx+n}$ & $(a^m-1) \ a^{mx+n}$ & 
$\frac{1}{a^m-1} \ a^{mx+n} + c$ \\
5 & $u(x)+v(x)$ & $\Delta u(x) + \Delta v(x)$ & $\Sigma u(x) \ \mathrm{d}x + \Sigma v(x) \ \mathrm{d}x + c$ \\
6 & $k \ u(x)$ 
& $k \ \Delta u(x)$ &  $k \  \Sigma u(x) \ \mathrm{d}x + c$ \\
7 & $u(x) \ v(x)$ & $v(x+1) \ \Delta u(x) + u(x) \ \Delta v(x)$ &  \\
8 &$u(x) \ \Delta v(x)$ & & $u(x) \ v(x) - \Sigma v(x+1) \ \Delta u(x) \ \mathrm{d}x + c$   \\ \hline \hline
\end{tabular}
\caption{A set of \finitecalculus rules 
used in our approach.}

\label{table:finite-calculus}
\end{table}

The rules in Table~\ref{table:finite-calculus} are applied to the
resulting expression until it does not contain any integral nor
summation. Note that $u-v$ and $\frac{u}{v}$ can rewritten as $u +
(-v)$ and $u \ \frac{1}{v}$ respectively. However, we use the
corresponding specialized rules for the subtraction and division.

For illustration purposes, we include here a simple and a more complex
example of the application of such rules.

\begin{example}
In order to find a closed form of $\summation{x=1}{a}{2^x}$, we proceed as follows:
\begin{enumerate}
\item Rewrite it as $\summation{1}{a+1}{2^x \ \mathrm{d}x}$, according
  to Theorem~\ref{th:fundamental-finite-calculus}.
\item 
\label{point:indef-integral}
Compute the corresponding discrete indefinite integral $\sum 2^x
  \ \mathrm{d}x$. This is done by using integration rule 2, so that
  $\sum 2^x \ \mathrm{d}x = 2^x$. Note that we omit the constant $c$
  that appears in the rules of Table~\ref{table:finite-calculus} since
  it is not relevant for the final result.

\item By using Definition~\ref{def:discrete-definite-integral} and the
  previous results, we have that:

$
\begin{array}{l}
\summation{x=1}{a}{2^x} = \summation{1}{a+1}{2^x \ \mathrm{d}x} \ \text{(Theorem~\ref{th:fundamental-finite-calculus})} \\
= 2^x  | ^{a+1}_{1} \ \text{(Definition~\ref{def:discrete-definite-integral} and integration rule 2)} \\
= 2^{a+1}-2^1 = 2^{a+1}-2
\end{array}
$
\end{enumerate}

\end{example}

\begin{example}
A closed form of $\summation{x=1}{a}{x \ 2^{a-x}}$ is obtained as
follows:
\begin{enumerate}
\item Rewrite it as $\summation{1}{a+1}{x \ 2^{a-x} \ \mathrm{d}x}$
(Theorem~\ref{th:fundamental-finite-calculus}). 

\item Compute the corresponding discrete indefinite integral $\sum x
  \ 2^{a-x} \ \mathrm{d}x$ by using integration by parts rule 8,
  making $u(x) = x$ and $\Delta v(x) = 2^{a-x} \ \mathrm{d}x$. Thus,
  $\Delta u(x) = 1 \ x^{\underline{0}} \ \mathrm{d}x = \mathrm{d}x$
  (derivative rule 1), and $v(x) = \sum \ 2^{a-x} \ \mathrm{d}x =
  \frac{1}{2^{-1}-1} \ 2^{a-x} = -2 \ 2^{a-x}$ (integration rule
  4). Now, we have:

$
\begin{array}{l}
 \sum x \ 2^{a-x} \ \mathrm{d}x = x \ (-2 \ 2^{a-x}) - \sum -2 \ 2^{a-(x+1)} \ \mathrm{d}x \\
 = - x \ 2^{a+1-x} - \sum -2^{a-x} \mathrm{d}x = - x \ 2^{a+1-x} + \sum  2^{a-x} \mathrm{d}x  \\
 = - x \ 2^{a+1-x} + (-2 \ 2^{a-x}) \ \text{(integration rule 4, as before)}  \\
 = - x \ 2^{a+1-x} - 2^{a+1-x} = -2^{a+1-x} \ (x+1) 
\end{array}
$

\item 
By Definition~\ref{def:discrete-definite-integral} and the previous
result, we have that:

$
\begin{array}{l}
\summation{1}{a+1}{x \ 2^{a-x} \ \mathrm{d}x} = -2^{a+1-x} \ (x+1) | ^{a+1}_{1}  \\
 = - 2^{a+1-(a+1)} \ ((a+1)+1) + 2^{a+1-1} \ (1+1) = - 2^0 \ (a+2) + 2^{a} \ 2  \\
 = 2^{a+1}-a-2
\end{array}
$

\item
Thus, $\summation{x=1}{a}{x \ 2^{a-x}} = \summation{1}{a+1}{x \ 2^{a-x} \ \mathrm{d}x} = 2^{a+1}-a-2$. 

\end{enumerate}

\end{example}

\paragraph{\bf Termination of the algorithm.}

The proof of termination of the recursive application of the rules of
Table~\ref{table:finite-calculus} is based on: a) in any of the
derivative rules (third column), the depth of the resulting
expression, with respect to the derivative operator $\Delta$, is
always $0$ (rules 1 to 4) or decreases by $1$ (rules 5 to 7); and b)
in any of the integration rules (fourth column), the depth of the
resulting expression, with respect to the integral operator
$\intoperator$, is always $0$ (rules 1 to 4) or decreases by $1$
(rules 5, 6 and 8). In addition, in integration rule 8, we apply the
derivative rules to the polynomial part, so that eventually, the depth
of the resulting expression will shrink down to a constant.

 Finally, as already said at the beginning of this section, our
approach for comparing summation functions consists in transforming
any summation cost function $C$ into an equivalent closed form cost
function $C^{t}$ that does not contain any summation subexpressions,
and then applying the comparison techniques explained in the previous
sections to the resulting closed form functions.  In general, such
transformation
is an \emph{undecidable} problem. However,
Table~\ref{table:finite-calculus} provides a decidable fragment of
summation expressions, which cover a large class of the functions that
are produced by the analysis that we use.  In addition, we detect
functions that are not covered by our approach and report them to the
user.

\subsection{Multiple Variable Cost Function Comparison}
\label{sec:multi-var-function-comparison}

Given two resource usage functions $\Psi_{1}(\bar n)$ and
$\Psi_{2}(\bar n)$, where $\bar n$ is the abbreviation of $k$
variables $n_1 \dots n_k$ representing input data sizes, we want to
know which values of $\bar n$ meet the constraint $\Psi_{1}(\bar n)
\leq \Psi_{2}(\bar n)$, so that we can view this problem as a
constraint satisfaction problem (CSP).

If the functions involved are \emph{linear functions} the problem can
be solved by using standard constraint programming techniques.  In our
implementation we use the Parma Polyhedra Library (PPL) to compute the
solutions in this case.
However, constraint programming cannot solve
the problem for \emph{polynomial functions} in general.

Unlike the case of single-variable cost functions, where we have
numerically bounded intervals as (input data size) preconditions, in
case of multiple-variable cost functions we need to be able to express
relations between variables as preconditions.
For example, given a function $x + y - 10 \leq 0$, all combinations of values
for $x$ and $y$ that satisfy the inequality can not be concisely represented
as intervals in the preconditions.
Therefore, instead of
using only intervals represented as pairs of numbers, we use
arithmetic expressions that represent more general
\emph{size constraints}. Table~\ref{tab:rules-function-interval}
summarizes the sufficient conditions used by our general verification
process, which can be applied to both multi- and single-variable cost
functions, showing the size constraints 
that need to be checked for different cases, depending on whether the
specification provides an \textbf{Upper Bound} cost function (denoted
as $S_{ub}$), a \textbf{Lower Bound} cost function ($S_{lb}$), or both
(columns 2 to 4 respectively). A symbol representing the result of the
verification process ($T$, $F$ or $C$), when such size constraints are
true, is shown at the right hand side of the implication symbol
($\rightarrow$), meaning that the specification has been verified
($T$), is false ($F$), or it cannot be proved whether the
specification 
is true or false.  Short names for the size constraints ($c_1$ to
$c_4$) are also used in order to achieve a compact representation. The
first column (\textbf{Analysis}) divides the table into three
different scenarios, each one corresponding to a row, depending on
whether the available analysis is able to infer upper-bound cost
functions, lower bounds, or both. As already explained, in this work
we use the parametric resource analysis integrated in \ciaopp
(see~\cite{resource-iclp07,plai-resources-iclp14} and its references),
which infers both upper and lower bounds.  
Note that the conditions $c_1 \wedge c_4$ and $c_2 \vee c_3$ given in
the last column and row of Table~\ref{tab:rules-function-interval},
correspond to sufficient conditions~\ref{suffcond1} and
~\ref{suffcond2} of Corollary~\ref{func-comparison-corollary}
respectively. Such conditions assume that both lower- and upper-bound
cost functions are available for both analysis and specification.
Either condition $c_1$ or $c_4$ in isolation is also equivalent to
sufficient condition~\ref{suffcond1} of
Corollary~\ref{func-comparison-corollary} if default, safe values for
the corresponding missing bounds are assumed. The same applies to
conditions $c_2$ and $c_3$, which are equivalent to sufficient
condition~\ref{suffcond2} of
Corollary~\ref{func-comparison-corollary}.

\begin{table}[t]
  \begin{center}
    \renewcommand\arraystretch{1.25} 
    \begin{tabular}{c l|l|l|l}\hline\hline
      \renewcommand\arraystretch{2} 
      & & \multicolumn{3}{c} {\bf Specification } \\ \cline{3-5}
      \renewcommand\arraystretch{1.5} 
      & & {\bf Upper}           & {\bf Lower }           & {\bf Upper \& Lower } \\
      & & {\bf Bound ($S_{ub})$} & {\bf Bound ($S_{lb}$) } & {\bf Bound} \\ \cline{1-5}
      \multicolumn{1}{c|}{\multirow{9}{*}{\bf\begin{turn}{90}Analysis\end{turn}}} & {\bf Upper} &
      $c_1 \rightarrow T$, where  & $c_3 \rightarrow F$, where  & 
$c_3 \rightarrow F$ 
\\ 
      \multicolumn{1}{c|}{} & {\bf Bound ($A_{ub}$)} & $c_1 \equiv S_{ub} \geq A_{ub}$  & $c_3 \equiv S_{lb}> A_{ub}$ & 
$\neg c_3 \rightarrow C$ 
\\ \cline{2-5}
      \multicolumn{1}{c|}{}& {\bf Lower} &
      $c_2 \rightarrow F$, where & $c_4 \rightarrow T$ & 
 $c_2 \rightarrow F$
\\ 
      \multicolumn{1}{c|}{}& {\bf Bound ($A_{lb}$)} & $c_2 \equiv S_{ub} < A_{lb}$  & $c_4 \equiv S_{lb} \leq A_{lb}$ & 
$\neg c_2 \rightarrow C$
      \\ \cline{2-5}
      \multicolumn{1}{c|}{}& {\bf Upper \& }& $c_1 \rightarrow T$ & $c_4 \rightarrow T$ & $c_1 \wedge c_4 \rightarrow T$
      \\
      \multicolumn{1}{c|}{}& {\bf Lower} & $c_2 \rightarrow F$ & $c_3 \rightarrow F$ & $c_2 \vee c_3 \rightarrow F$
      \\
      \multicolumn{1}{c|}{}& {\bf Bound} & $\neg c_1 \wedge \neg c_2 \rightarrow C$ & 
      $\neg c_3 \wedge \neg c_4 \rightarrow C$ & $\neg(c_1 \wedge c_4) \wedge \neg(c_2 \vee c_3) \rightarrow C$
      \\
      \hline\hline
    \end{tabular}
    \renewcommand\arraystretch{1} 

  \end{center}

 \caption{Sufficient conditions checked by our general verification
   process for different scenarios depending on the available bounds.}

  \label{tab:rules-function-interval}
\end{table}

\begin{examplebox}

Consider the \texttt{inc\_append/3} predicate in
Fig.~\ref{fig:inc_append_example}, which is an extension of the
classical \texttt{append/3}, also concatenating two 
lists of numbers, $A$ and $B$, 
but which also increments by 1 all the elements of the second list ($B$)
beforehand. The user assertion specifies that the upper bound on the
cost of the program, in terms of the number of resolution steps, is
$2*length(A)-10$ where $A$ is the first list to append.
\prettylstciao
\begin{figure}[t]
\prettylstciao
\begin{lstlisting}
:- check pred inc_append(A,B,C) + (cost(ub, steps, 2*length(A)-10)).

inc_append(A, B, C) :-
    inc_list(B, B1),
    append(A, B1, C).

inc_list([], []).
inc_list([E|R], [E1|T]) :-
    E1 is E + 1,
    inc_list(R, T).

append([],L,L).
append([A|R],S,[A|L]) :-
    append(R,S,L).
\end{lstlisting}
\caption{Append with increment example.}
\label{fig:inc_append_example}
\end{figure}
The analysis infers both an upper and a lower bound cost function,
which in this case both bounds coincide, namely
$length(B)+length(A)+3$. The output of the assertion checking is:
\begin{tabbing}
\texttt{:- \textcolor{false}{false}} \= \texttt{pred inc\_append(A,B,C)}\\
         \> \texttt{: intervals([[lt(-13,-length(A)+length(B))]])}\\
         \> \texttt{+ cost(ub,steps,2*length(A)-10).}\\

\texttt{:- \textcolor{checked}{checked}} \= \texttt{pred inc\_append(A,B,C)}\\
         \> \texttt{: intervals([[leq(13,length(A)-length(B))]])}\\
         \> \texttt{+ cost(ub,steps,2*length(A)-10).}
\end{tabbing}
meaning that when $-13<-length(A)+length(B)$ the assertion is false,
and when $13\leq length(A)-length(B)$ the assertion is correct.
\end{examplebox}

\section{Generic Implementation and Experimental Results}
\label{sec:implementation}

In order to assess the accuracy and efficiency (as well as the
scalability) of the resource usage verification techniques presented,
we have implemented
and integrated them in 
by extending the function comparison capabilities of the 
\ciao/\ciaopp~framework.

\renewcommand\arraystretch{1.1}
\begin{small}
\begin{table}
\footnotesize
\centering
\setlength\tabcolsep{2pt}
\begin{tabular}{|l|l|l|l|c|c|} \cline{1-6} 
  \multirow{2}{*}{
  \centering
  \parbox[l]{4cm}{
     \textbf{Program}$+$ \\
     \textbf{Analysis Info} $+$ \textbf{AvT}
  }}
  & \multirow{2}{*}{\textbf{ID}} & \multirow{2}{*}{\textbf{Assertion}} &
\multirow{2}{*}{\textbf{Verif. Result}} & \multicolumn{2}{c|}{\textbf{Time} (ms)}\\ \cline{5-6}
& & & & \textbf{Tot}&\textbf{Avg} \\ [1ex]
 \cline{1-6}

\underline{\emph{Fibonacci}} & \emph{A1} &:- pred fib(N,R)                       & F in $[0,10]$&  \multirow{14}{*}{106.4}& \multirow{14}{*}{35.4}\\
\textbf{lb,ub:} $1.45 * 1.62^x $ && +cost(ub,steps,    & T in $[11,\infty]$& &\\ 
\ \ \ $+ 0.55 * -0.62^x- 1$ && \quad exp(2,\nmetric(N))-1000). & & & \\
\cline{2-4}
x = \nmetric(N) &\emph{A2}& :- pred fib(N,R)            & F in $[0,10] \cup [15, \infty]$ & &\\ 
\textbf{AvT}$= \frac{1402.6 \ ms}{65 \ a}=21.5\frac{ms}{a} $&& + (cost(ub,steps,   & T in $[11,13]$ & &\\ 
                   &&      \quad exp(2,\nmetric(N))-1000),         & C in $[14,14]$ && \\
&& cost(lb,steps,    & & &\\ 
\textbf{AvT}$= \frac{\textbf{VTime}}{\textbf{\#Asser}}$                 && \quad exp(2,\nmetric(N))-10000)).                &&& \\
\cline{2-4}
                 &\emph{A3}& :- pred fib(N,R)       &F in $[1,10]$& &\\ 
                 &&   :(intervals(\nmetric(N),[i(1,12)])) & T in $[11,12]$ & &\\
                 && + (cost(ub,steps,   & & &\\ 
                 && \quad exp(2,\nmetric(N))-1000),        &&& \\
                 && cost(lb,steps,    & &&\\ 
                 && \quad exp(2,\nmetric(N))-10000)).     &&& \\  \cline{1-6} 
\underline{\emph{Naive Reverse}}    &\emph{B1}& :- pred nrev(A,B)  & F in $[0,3]$ & \multirow{7}{*}{59.1}&\multirow{7}{*}{29.5}\\
\textbf{lb,ub:} $0.5x^2+1.5x+1$ && + ( cost(lb,steps,length(A)),  & T in $[4,\infty]$ & &\\
x = length(A)    && cost(ub,steps,  &  &  &\\ 
\textbf{AvT}$= \frac{1171.5 \ ms}{54 \ a}=21.6\frac{ms}{a}$ && \quad exp(length(A),2))).              &&& \\\cline{2-4}
  &\emph{B2}&:- pred nrev(A,\_1)   & F in $[0,0] \cup [17,\infty]$ & & \\
     &&+ (cost(lb, steps, length(A)),   & T in $[1,16]$ & & \\ 
                        && cost(ub, steps, 10*length(A))).& & & \\  \cline{1-6} 

\underline{\emph{Quick Sort}}            &\emph{C1}&:- pred qsort(A,B)    & F in $[0,2]$ & \multirow{6}{*}{160.8}&\multirow{6}{*}{80.4} \\
\textbf{lb:} $x+5$           && + cost(ub, steps,  & C in $[3,\infty]$ & &\\  
\textbf{ub:} $(\sum_{j=1}^{x}j2^{x-j})+x2^{x-1}$         &&  \quad exp(length(A),2)).   &&& \\ \cline{2-4}
\ \ \ $+2*2^{x}-1$      &\emph{C2}& :- pred qsort(A,B) & C in $[0,\infty]$ &  &\\
x = length(A)       && + cost(ub, steps, &  & &\\ 
\textbf{AvT}$= \frac{1028.2 \ ms}{56 \ a}=18.3\frac{ms}{a}$    &&\quad exp(length(A),3)).  &  & &\\  [1ex] \cline{1-6} 

\underline{\emph{Client}}              &\emph{D1}&:- pred main(Op, I, B)   &C in $[1,7]$  & \multirow{11}{*}{31.8}& \multirow{11}{*}{10.6}\\
\textbf{ub:} $8x$    && + cost(ub, bits\_received,  &T in $[0,0] \cup [8,\infty]$  & &\\ 
x = length(I)        && exp(length(I),2)).  &  & &\\ \cline{2-4}
\textbf{AvT}$= \frac{1682.7 \ ms}{60 \ a}=28.04\frac{ms}{a}$          &\emph{D2}&:- pred main(Op, I, B)  & T in $[0,\infty]$ & &\\ 
         && + cost(ub, bits\_received, &  & &\\ 
             && \quad 10*length(I)).  &  & &\\ \cline{2-4}
                 &\emph{D3}&:- pred main(Op, I, B)  & T in $[1,10] \cup [100,\infty]$ &  &\\ 
                 &&: intervals(length(I),  &  &  &\\
                 && \quad [i(1,10),i(100,inf)])     & & & \\
                 && + cost(ub, bits\_received, &  & &\\ 
                 && \quad 10*length(I)).  &  & &\\  \cline{1-6} 
\underline{\emph{Reverse}}   &\emph{E1}&:- pred reverse(A, B)   & F in $[0,0]$ & \multirow{5}{*}{30.0}&\multirow{5}{*}{30.0}\\
\textbf{lb,ub:} $x+2$   && + (cost(ub, steps,  & T in $[1,\infty]$  & &\\ 
x = length(A)   && \quad 500 * length(A))). &  & &\\ 
\textbf{AvT}$= \frac{760.9 \ ms}{60 \ a}=12.6\frac{ms}{a}$    &&  &  & &\\  [1ex] \cline{1-6} 
\underline{\emph{Palindrome}} & \emph{F1}& :-  pred palindrome(X,Y)      & F in $[0,\infty]$ &\multirow{6}{*}{31.5}&\multirow{6}{*}{15.7} \\
 \textbf{lb,ub:} $x2^{x-1}+2*2^{x}-1$                   &   & + cost(ub,output\_elements, &        &   &  \\ 
x=length(X)   &   & \quad exp(length(X),2)). & & & \\ \cline{2-4}
\textbf{AvT}$= \frac{1187.1 \ ms}{52 \ a}=22.8\frac{ms}{a}$ &\emph{F2}& :-  pred palindrome(X,Y)&F in $[0,2] \cup [5,\infty]$ & & \\
                      &   & + cost(ub,output\_elements,& T in $[3,4]$ &   &  \\ 
          &&\quad exp(length(X),3)).       &&& \\  \cline{1-6} 
\underline{\emph{Powerset}} & \emph{G1} & :- pred powset(A,B)      & C in $[0,1] \cup [17, \infty]$ &\multirow{6}{*}{35.5}&\multirow{6}{*}{35.5} \\
\textbf{ub:} $0.5 * 2^{x+1}$                  &    & + cost(ub,output\_elements, &T in $[2,16]$ & & \\
x = length(A)&    & exp(length(A),4)). & & & \\ 
\textbf{AvT}$= \frac{880.9 \ ms}{49 \ a}=17.9\frac{ms}{a}$    &   &    &  & & 
  \\  [1ex] \cline{1-6} 
%
\underline{\emph{Hanoi}} & \emph{H1} & :- pred hanoi(A,B,C,D)      & F in $[0,1] \cup [5, \infty]$ &\multirow{6}{*}{121.2}&\multirow{6}{*}{121.2} \\
  \textbf{lb,ub:} $2^{x+1} - 2$                  &    & + costb(steps, exp(2,\nmetric(A)-3) + 2, &T in $[2,4]$ & & \\
x = \nmetric(A)&    & exp(2,\nmetric(A)-3) + 30). & & & \\ 
\textbf{AvT}$= \frac{1114.6 \ ms}{64 \ a}=17.41\frac{ms}{a}$    &   &    &  & & \\ [1ex] \cline{1-6} 

\end{tabular} 
\caption{Results of the interval-based static assertion checking integrated into \ciaopp.} 
\label{tab:result-eval}
\end{table}
\renewcommand\arraystretch{1}

\end{small}

\begin{table}[ht]
\centering
\begin{tabular}{|c|c|c|c|c|c|}  \cline{1-6} 
\multirow{2}{*}{\centering \textbf{ID}}  & \multirow{2}{*}{\centering \textbf{Method}} & \multicolumn{4}{c|}{\textbf{Intervals}}\\ [2ex]\cline{3-6}
  &  &         [1,12]&[1,100]&[1,1000]&[1,10000]  \\ [2ex] \cline{1-6} 
\emph{A3}  & \textbf{Root} & 58.1 & 64.6 &  71.7 &  66.5  \\ [2ex] \cline{2-6}
           & \textbf{Eval} & 257 & 256.2 & 261.1 & 262.9  \\ [2ex] \cline{1-6} 
\emph{D3}  & \textbf{Root} & 11.2 & 9 &  8.2 &  9.3  \\ [2ex] \cline{2-6}
           & \textbf{Eval} & 39.7 & 41.5 &  38.8 & 55.2  \\ [2ex] \cline{1-6} 
\end{tabular}

\caption{Comparison of assertion checking times for two methods.}
\label{tab:result-eval-bounded-domain}
\end{table}

Table~\ref{tab:result-eval} shows some experimental results obtained
with our prototype implementation on an Intel Core i5 2.5 GHz with
2 cores, 10GB 1333 MHz DDR3 of RAM, running MacOS Sierra 10.12.6.
The column labeled \textbf{Program} shows the name of
the program to be verified, the upper (\textbf{ub}) and lower
(\textbf{lb}) bound resource usage functions inferred by \ciaopp's
analyzers, the input arguments, and the size measure used.

The scalability of the different analyses required is beyond the scope
of this paper. We will just mention that in the case of the core
resource analysis, i.e., the one that processes the \hcir (to which
other languages are translated into), and infers cost functions, its
scalability follows generally from its compositional nature.
Our study focuses on the scalability of the assertion comparison
process.  To this end, we have added a total number of $390$
assertions to 
several programs that are then statically checked.  Column
\textbf{Program} shows an expression $\textbf{AvT} =
\frac{\textbf{VTime}}{\textbf{\#Asser}}$ for each program giving the
total time \textbf{VTime} in milliseconds spent by the verification of
the number assertions given by the denominator \textbf{\#Asser}, and
the resulting average time per assertion (\textbf{AvT}). A few of
those assertions are shown as examples in column \textbf{Assertion},
where \textbf{ID} is the assertion identifier.  Some assertions
specify both upper and lower bounds (e.g., \emph{A2} or \emph{A3}),
but others only specify upper bounds (e.g., \emph{A1} or \emph{C1}).
Also, some assertions include preconditions expressing intervals
within which the input data size of the program is supposed to lie
(\emph{A3} and \emph{D3}). The column \textbf{Verif. Result} shows the
result of the verification process for the assertions in column
\textbf{Assertion}, which in general express intervals of input data
sizes for which the assertion is true (\texttt{T}), false
(\texttt{F}), or it has not been possible to determine whether it is
true or false (\texttt{C}). Column \textbf{Tot} (under \textbf{Time})
shows the total time (in milliseconds) spent by the verification of
the assertions shown in column \textbf{Assertion} and \textbf{Avg}
shows the average time per assertion for these assertions.  In all the
experiments in Table~\ref{tab:result-eval}, the comparison of resource
usage functions was precise, in the sense that the input data size
intervals for which one function is greater, equal or smaller than
another were exact, i.e., coincided with the actual intervals.

Note that, as mentioned before, the system can deal with different
types of resource usage functions: polynomial functions (e.g., 
\emph{Naive Reverse}), exponential functions (e.g., \emph{Fibonacci}),
and summation functions (\emph{Quick Sort}). In general, polynomial
functions are faster to check than other functions, because they do
not need additional processing for approximation. However, the
additional time to compute approximations is very reasonable in
practice. Finally, note that the prototype was not able to determine
whether the assertion \emph{C2} in the \emph{Quick Sort} program is
true or false. 
This is because of two reasons: a) the analysis inferred an imprecise
upper-bound cost function, exponential, and b) our approach to finding
the data size intervals based on a transformation for removing
summations and an approximation by polynomials did not cover such
function. In some cases, either reason a) or b) in isolation can be
the cause for our approach to fail to prove a given assertion. Even in
the case when the cost bound function inferred by the analysis is precise,
if it is too complex, our approach may still fail to find roots and
data size intervals, and hence to prove the assertion.

Table~\ref{tab:result-eval-bounded-domain} shows assertion checking
times (in milliseconds) for different input data size intervals
(columns under \textbf{Intervals}) and for two methods: the one
described so far (referred to as \textbf{Root}), and a simple method
(\textbf{Eval}) that evaluates the resource usage functions for all
the (natural) values in a given input data size interval and compares
the results.  Column \textbf{ID} refers to the assertions in
Table~\ref{tab:result-eval}.  We can see that checking time grows
quite slowly compared to the length of the interval, which grows
exponentially.

\textbf{Root} is expected to be slower than \textbf{Eval}
in the comparison of non-polynomial functions (A3),
because \textbf{Root} must look for the functions intersections, and then
must check every value in the intervals to ensure the absence of other roots.
This behaviour is not exhibited in this experiment because the
intervals encountered by \textbf{Root} are narrow, 
and therefore the cost of checking every value in them is negligible.
On the other hand, in the last interval,
which grows wider as we increase the input data size interval,
\textbf{Eval} is penalized by the task of checking every value
in the interval, but \textbf{Root} is not penalized because
it uses syntactic comparison.

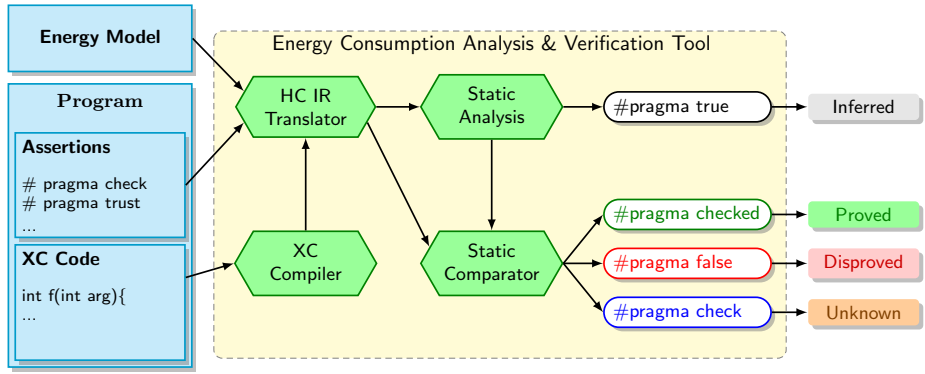
\begin{figure*}[t]
\resizebox{\textwidth}{!}{
  \centering 


\pgfdeclarelayer{background}
\pgfdeclarelayer{foreground}
\pgfsetlayers{background,main,foreground}

\tikzstyle{source}=[draw, draw=cyan!80!black!100, fill=cyan!20, text width=8em, font=\sffamily,
    thick,
    minimum height=2.5em,drop shadow]
\tikzstyle{tool}=[draw=green!50!black!100, fill=green!40, text width=5em, font=\sffamily, 
    thick,
    text centered, 
    chamfered rectangle, chamfered rectangle angle=30, chamfered rectangle xsep=2cm]
\tikzstyle{midresult}=[draw, fill=white!40, text width=7.6em, font=\sffamily,
    thick,
    rounded rectangle,
    minimum height=1em,drop shadow]
\tikzstyle{warnresult}=[color=orange!50!black!100, align=center, fill=orange!40, text width=5em, font=\sffamily, 
    thick,
    rounded corners=2pt,
    minimum height=1em,drop shadow]
\tikzstyle{errresult}=[color=red!80!black!100,  align=center, fill=red!20, text width=5em, font=\sffamily, 
    thick,
    rounded corners=2pt,
    minimum height=1em,drop shadow]
\tikzstyle{inforesult}=[align=center, fill=black!10, text width=5em, font=\sffamily, 
    thick,
    rounded corners=2pt,
    minimum height=1em,drop shadow]
\tikzstyle{okresult}=[color=green!50!black!100,  align=center, fill=green!40, text width=5em, font=\sffamily, 
    thick,
    rounded corners=2pt,
    minimum height=1em,drop shadow]
\tikzstyle{certresult}=[color=blue!50!black!100, fill=blue!40, text width=5em, font=\sffamily, 
    thick,
    rounded corners=2pt,
    minimum height=1em,drop shadow]
\tikzstyle{coderesult}=[color=blue!50!black!100, fill=blue!40, text width=5em, font=\sffamily, 
    thick,
    rounded corners=2pt,
    minimum height=1em,drop shadow]

\begin{tikzpicture}[>=latex,scale=1]
  \node (assertions) [source] {
    \textbf{Assertions}\\[2ex] \footnotesize
    \# pragma check\\
    \# pragma trust\\
    ...
  };
  \path (assertions)+(0,-6em) node (codeXC) [source] {
    \textbf{XC Code}\\[2ex]\footnotesize
     int f(int arg)\{\\
    ...\\[2ex]~
    
  };

  \path (assertions)+(0,7.5em) node (model) [source, align=center, text width=8.7em, minimum height=1.1cm ] {
    \textbf{Energy Model}\\
    };
  
\path (assertions)+(10.5em,4em) node (translator) [tool] {HC IR Translator};
\path (translator)+(0,-8em) node (compiler) [tool] {XC Compiler };
\path (translator)+(9.5em,0) node (statana) [tool] {Static Analysis};
\path[color=black] (statana)+(10em,0em) node (true) [midresult] {\#pragma true};
\path (statana)+(0em,-8em) node (comparator) [tool] {Static Comparator};
\path[color=blue] (comparator)+(10em,-2.5em) node (check) [midresult] {\#pragma check};
\path[color=red] (comparator)+(10em,0em) node (false) [midresult] {\#pragma false};
\path[color=green!50!black] (comparator)+(10em,2.5em) node (checked) [midresult] {\#pragma checked};

\path (true)+(9em,0em) node (inferred) [inforesult] {Inferred};
\path (false)+(9em,0em) node (cterror) [errresult] {Disproved};
\path (check)+(9em,0em) node (verifwarn) [warnresult] {Unknown};
\path (checked)+(9em,0) node (verified) [okresult] {Proved};

\path [draw, thick, ->] (model.east) -- node [] {} (translator.north west) ;
\path [draw, thick, ->] (assertions.east) -- node [] {} (translator.south west) ;
\path [draw, thick, ->] (codeXC) -- node [] {} (compiler.west) ;
\path [draw, thick, ->] (compiler) -- node [] {} (translator) ;
\path [draw, thick, ->] (translator.south east) -- node [] {} (comparator.north west) ;
\path [draw, thick, ->] (translator.east) -- node [] {} (statana.west) ;
\path [draw, thick, ->] (comparator.east) -- node [] {} (check.west) ;
\path [draw, thick, ->] (comparator.east) -- node [] {} (false.west) ;
\path [draw, thick, ->] (comparator.east) -- node [] {} (checked.west) ;
\path [draw, thick, ->] (statana) -- node [] {} (true) ;
\path [draw, thick, ->] (statana) -- 
(comparator) ;
\path [draw, thick, ->] (true) -- node [] {} (inferred) ;
\path [draw, thick, ->] (check) -- node [] {} (verifwarn) ;
\path [draw, thick, ->] (false) -- node [] {} (cterror) ;
\path [draw, thick, ->] (checked) -- node [] {} (verified) ;

\path [color=black] (statana.north)+(0,1.5em) node (preprocessor) {\sffamily \normalsize Energy Consumption Analysis \& Verification Tool};

\path [color=black] (assertions.north)+(0,1.5em) node (program) {\sffamily  \bf Program};

\begin{pgfonlayer}{background}
  \path (translator.north west)+(-0.5,1.0) node (g) {};
  \path (check.south east)+(0.5,-0.5) node (h) {};
  
  \path[fill=yellow!20,rounded corners, draw=black!50, densely dashed] (g) rectangle (h);
\end{pgfonlayer}

\begin{pgfonlayer}{background}
  \path (assertions.north west)+(-0.1,0.8) node (g) {};
  \path (codeXC.south east)+(0.1,-0.1) node (h) {};
  
  \path[source] (g) rectangle (h);
\end{pgfonlayer}

\end{tikzpicture}


}
\caption{Specialization of \ciaopp for energy consumption verification
  in XC programs.}
\label{fig:analysis-verif-hcir}.
\end{figure*}

\section{Application to Energy Verification of Imperative/Embedded Programs}
\label{sec:application-energy}

As an application of the techniques presented, in this section we
provide an overview of a prototype tool that we have developed for
performing \emph{static energy consumption verification} of \xc
programs running on the XMOS XS1-L architecture. The tool has been
implemented by specializing the \ciaopp general verification
framework 
to process \xc source, \llvmir~\cite{LattnerLLVM2004}, and 
\isa
code.
Fig.~\ref{fig:analysis-verif-hcir} shows an overview diagram of the
architecture of the tool. Hexagons represent different tool components and arrows
indicate the communication paths among them. The \tool takes as input
an \xc source program (left part of
Fig.~\ref{fig:analysis-verif-hcir}) that can optionally contain
assertions in a C-style syntax. As 
explained in Sect.~\ref{sec:intro}, such assertions are translated
into the \ial.

In our \tool the user can choose between performing the analysis at
the \isa or \llvmir \levels (or both).
We refer the reader to~\cite{isa-vs-llvm-fopara} for an experimental study that sheds light on the
trade-offs implied by performing the analysis at each of these two
\levels, which can help the user to choose the level that fits the
problem best.\footnote{As a brief summary of the conclusions
  of~\cite{isa-vs-llvm-fopara}, the \isa level allows somewhat tighter
  bounds when the analyzer can generate precise functions, but the
  \llvmir level allows the analyzer to produce precise functions more
  often, because more structural information is preserved at that
  level. Overall, the \llvmir level emerges as a good compromise.}

The associated \isa and/or \llvmir representations of the \xc program
are generated using the xcc compiler. Such representations include
useful metadata.
The \emph{\hcir translator} component (which will be described in
Sect.~\ref{sec:llvm-ciao-translation}) produces the internal
representation used by the tool, \hcir, which includes the program and
possibly specifications and/or trusted information (expressed in the
\ial).
The \hcir translator
performs several tasks:

\begin{enumerate}

\item Transforming the \isa and/or \llvmir into 
  \hcir. 

\item Transforming specifications (and trusted information) written as
  C-like assertions (as described in
  Sect.~\ref{sec:xc-assertion-language}) into the \ial.

\item Transforming the energy model at the \isa
  \level~\cite{Kerrison13}, expressed in JSON format, into the
  \ial. In this \specialization, such assertions express the energy
  consumed by individual \isa instruction representations, information
  which is required by the analyzer in order to propagate it during
  the static analysis of a program through code segments,
  conditionals, loops, recursions, etc., in order to infer analysis
  information (energy consumption functions) for higher-level entities
  such as procedures, functions, or loops in the
  program, as mentioned in~Example~\ref{ex:fact}.
  Fig.~\ref{fig:energymodel} shows the transformed energy
  model in the \ial. Each trust assertion provides information for one
  machine instruction.  The model of the figure is simple, providing
  just constant upper and a lower bounds (and which are the same in
  most cases), but the bounds given (model for the instruction) can be
  functions of input data to the instruction (such as operand sizes)
  or context variables (such as voltage or clock speed, previous
  instruction, pipeline state, cache state, etc.).

\item In the case that the analysis is performed at the \llvmir level,
  the \emph{\hcir translator} component produces a set of \ciao
  assertions expressing the energy consumption corresponding to
  \llvmir block representations in \hcir. Such information is produced
  from a mapping of \llvmir instructions with sequences of \isa
  instructions and the \isa-\level energy model. The mapping
  information is produced by the \emph{mapping tool} that was first
  outlined in~\cite{entra-d3.2} (Sect. 2 and Attachments D3.2.4 and
  D3.2.5) and is described in detail in~\cite{Georgiou14}.

\end{enumerate}

Then,
the \ciaopp parametric static resource usage
analyzer~\cite{resource-iclp07,NMHLFM08,plai-resources-iclp14} takes
the \hcir, together with the assertions which express the energy
consumed by \llvmir blocks and/or individual \isa instructions, and
possibly some additional (trusted) information, and processes them,
producing the analysis results, which are expressed also using \ciao
assertions.  Such results include energy usage functions (which depend
on input data sizes) for each block in the \hcir (i.e., for the whole
program and for all the procedures and functions in it.).
The procedural interpretation of the \hcir programs, coupled with the
resource-related information contained in the (\ciao) assertions,
together allow the resource analysis to infer static bounds on the
energy consumption of the \hcir programs that are applicable to the
original \llvmir and, hence, to their corresponding \xc programs.

The verification of energy specifications is performed by the
general component already described (see
Sect.~\ref{sec:intro} and
Fig.~\ref{fig:analysis-verif-hcir}), which compares the energy
specifications with the (safe) approximated information inferred by
the static resource analysis, and produces the possible verification
outcomes for different input-data size intervals.

\subsection{\Isa/\llvmir to \hcir Transformation}
\label{sec:llvm-ciao-translation}

In this section we briefly describe
the transformations into 
the \hcir representation described in
Sect.~\ref{sec:ciao-assertion-language} that we developed in order to
achieve the verification tool presented in
Sect.~\ref{sec:intro}
and depicted in
Fig.~\ref{fig:analysis-verif-hcir}. The transformation of \isa code
into \hcir was described in~\cite{isa-energy-lopstr13-final}.  
We provide herein an overview of
the \llvmir to \hcir transformation.

\llvmir programs are expressed using typed assembly-like
instructions. Each function is in SSA form, represented as a sequence
of basic blocks. Each basic block is a sequence of \llvmir
instructions that are guaranteed to be executed in the same order.
Each block ends in either a branching or a return instruction. In
order to represent each of the basic blocks of the \llvmir in the
\hcir, we follow a similar approach as in 
the \isa-\level transformation~\cite{isa-energy-lopstr13-final}.
However, the \llvmir includes an additional type transformation as
well as better memory modelling. It is explained in detail in~\cite{isa-vs-llvm-fopara}.
The main aspects of this process, are the following:

\begin{enumerate}

\item Infer input/output parameters to each block. 

\item Transform \llvmir types into \hcir types.

\item Represent each \llvmir block as an \hcir block and each
  instruction in the \llvmir block as a literal ($S_i$).
  
\item Resolve branching to multiple blocks by creating clauses with
  the same signature (i.e., the same name and arguments in the head),
where each clause denotes one of the blocks the branch may jump to.

\end{enumerate}

The translator component is also in charge of translating the \xc
assertions to \ialassertions and back.  Assuming the \ciao type of the
input and output of the function is known, the translation of
assertions from \ciao to \xc (and back) is relatively straightforward.
Assuming the schema for \texttt{pred} assertions described in Sect.~\ref{sec:ciao-assertion-language},
the \nt{Pred} field of the \ialassertion is obtained from the scope of
the \xc assertion to which an extra argument is added representing the
output of the function.  The \nt{Precond} fields are produced directly
from the type of the input arguments: for each input variable, its
regular type and its regular type size are added to the precondition,
while the added output argument is declared as a free variable.
Finally the \nt{Comp-Props} field is set to the usage of the resource
\texttt{energy} by using the \texttt{costb} property, which also includes
the lower and upper bounds from the \xc energy consumption
specification.

\begin{small}
\begin{figure}[tbp]
\begin{center}
\begin{tabular}{l r l}
\hline
\sym{assertion} & \gis & \terminal{\#pragma} \sym{status} \sym{scope} \terminal{:}  \sym{body} \\
\sym{status} & \gis & 
   \terminal{check}    \gor 
   \terminal{trust}    \gor 
   \terminal{true}     \gor  
   \terminal{checked}  \gor 
   \terminal{false} \\
\sym{scope} & \gis & \sym{identifier} \terminal{(} \terminal{)} \gor
\sym{identifier} \terminal{(} \sym{arguments} \terminal{)} \\
\sym{arguments} & \gis & \sym{identifier}  \gor \sym{arguments} \terminal{,}  \sym{identifier} \\
\sym{body} & \gis & \sym{precond}  \terminal{==>} \sym{cost\_bounds}  \gor  \sym{cost\_bounds} \\
\sym{precond} & \gis & \sym{upper\_cond} \gor \sym{lower\_cond} \gor \sym{lower\_cond}
\terminal{\&\&} \sym{upper\_cond} \\
\sym{lower\_cond} & \gis & \sym{ground\_expr} \terminal{<=} \sym{identifier} \\
\sym{upper\_cond} & \gis &  \sym{identifier} \terminal{<=} \sym{ground\_expr} \\
\sym{cost\_bounds} & \gis & \sym{lower\_bound} \gor \sym{upper\_bound}
\gor\sym{lower\_bound} \terminal{\&\&} \sym{upper\_bound} \\
\sym{lower\_bound} & \gis & \sym{expr} \terminal{<=} \terminal{energy\_nJ} \\
\sym{upper\_bound} & \gis & \terminal{energy\_nJ} \terminal{<=}  \sym{expr} \\
\sym{expr} & \gis & \sym{expr} \terminal{+} \sym{mult\_expr} \gor \sym{expr}
\terminal{-} \sym{mult\_expr} \\
\sym{mult\_expr} & \gis & \sym{mult\_expr} \terminal{*} \sym{unary\_expr} \gor
\sym{mult\_expr} \terminal{/} \sym{unary\_expr} \\
\sym{unary\_expr} & \gis & \sym{identifier} \\ 
& & \gor  \sym{integer} \\
& & \gor \terminal{sum} \terminal{(} \sym{identifier} \terminal{,} \sym{expr} \terminal{,} \sym{expr} \terminal{,} \sym{expr} \terminal{)} \\
& & \gor \terminal{prod} \terminal{(} \sym{identifier} \terminal{,} \sym{expr} \terminal{,} \sym{expr}\terminal{,} \sym{expr} \terminal{)} \\
  & & \gor \terminal{power} \terminal{(} \sym{expr} \terminal{,} \sym{expr} \terminal{)} \\
  & & \gor \terminal{log} \terminal{(} \sym{expr} \terminal{,} \sym{expr} \terminal{)} \\
  & & \gor \terminal{(} \sym{expr} \terminal{)}  \\
  & & \gor \terminal{+} \sym{unary\_expr} \\
  & & \gor \terminal{-} \sym{unary\_expr} \\
  & & \gor \terminal{min} \terminal{(} \sym{identifier} \terminal{)} \\
  & & \gor \terminal{max} \terminal{(} \sym{identifier} \terminal{)} \\
\hline
\end{tabular}
\vspace*{-6mm}
\end{center}
\caption{Syntax of the \xc Assertion Language.}
\label{fig:xc-assrt-syntax}
\end{figure}
\end{small}

\subsection{The \xc Assertion Language}
\label{sec:xc-assertion-language}

The assertions within \xc files are essentially equivalent to those of
the \ial, but written using a syntax that is closer to standard C
notation and friendlier for C developers.
These assertions are transparently translated into
\ialassertions~\cite{assert-lang-disciplbook,hermenegildo11:ciao-design-tplp}
when \xc files are loaded into the tool. The \ialassertions output by
the analysis are also translated back into \xc assertions and added
inline to a copy of the original \xc file.

More specifically, the syntax of the \xc assertions accepted by our tool
is given by the grammar in Fig.~\ref{fig:xc-assrt-syntax}, where the
non-terminal \sym{identifier} stands for a standard C identifier,
\sym{integer} stands for a standard C integer, and the
non-terminal \sym{ground\_expr} for a ground expression, i.e., an
expression of type \sym{expr} that does not contain any C identifiers
that appear in the assertion scope (the non-terminal \sym{scope}).

\xc assertions are  directives starting with the token
{\tt \#pragma} followed by the assertion {\em status}, the assertion
{\em scope}, and the assertion {\em body}.  The assertion {\em
  status} 
can take
several values, 
including \texttt{check}, \texttt{checked}, \texttt{false},
\texttt{trust} or \texttt{true}, with the same meaning as in the 
\ialassertions. 

The assertion scope identifies the function the assertion is referring
to, and provides the local names for the arguments of the function to
be used in the body of the assertion.  
For instance, the scope {\tt biquadCascade(state, xn, N)} refers to the
function {\tt biquadCascade} and binds the arguments within the body
of the assertion to the respective identifiers {\tt state}, {\tt xn},
{\tt N}.  
While the arguments do not need to be named in a consistent way
w.r.t.\ the function definition, it is highly recommended for the sake
of clarity.  
The \emph{body} of the assertion expresses bounds on the energy
consumed by the function and optionally contains preconditions (the
left hand side of the {\tt ==>} arrow) that constrain the argument
sizes.

Within the body, expressions of type \sym{expr} are built from
standard integer arithmetic functions (i.e., {\tt +}, {\tt -}, {\tt *},
{\tt /}) plus the following extra functions:
\begin{itemize}
\item {\tt power(base, exp)} is the exponentiation of {\tt base} by {\tt exp};
\item {\tt log(base, expr)} is the logarithm of  {\tt expr} in base {\tt base};
\item {\tt sum(id, lower, upper, expr)} is the summation of the sequence of the values 
   of {\tt expr} for {\tt id} ranging from {\tt lower} to {\tt upper}; 
\item {\tt prod(id, lower, upper, expr)} is the product of the sequence of the values 
   of {\tt expr} for {\tt id} ranging from {\tt lower} to {\tt upper};
 \item {\tt min(arr)} is the minimal value of the array {\tt arr};
\item {\tt max(arr)} is the maximal value of the array {\tt arr}.
\end{itemize}
Note that the argument of {\tt min} and {\tt max} must be an identifier
appearing in the assertion scope that corresponds to an array of integers
(of arbitrary dimension).


\subsection{Using the Tool for Energy Verification: Example}

In this section we illustrate the use of the tool described above for
the energy verification application, in a scenario where an embedded
software developer has to decide values for program parameters that
meet an energy budget. In particular we consider the development of an
equalizer (\xc) program using a biquad filter. In
Fig.~\ref{fig:biquadmenu} we can see what the graphical user
interface of our prototype looks like, with the code of this biquad
example ready to be verified. The purpose of an equalizer is to take a
signal, and to attenuate / amplify different frequency bands. For
example, in the case of an audio signal, this can be used to correct
for a speaker or microphone frequency response. The energy consumed by
such a program directly depends on several parameters, such as the
sample rate of the signal, and the number of banks, typically between
3 and 30 for an audio equalizer. A higher number of banks enables the
designer to create more precise frequency response curves.

\begin{figure*}[t]
  \includegraphics[width=0.8\textwidth,clip,trim=565 350 0 0]
    {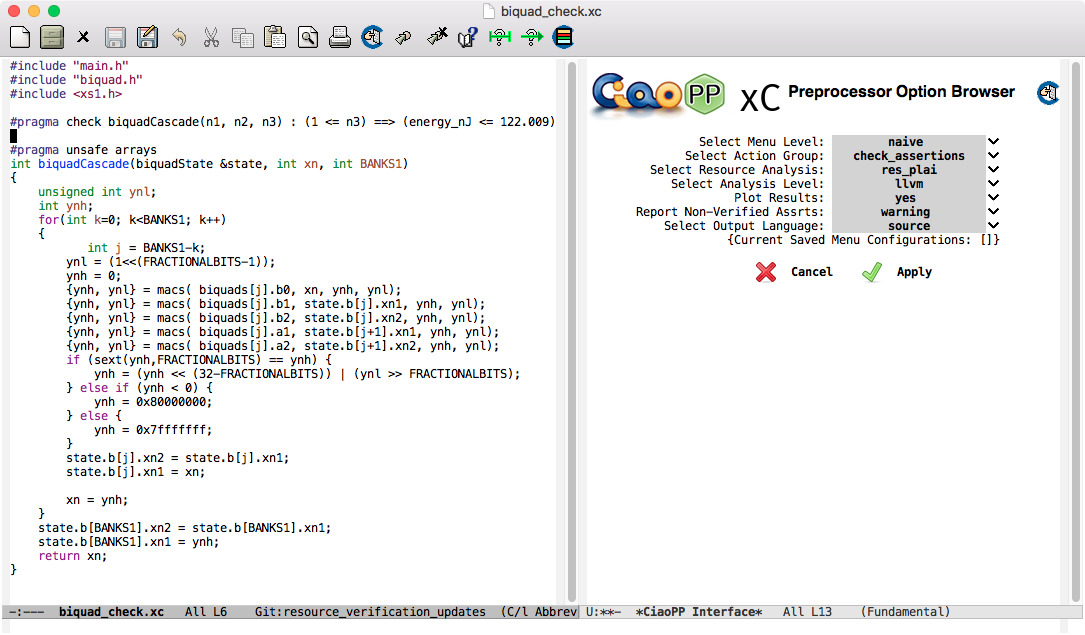}\\ [5mm]

  \includegraphics[width=0.8\textwidth,clip,trim=0 0 500 0]
    {./menu-biquad_nJ}

\caption{Graphical User Interface of the prototype with the \xc biquad program.}
\label{fig:biquadmenu}
\end{figure*}

Assume that the developer has to decide how many banks to use in order
to meet an energy budget while maximizing the precision of the frequency
response curves at the same time.  In this example, the developer
writes an \xc program where the number of banks is a variable, say
$\texttt{N}$. Assume also that the energy constraint to be met is that
an application of the biquad program should consume less or equal than
122 nJ (nanojoules).  
This constraint is expressed by the following check assertion:

\begin{tabbing}
\texttt{\#pragma} \= \texttt{\textcolor{check}{check} biquadCascade(state,xn,N) :} \\     
                  \> \texttt{(1 <= N) ==> (energy\_nJ <= 122)}
\end{tabbing}

\noindent
where the precondition $\texttt{1 <= N}$ in the assertion (left hand
side of $\texttt{==>}$) expresses that the constraint should hold when
the number of banks is greater than 1.

Then, the developer makes use of the tool by selecting the following
menu options, as shown in the right hand side of
Fig.~\ref{fig:biquadmenu}: \texttt{check\_assertions}, for
\texttt{Action Group}; \texttt{res\_plai}, for \texttt{Resource
  Analysis}; \texttt{llvm}, for \texttt{Analysis Level} (which will
tell the analysis to take the \llvmir option by compiling the source
code into \llvmir and transforming it into \hcir for analysis);
\texttt{source}, for \texttt{Output Language} (the language in which
the analysis / verification results are shown, in this case the
original XC source); and finally \texttt{yes} for \texttt{Plot
  results} (in order to obtain a graphical representation of the
results). After clicking on the \texttt{Apply} button below the menu
options, the analysis is performed, which infers a lower and an upper
bound function for the consumption of the program.  Specifically, those
bounds are represented by the following assertion, which is included
in the output of the tool:

\begin{tabbing}
\texttt{\#pragma} \= \texttt{\textcolor{true}{true} biquadCascade(A,B,C) : }\\
     \> \texttt{(16.502*C+5.445 <= energy\_nJ \&\& energy\_nJ <= 16.652*C+5.445)}
\end{tabbing}

Then, the verification of the specification, i.e., check assertion, is
performed by comparing the energy bound functions above with the upper
bound expressed in the specification, i.e.,
122 nJ, a constant value in this case, as illustrated in
Fig.~\ref{fig:plots}. Such figure has been automatically generated
by our tool
and includes the plots of both
the specification and the analysis results, which contributes to a
better understanding of the results.
The $x$ axis represents the input data size, in this case, the number
of banks given by $\mathtt{N}$, on which the cost function depends,
and the $y$ axis represents the energy consumption. The flat (blue)
region corresponds to the specification whereas the sloping
green region which lies between two red lines represents the area bounded by the cost functions
automatically inferred by the analyzer.

As a result of the comparison, the following two assertions are
produced and included in the output file of the tool:

\begin{tabbing}
\texttt{\#pragma} \= \texttt{\textcolor{checked}{checked} biquadCascade(state,xn,N) : }\\
     \> \texttt{   (1 <= N \&\& N <= 7) ==> (energy\_nJ <= 122) }\\
\texttt{\#pragma} \= \texttt{\textcolor{false}{false} biquadCascade(state,xn,N) : }\\
     \> \texttt{   (8 <= N) ==> (energy\_nJ <= 122) }
 \end{tabbing}

\begin{figure*}[t]
\centerline{\includegraphics[scale=0.5]{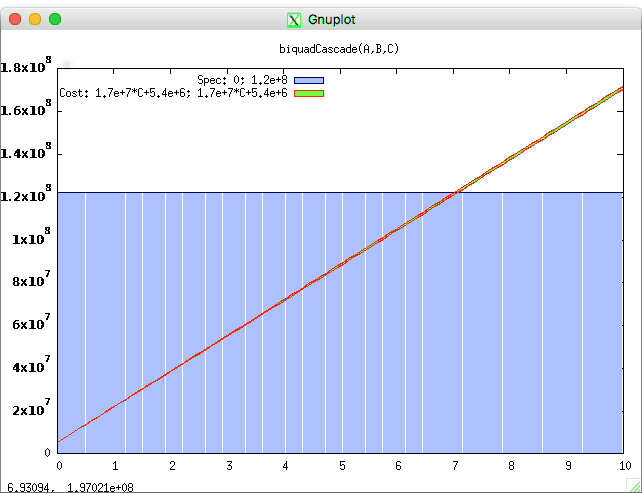}}
\caption{Visualization of analysis results and specifications in the tool.}
\label{fig:plots}
\end{figure*}

The first one expresses that the original assertion holds subject to a
precondition on the parameter $\texttt{N}$, i.e., in order to meet the
energy budget of 
122 nanojoules, the number of banks $\texttt{N}$
should be a natural number in the interval $[1,~7]$ (precondition
$\texttt{1 <= N \&\& N <= 7}$). The second one expresses that the
original specification is not met (status \texttt{false}) if the
number of banks is greater or equal to $8$. 

 Since the goal is to maximize the precision of the frequency response
 curves and to meet the energy budget at the same time, the number of
 banks should be set to 7. The developer could also be interested in
 meeting an energy budget but this time ensuring a lower bound on the
 precision of the frequency response curves. For example by ensuring that
 $\texttt{N} \geq 3$, the acceptable values for $\texttt{N}$ would be
 in the range $[3,~7]$.

 In the more general case where the energy function inferred by the
 tool depends on more than one parameter, the determination of the
 values for such parameters is reduced to a constraint solving
 problem. The advantage of this approach is that the parameters can be
 determined analytically at the program development phase, without the
 need of determining them experimentally by measuring the energy of
 expensive program runs with different input parameters, which in any
 case cannot provide hard guarantees.

\begin{figure*}[t]
\centerline{\includegraphics[scale=0.25]{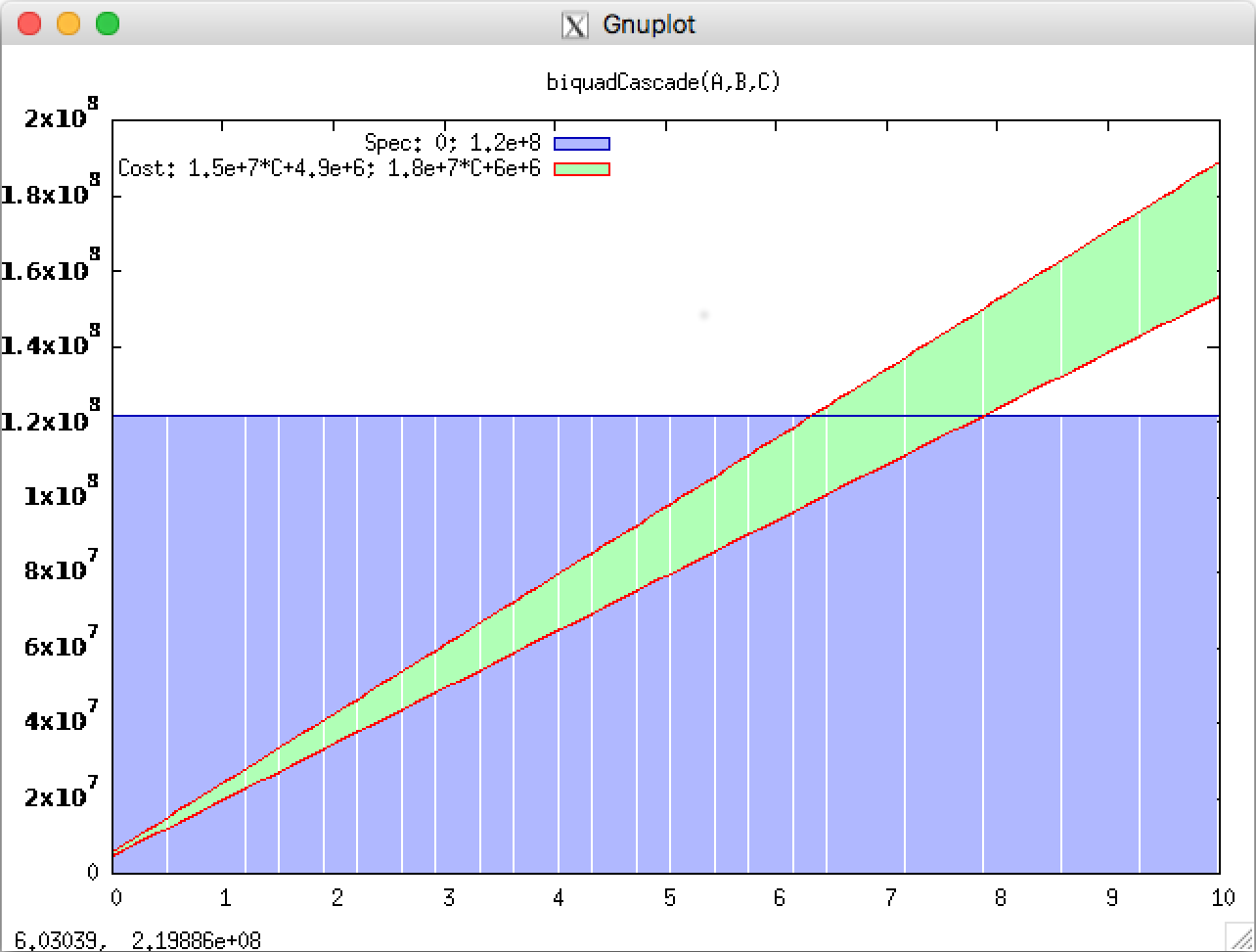}}
\caption{Visualization of analysis results and specifications, using a different
energy model.}
\label{fig:biquad-otherresult}
\end{figure*}

Our tool produces sound results, provided of course that the energy
model expresses correct information. Also, the accuracy of the bounds
obtained depends on the accuracy of the energy model.
Note that, if the objective is to choose parameters that guarantee
completely that the specifications are met, even not very tight bounds
will be better than testing/profiling, which, as mentioned before,
cannot provide hard guarantees. On the other hand, having tight bounds
is always desirable, in order to get more efficient values.

In order to illustrate this, assume that the user uses a slightly
different energy model for the verification, which considers a $10\%$
error in its energy measurements, and assume that this model,
expressed again as a set of $trust$ assertions in the \ciao~assertion
language, as in Fig.~\ref{fig:energymodel}, is contained in file
\texttt{energy\_llvm\_10}. In this case, the user needs to provide
this information to the tool as follows:

\begin{verbatim}
#pragma model <energy_llvm_10> 
\end{verbatim}

\noindent
Following the same procedure as before, after running the tool
the following results are obtained:

\begin{tabbing}
\texttt{\#pragma} \= \texttt{\textcolor{true}{true} biquadCascade(A,B,C) : }\\
     \> \texttt{(14.851*C+4.9 <= energy\_nJ \&\& energy\_nJ <= 18.317*C+5.989) }\\
\ \\
\texttt{\#pragma} \= \texttt{\textcolor{checked}{checked} biquadCascade(state,xn,N) : }\\
     \> \texttt{   (1 <= N \&\& N <= 6) ==> (energy\_nJ <= 122) }\\
\texttt{\#pragma} \= \texttt{\textcolor{check}{check} biquadCascade(state,xn,N) : }\\
     \> \texttt{   (7 <= N \&\& N <= 7) ==> (energy\_nJ <= 122) }\\
\texttt{\#pragma} \= \texttt{\textcolor{false}{false} biquadCascade(state,xn,N) : }\\
     \> \texttt{   (8 <= N) ==> (energy\_nJ <= 122) }
 \end{tabbing}

As we can see, the area delimited by the lower and upper bound
functions inferred is wider, and the verification results include an
additional \texttt{check} assertion for $N=7$. The assertion with
status \textit{check} indicates that for the value of the argument
$N=7$, the verification cannot conclude if the energy budget will be met or
not. This fact is represented in Fig.~\ref{fig:biquad-otherresult},
where the sloping/green analysis region intersects the flat/blue
specification region but is not completely included in it.

\section{Related Work}
\label{sec:relatedworks}

The closest related work we are aware of presents a method for
comparison of cost functions inferred by the COSTA system for Java
bytecode~\cite{AlbertAGHP09,Albert2015-cost-func-comp}. The method
proves whether a cost function is smaller than another one \emph{for
  all the values} of a given initial set of input data sizes.  The
result of this comparison is a boolean value. However, as mentioned
before, in our approach the result is in general a set of
intervals in which the initial set of input data sizes is
partitioned, so that the result of the comparison is different for
each subset.  Also, \cite{AlbertAGHP09} differs in that comparison is
syntactic, using a method similar to what was already being done in
the~\ciaopp~system: performing a function normalization and then using
some syntactic comparison rules.
In this work we go beyond these syntactic comparison
rules.
Note also that, although we have presented our work applied to Horn
clause programs and XC programs, the \ciaopp system can also deal with
Java bytecode~\cite{resources-bytecode09,decomp-oo-prolog-lopstr07}.

In a more general context, using abstract interpretation in
verification, debugging, and related tasks has now become well
established. To cite some early work, abstractions were used in the
context of algorithmic debugging in~\cite{Lichtenstein88}.  Abstract
interpretation has been applied by Bourdoncle~\cite{Bourd} to
debugging of imperative programs and by Comini et al.\ to the
algorithmic debugging of logic programs~\cite{CominiDD95} (making use
of partial specifications in~\cite{abs-diag-jlp}), and by
P.\ Cousot~\cite{Cousot-VMCAI03} to verification, among others.  The
\ciaopp~framework~\cite{aadebug97-informal,prog-glob-an,ciaopp-sas03-journal-scp}
was pioneering, offering an integrated approach combining
abstraction-based verification, debugging, and run-time checking with
an assertion language.  This approach has recently also been applied
in a number of contract-based
systems~\cite{clousot-2010-short,DBLP:conf/oopsla/Tobin-HochstadtH12-short,DBLP:conf/pldi/NguyenH15-short},

Horn clauses are used in many different applications nowadays as compilation
targets or intermediate representations in analysis and verification
tools~\cite{resources-bytecode09,decomp-oo-prolog-lopstr07,DBLP:conf/tacas/GrebenshchikovGLPR12,DBLP:conf/fm/HojjatKGIKR12-short,z3,hcvs14,kafle-cav2016}.

\secbeg
\section{Conclusions}
\secend
\label{sec:conclusions}

Taking as starting point our configurable framework
for 
static resource usage verification
where specifications can include both lower and upper bound, data
size-dependent resource usage functions, we have reviewed how this
framework supports different programming languages (both declarative
and imperative) as well as different compiler representations. This is
achieved by a translation of the corresponding \inputlanguage to an
internal representation based on Horn clauses (\hcir). The framework
is architecture independent, since we use low-level \resourceusage
models that are specific for each architecture, describing the
\resourceusage of basic elements and operations.

We have also generalized the assertions supported to include
preconditions expressing
intervals within which the input data size of a program is supposed to
lie (i.e., intervals for which each assertion is applicable). 
These extended assertions can be used both in specifications and in
the output of the analyzers. 
In addition, we have provided a formalization of how
the traditional framework is extended for the data size
interval-dependent verification of resource usage properties.

Our framework can deal with different types of
\resourceusage functions (e.g., polynomial, exponential, summation or
logarithmic functions), in the sense that the analysis can infer them,
and the specifications can involve them.

A key aspect of the framework is to be able to compare these
mathematical functions.  We have
proposed methods for function comparison that are safe/sound, in the
sense that the results of verification
either give a valid answer (true or false) or return ``unknown.''
In the case where the resource usage functions being compared depend
on one variable (which represents some input argument size) our method
reveals particular numerical intervals for such variable, if they
exist, which might result in different answers to the verification
problem: a given specification might be proved for some intervals but
disproved for others.  Our current method 
computes such intervals with precision for polynomial and exponential
resource usage functions, and in general for functions that can be
accurately approximated by polynomials near the point $x=0$.
Moreover, we have proposed an iterative
post-process to safely tune up the interval bounds by taking as
starting values the previously computed roots of the polynomials.

We have also reported on a prototype implementation of the proposed
general framework for resource usage verification and provided
experimental results, which are encouraging, suggesting that our
techniques are feasible and accurate in practice.  We have also
specialized such implementation for verifying energy consumption
specifications of imperative/embedded programs.
Finally, we have shown through an example, and using the prototype
implementation
for the \xc language and XS1-L architecture, how our verification
system can prove whether 
energy consumption specifications are met or not, or infer particular
conditions under which the specifications hold. 
We have illustrated through this example how embedded software
developers can use this tool, in particular for determining values
for program parameters that ensure meeting a given energy budget while
minimizing the loss in quality of service.

\begin{small}

\end{small}

\end{document}